\documentclass[twoside,11pt]{article}

\usepackage{jmlr2e}
\usepackage{enumerate}
\usepackage{url}
\usepackage{graphicx}
\usepackage{psfrag}
\usepackage{epsf}
\usepackage{url}
\usepackage{amsbsy}
\usepackage{amsfonts}
\usepackage{amsmath}
\usepackage{amssymb}
\usepackage{enumerate}
\usepackage{bm}
\usepackage{dsfont}
\usepackage{tikz}
\usepackage{fancyhdr}
\usepackage{nicefrac}
\usepackage{microtype}
\usepackage{booktabs}
\usepackage{multirow}
\usepackage{comment}
\usepackage{xspace}
\usepackage{hyperref}
\usepackage{amsbsy}
\usepackage{mathtools}
\usepackage{algorithm, algorithmic}

\newcommand{\iprod}[2]{\left \langle #1, #2 \right\rangle}
\newcommand{\rank}{\mathop{\sf rank}}
\newcommand{\norm}[1]{\|{#1} \|}
\newcommand{\argmin}{\mathop{\rm argmin}}
\newcommand{\argmax}{\mathop{\rm argmax}}
\newcommand{\diag}{\mathop{\text{diag}}}
\newcommand{\wh}{\widehat}
\newcommand{\wt}{\widetilde}
\newcommand{\ba}{\[
	\begin{aligned}}
\newcommand{\ea}{
	\end{aligned}
	\]}
\newcommand{\baa}{\begin{equation}
	\begin{aligned}}
\newcommand{\eaa}{
	\end{aligned}
	\end{equation}}

\newcommand{\fnorm}[1]{\|#1\|_{\rm F}}
\newcommand{\Tr}{\mathop{\sf Tr}}

\jmlrheading{21}{2020}{1-61}{2/19; Revised 2/20}{6/20}{19-123}{Youngseok Kim and Chao Gao}
\ShortHeadings{Bayesian Model Selection with Graph Structured Sparsity}{Kim and Gao}
\firstpageno{1}

\begin{document}
	

	
\title{Bayesian Model Selection with Graph Structured Sparsity}

\author{\name Youngseok Kim \email youngseok@uchicago.edu\\
	\name Chao Gao \email chaogao@galton.uchicago.edu \\
	\addr{Department of Statistics\\
		University of Chicago\\
		Chicago, IL 60637, USA}}
\editor{Francois Caron}
\date{\today}

\maketitle

\begin{abstract}%
	We propose a general algorithmic framework for Bayesian model selection. A spike-and-slab Laplacian prior is introduced to model the underlying structural assumption. Using the notion of effective resistance, we derive an EM-type algorithm with closed-form iterations to efficiently explore possible candidates for Bayesian model selection. The deterministic nature of the proposed algorithm makes it more scalable to large-scale and high-dimensional data sets compared with existing stochastic search algorithms. When applied to sparse linear regression, our framework recovers the EMVS algorithm \citep{rovckova2014emvs} as a special case. We also discuss extensions of our framework using tools from graph algebra to incorporate complex Bayesian models such as biclustering and submatrix localization. Extensive simulation studies and real data applications are conducted to demonstrate the superior performance of our methods over its frequentist competitors such as $\ell_0$ or $\ell_1$ penalization.
\end{abstract}

\vspace{0.07in}



\section{Introduction}
\label{sec:introduction}

Bayesian model selection has been an important area of research for several decades. While the general goal is to estimate the most plausible sub-model from the posterior distribution \citep{barry1993bayesian,diebolt1994estimation,richardson1997bayesian,bottolo2010evolutionary} for a wide class of learning tasks, most of the developments of Bayesian model selection have been focused on variable selection in the setting of sparse linear regression \citep{hans2007shotgun,li2010bayesian,ghosh2011rao,rovckova2014emvs,wang2018simple}. One of the main challenges of Bayesian model selection is its computational efficiency. Recently, 
\citet{rovckova2014emvs} discovered that Bayesian variable selection in sparse linear regression can be solved by an EM algorithm \cite{dempster1977maximum,neal1998view} with a closed-form update at each iteration. Compared with previous stochastic search type of algorithms such as Gibbs sampling \citep{george1993variable,george1997approaches}, this deterministic alternative greatly speeds up computation for large-scale and high-dimensional data sets.

The main thrust of this paper is to develop of a general framework of Bayesian models that includes sparse linear regression, change-point detection, clustering and many other models as special cases. We will derive a general EM-type algorithm that efficiently explores possible candidates for Bayesian model selection. When applied to sparse linear regression, our model and algorithmic frameworks naturally recover the proposal of \cite{rovckova2014emvs}. The general framework proposed in this paper can be viewed as an algorithmic counterpart of the theoretical framework for Bayesian high-dimensional structured linear models in \cite{gao2015general}. While the work \cite{gao2015general} is focused on optimal posterior contraction rate and oracle inequalities, the current paper pursues a general efficient and scalable computational strategy.

In order to study various Bayesian models from a unified perspective, we introduce a spike-and-slab Laplacian prior distribution on the model parameters. The new prior distribution is an extension of the classical spike-and-slab prior \citep{mitchell1988bayesian,george1993variable,george1997approaches} for Bayesian variable selection. Our new definition incorporates the graph Laplacian of the underlying graph representing the model structure, and thus gives the name of the prior. Under this general framework, the problem of Bayesian model selection can be recast as selecting a subgraph of some base graph determined by the statistical task. Here, the base graph and its subgraphs represent the structures of the full model and the corresponding sub-models, respectively. Various choices of base graphs lead to specific statistical estimation problems such as sparse linear regression, clustering and change-point detection. In addition, the connection to graph algebra further allows us to build prior distributions for even more complicated models. For example, using graph products such as Cartesian product or Kronecker product \citep{imrich2000product,leskovec2010kronecker}, we can construct prior distributions for biclustering models from the Laplacian of the graph products of row and column clustering structures. This leads to great flexibility in analyzing real data sets of complex structures.

Our Bayesian model selection follows the procedure of 
\citet{rovckova2014emvs} that evaluates the posterior probabilities of sub-models computed from the solution path of the EM algorithm.
However, the derivation of the EM algorithm under our general framework is indeed nontrivial task. When the underlying base graph of the model structure is a tree, the derivation of the EM algorithm is straightforward by following the arguments in \cite{rovckova2014emvs}. On the other hand, for a general base graph that is not a tree, the arguments in \cite{rovckova2014emvs} do not apply. To overcome this difficulty, we introduce a relaxation through the concept of effective resistance \citep{lovasz1993random,ghosh2008minimizing,spielman2007spectral} that adapts to the underlying graphical structure of the model. The lower bound given by this relaxation is then used to derive a variational EM algorithm that works under the general framework.

Model selection with graph structured sparsity has also been studied in the frequentist literature. For example, generalized Lasso \citep{tibshirani2011solution,arnold2014genlasso} and its multivariate version network Lasso \citep{hallac2015network} encode the graph structured sparsity with $\ell_1$ regularization. Algorithms based on $\ell_0$ regularization have also been investigated recently \citep{fan2018approximate, xu2019iterative}. Compared with these frequentist methods, our proposed Bayesian model selection procedure tends to achieve better model selection performance in terms of false discovery proportion and power in a wide range of model scenarios, which will be shown through an extensive numerical study under various settings.

The rest of the paper is organized as follows. In Section \ref{sec:model}, we introduce the general framework of Bayesian models and discuss the spike-and-slab Laplacian prior. The EM algorithm will be derived in Section \ref{sec:EM} for both the case of trees and general base graphs. In Section \ref{sec:clustering}, we discuss how to incorporate latent variables and propose a new Bayesian clustering models under our framework. Section \ref{sec:algebra} introduces the techniques of graph products and several important extensions of our framework. We will also discuss a non-Gaussian spike-and-slab Laplacian prior in Section \ref{sec:iso} with a natural application to reduced isotonic regression \citep{schell1997reduced}. Finally, extensive simulated and real data analysis will be presented in Section \ref{sec:num}.



\section{A General Framework of Bayesian Models}
\label{sec:model}

In this section, we describe a general framework for building Bayesian structured models on graphs. To be specific, the prior structural assumption on the parameter $\theta\in\mathbb{R}^p$ will be encoded by a graph. Throughout the paper, $G=(V,E)$ is an undirected graph with $V=[p]$ and some $E\subset\{(i,j):1\leq i<j\leq p\}$. It is referred to as the \textit{base graph} of the model, and our goal is to learn a sparse subgraph of $G$ from the data. We use $p=|V|$ and $m=|E|$ for the node size and edge size of the base graph.

\subsection{Model Description}\label{sec:model-des}

We start with the Gaussian linear model $y\,|\,\beta,\sigma^2\sim N(X\beta,\sigma^2 I_n)$ that models an $n$-dimensional observation. The design matrix $X\in\mathbb{R}^{n\times p}$ is determined by the context of the problem. Given some nonzero vector $w\in\mathbb{R}^p$, the Euclidean space $\mathbb{R}^p$ can be decomposed as a direct sum of the one-dimensional subspace spanned by $w$ and its orthogonal complement. In other words, we can write
$$\beta=\frac{1}{\|w\|^2}ww^T \beta + \left(I_p-\frac{1}{\|w\|^2}ww^T\right)\beta.$$
The structural assumption will be imposed by a prior on the second term above. To simplify the notation, we introduce the space $\Theta_w=\left\{\theta\in\mathbb{R}^p: w^T\theta=0\right\}$. Then, any $\beta\in\mathbb{R}^p$ can be decomposed as $\beta=\alpha w+\theta$ for some $\alpha\in\mathbb{R}$ and $\theta\in\Theta_w$. The likelihood is thus given by
\begin{equation}
y\,|\,\alpha,\theta,\sigma^2 \sim N( X(\alpha w + \theta), \sigma^2 I_n). \label{eq:likelihood}
\end{equation}

The prior distribution on the vector $\alpha w + \theta$ will be specified by independent priors on $\alpha$ and $\theta$. They are given by
\begin{eqnarray}
\label{eq:alpha-prior} \alpha \,|\, \sigma^2 &\sim& N(0,\sigma^2/\nu), \\
\label{eq:Laplacian-prior} \theta\,|\,\gamma,\sigma^2 &\sim& p(\theta\,|\,\gamma,\sigma^2) \propto \prod_{(i,j)\in E}\exp\left( -\frac{(\theta_i - \theta_j)^2}{2\sigma^2[v_0\gamma_{ij} + v_1(1-\gamma_{ij})]} \right) \mathbb{I}\{\theta \in \Theta_w \}.
\end{eqnarray}
Under the prior distribution, $\alpha$ is centered at $0$ and has precision $\nu/\sigma^2$. The parameter $\theta$ is modeled by a prior distribution on $\Theta_w$ that encodes a pairwise relation between $\theta_i$ and $\theta_j$. Here, $v_0$ is a very small scalar and $v_1$ is a very large scalar. For a pair $(i,j)\in E$ in the base graph, the prior enforces the closedness between $\theta_i$ and $\theta_j$ when $\gamma_{ij}=1$. Our goal is then to learn the most probable subgraph structure encoded by $\{\gamma_{ij}\}$, which will be estimated from the posterior distribution.

We finish the Bayesian modeling by putting priors on $\gamma$ and $\sigma^2$. They are given by
\begin{eqnarray}
\label{eq:prior-gamma} \gamma\,|\,\eta &\sim& p(\gamma\,|\,\eta)\propto\prod_{(i,j)\in E}\eta^{\gamma_{ij}}(1-\eta)^{1-\gamma_{ij}}\mathbb{I}\{\gamma\in\Gamma\}, \\
\label{eq:prior-eta}\eta &\sim& \text{Beta}(A,B), \\
\label{eq:prior-sigma}\sigma^2 &\sim& \text{InvGamma}(a/2,b/2).
\end{eqnarray}
Besides the standard conjugate priors on $\eta$ and $\sigma^2$, the independent Bernoulli prior on $\gamma$ is restricted on a set $\Gamma\subset\{0,1\}^m$. This restriction is sometimes useful for particular models, but for now we assume that $\Gamma=\{0,1\}^m$ until it is needed in Section \ref{sec:clustering}.

The Bayesian model is now fully specified. The joint distribution is
\begin{equation}
p(y,\alpha,\theta,\gamma,\eta,\sigma^2)=p(y\,|\,\alpha,\theta,\sigma^2)p(\alpha\,|\,\sigma^2)p(\theta\,|\,\gamma,\sigma^2)p(\gamma\,|\,\eta)p(\eta)p(\sigma^2).\label{eq:joint}
\end{equation}
Among these distributions, the most important one is $p(\theta|\gamma,\sigma^2)$. To understand its properties, we introduce the \textit{incidence matrix} $D\in\mathbb{R}^{m\times p}$ for the base graph $G=(V,E)$. The matrix $D$ has entries $D_{ei}=1$ and $D_{ej}=-1$ if $e=(i,j)$, and $D_{ek}=0$ if $k\neq i,j$. We note that the definition of $D$ depends on the order of edges $\{(i,j)\}$ even if $G$ is an undirected graph. However, this does not affect any application that we will need in the paper. We then define the Laplacian matrix
$$L_{\gamma}=D^T\diag\left(v_0^{-1}\gamma+v_1^{-1}(1-\gamma)\right)D.$$
It is easy to see that $L_{\gamma}$ is the graph Laplacian of the weighted graph with adjacency matrix $\{v_0^{-1}\gamma_{ij}+v_1^{-1}(1-\gamma_{ij})\}$. Thus, we can write (\ref{eq:Laplacian-prior}) as
\begin{equation}
p(\theta\,|\,\gamma,\sigma^2) \propto \exp\left(-\frac{1}{2\sigma^2}\theta^TL_{\gamma}\theta\right)\mathbb{I}\{\theta\in\Theta_w\}. \label{eq:prior-Laplacian-form}
\end{equation}
Given its form, we name (\ref{eq:prior-Laplacian-form}) the \textit{spike-and-slab Laplacian prior}.
\begin{proposition}\label{prop:Laplacian}
	Suppose $G=(V,E)$ is a connected base graph. For any $\gamma\in\{0,1\}^m$ and $v_0,v_1\in(0,\infty)$, the graph Laplacian $L_{\gamma}$ is positive semi-definite and has rank $p-1$. The only eigenvector corresponding to its zero eigenvalue is proportional to $\mathds{1}_p$, the vector with all entries $1$. As a consequence, as long as $\mathds{1}_p^Tw\neq 0$, the spike-and-slab Laplacian prior is a non-degenerate distribution on $\Theta_w$. Its density function with respect to the Lebesgue measure restricted to $\Theta_w$ is
	$$p(\theta\,|\,\gamma,\sigma^2)=\frac{1}{(2\pi\sigma^2)^{(p-1)/2}}\sqrt{\text{det}_w(L_{\gamma})}\exp\left(-\frac{1}{2\sigma^2}\theta^TL_{\gamma}\theta\right)\mathbb{I}\{\theta\in\Theta_w\},$$
	where $\text{det}_w(L_{\gamma})$ is the product of all nonzero eigenvalues of the positive semi-definite matrix $\left(I_p-\frac{1}{\|w\|^2}ww^T\right)L_{\gamma}\left(I_p-\frac{1}{\|w\|^2}ww^T\right)$.
\end{proposition}

The proposition reveals two important conditions that lead to the well-definedness of the spike-and-slab Laplacian prior: the connectedness of the base graph $G=(V,E)$ and $\mathds{1}_p^Tw\neq 0$. Without either condition, the distribution would be degenerate on $\Theta_w$. Extensions to a base graph that is not necessarily connected is possible. We leave this task to Section \ref{sec:clustering} and Section \ref{sec:algebra}, where tools from graph algebra are introduced.

\subsection{Examples}\label{sec:ex}

The Bayesian model (\ref{eq:joint}) provides a very general framework. By choosing a different base graph $G=(V,E)$, a design matrix $X$, a grounding vector $w\in\mathbb{R}^p$ and a precision parameter $\nu$, we then obtain a different model. Several important examples are given below.

\begin{example}[Sparse linear regression]\label{ex:slr}
	The sparse linear regression model $y\,|\,\theta,\sigma^2\sim N(X\theta,\sigma^2 I_n)$ is a special case of (\ref{eq:likelihood}). To put it into the general framework, we can expand the design matrix $X\in\mathbb{R}^{n\times p}$ and the regression vector $\theta\in\mathbb{R}^p$ by $[0_n,X]\in\mathbb{R}^{n\times (p+1)}$ and $[\theta_0;\theta]\in\mathbb{R}^{p+1}$. With the grounding vector $w=[1;0_p]$, the sparse linear regression model can be recovered from  (\ref{eq:likelihood}). For the prior distribution, the base graph $G$ consists of nodes $V=\{0,1,...,p\}$ and edges $\{(0,i): i\in[p]\}$. We set $\nu=\infty$, so that $\theta_0=0$ with prior probability one. Then, (\ref{eq:Laplacian-prior}) is reduced to
	$$\theta\,|\,\gamma,\sigma^2\sim p(\theta\,|\,\gamma,\sigma^2) \propto\prod_{i=1}^p\exp\left(-\frac{\theta_i^2}{2\sigma^2[v_0\gamma_{0i}+v_1(1-\gamma_{0i})]}\right).$$
	That is, $\theta_i|\gamma,\sigma^2\sim N(0,\sigma^2[v_0\gamma_{0i}+v_1(1-\gamma_{0i})])$ independently for all $i\in[n]$. This is recognized as the spike-and-slab Gaussian prior for Bayesian sparse linear regression considered by \cite{george1993variable,george1997approaches,rovckova2014emvs}.
\end{example}

\begin{example}[Change-point detection]\label{ex:cpm}
	Set $n=p$, $X=I_n$, and $w=\mathds{1}_n$. We then have $y_i\,|\,\theta_i,\sigma^2\sim N(\alpha+\theta_i,\sigma^2)$ independently for all $i\in[n]$ from (\ref{eq:likelihood}). For the prior distribution on $\alpha$ and $\theta$, we consider $\nu=0$ and a one-dimensional chain graph $G=(V,E)$ with $E=\{(i,i+1): i\in[n-1]\}$. This leads to a flat prior on $\alpha$, and the prior on $\theta$ is given by
	$$\theta\,|\,\gamma,\sigma^2\sim p(\theta\,|\,\gamma,\sigma^2)\propto \prod_{i=1}^{n-1}\exp\left(-\frac{(\theta_i-\theta_{i+1})^2}{2\sigma^2[v_0\gamma_{i,i+1}+v_1(1-\gamma_{i,i+1})]}\right)\mathbb{I}\{\mathds{1}_p^T\theta=0\}.$$
	A more general change-point model on a tree can also be obtained by constructing a tree base graph $G$.
\end{example}

\begin{example}[Two-dimensional image denoising] \label{ex:2dimgrid}
	Consider a rectangular set of observations $y\in\mathbb{R}^{n_1\times n_2}$. With the same construction in Example \ref{ex:cpm} applied to ${\sf sec}(y)$, we obtain $y_{ij}\,|\,\theta_{ij},\sigma^2\sim N(\alpha+\theta_{ij},\sigma^2)$ independently for all $(i,j)\in[n_1]\times [n_2]$ from (\ref{eq:likelihood}). To model images, we consider a prior distribution that imposes closedness to nearby pixels. Consider $\nu=0$ and a base graph $G=(V,E)$ shown in the picture below.
	\begin{center}
		\begin{tikzpicture}
		[scale=.5,auto=left,every node/.style={}]
		\node (vn1) at (1,1) {$\theta_{n_11}$};
		\node (v31) at (1,3)  {$\vdots$};
		\node (v21) at (1,5)  {$\theta_{21}$};
		\node (v11) at (1,7)  {$\theta_{11}$};
		\node (vn2) at (3,1) {$\theta_{n_12}$};
		\node (v32) at (3,3) {$\vdots$}; 
		\node (v22) at (3,5) {$\theta_{22}$}; 
		\node (v12) at (3,7) {$\theta_{12}$};
		\node (vn3) at (5,1) {$\cdots$};
		\node (v33) at (5,3) {$\ddots$}; 
		\node (v23) at (5,5) {$\cdots$}; 
		\node (v13) at (5,7) {$\cdots$}; 
		\node (vnm) at (7,1) {$\theta_{n_1n_2}$};
		\node (v3m) at (7,3) {$\vdots$}; 
		\node (v2m) at (7,5) {$\theta_{2n_2}$};
		\node (v1m) at (7,7) {$\theta_{1n_2}$};
		\foreach \from/\to in {v11/v21,v21/v31,v31/vn1,v11/v12,v12/v13,v13/v1m,v21/v22,v22/v23,v23/v2m,v1m/v2m,v2m/v3m,v3m/vnm,vn1/vn2,vn2/vn3,vn3/vnm,v12/v22,v22/v32,v32/vn2,v31/v32,v13/v23,v23/v33,v32/v33,v33/vn3,v33/v3m}
		\draw (\from) -- (\to);
		\end{tikzpicture}
	\end{center}
	We then obtain a flat prior on $\alpha$, and
	$$\theta\,|\,\gamma,\sigma^2\sim p(\theta\,|\,\gamma,\sigma^2)\propto \prod_{(ik,jl)\in E}\exp\left(-\frac{(\theta_{ik}-\theta_{jl})^2}{2\sigma^2[v_0\gamma_{ik,jl}+v_1(1-\gamma_{ik,jl})]}\right)\mathbb{I}\{\mathds{1}_{n_1}^T\theta\mathds{1}_{n_2}=0\}.$$
	Note that $G$ is not a tree in this case.
\end{example}

\section{EM Algorithm}\label{sec:EM}

In this section, we will develop efficient EM algorithms for the general model. It turns out that the bottleneck is the computation of $\text{det}_w(L_{\gamma})$ given some $\gamma\in\{0,1\}^m$.
\begin{lemma}\label{lem:matrix-tree}
	Let $\text{spt}(G)$ be the set of all spanning trees of $G$. Then
	$$\text{det}_w(L_{\gamma})=\frac{(\mathds{1}_p^Tw)^2}{\|w\|^2}\sum_{T\in\text{spt}(G)}\prod_{(i,j)\in T}\left[v_0^{-1}\gamma_{ij}+v_1^{-1}(1-\gamma_{ij})\right].$$
	In particular, if $G$ is a tree, then $\text{det}_w(L_{\gamma})=\frac{(\mathds{1}_p^Tw)^2}{\|w\|^2}\prod_{(i,j)\in E}\left[v_0^{-1}\gamma_{ij}+v_1^{-1}(1-\gamma_{ij})\right]$.
\end{lemma}
The lemma suggests that the hardness of computing $\text{det}_w(L_{\gamma})$ depends on the number of spanning trees of the base graph $G$. When the base graph is a tree, $\text{det}_w(L_{\gamma})$ is factorized over the edges of the tree, which greatly simplifies the derivation of the algorithm. We will derive a closed-form EM algorithm in Section \ref{sec:alg-tree} when $G$ is a tree, and the algorithm for a general $G$ will be given in Section \ref{sec:alg-non}.

\subsection{The Case of Trees}\label{sec:alg-tree}

We treat $\gamma$ as latent. Our goal is to maximize the marginal distribution after integrating out the latent variables. That is,
\begin{equation}
\max_{\alpha,\theta\in\Theta_w,\eta,\sigma^2}\log \sum_{\gamma}p(y,\alpha,\theta,\gamma,\eta,\sigma^2),\label{eq:EM-hard-obj}
\end{equation}
where $p(y,\alpha,\theta,\gamma,\eta,\sigma^2)$ is given by (\ref{eq:joint}). Since the summation over $\gamma$ is intractable, we consider an equivalent form of (\ref{eq:EM-hard-obj}), which is
\begin{equation}
\max_q\max_{\alpha,\theta\in\Theta_w,\eta,\sigma^2}\sum_{\gamma}q(\gamma)\log\frac{p(y,\alpha,\theta,\gamma,\eta,\sigma^2)}{q(\gamma)}. \label{eq:EM-equivalent}
\end{equation}
Then, the EM algorithm is equivalent to iteratively updating $q,\alpha,\theta\in\Theta_w,\eta,\sigma^2$ \citep{neal1998view}.

Now we illustrate the EM algorithm that solves (\ref{eq:EM-equivalent}). The E-step is to update $q(\gamma)$ given the previous values of $\theta,\eta,\sigma$. In view of (\ref{eq:joint}), we have
\begin{equation}
q^{\rm new}(\gamma)\propto p(y,\alpha,\theta,\gamma,\eta,\sigma^2)\propto p(\theta\,|\,\gamma,\sigma^2)p(\gamma\,|\,\eta). \label{eq:E-step-joint}
\end{equation}
According to (\ref{lem:matrix-tree}), $p(\theta\,|\,\gamma,\sigma^2)$ can be factorized when the base graph $G=(V,E)$ is a tree. Therefore, with a simpler notation $q_{ij}=q(\gamma_{ij}=1)$, we can write the update for $q$ as $q^{\rm new}(\gamma)=\prod_{(i,j)\in E}(q_{ij}^{\rm new})^{\gamma_{ij}}(1-q_{ij}^{\rm new})^{1-\gamma_{ij}}$, where
\begin{equation}
q_{ij}^{\rm new}=\frac{\eta\phi(\theta_i-\theta_j;0,\sigma^2 v_0)}{\eta\phi(\theta_i-\theta_j;0,\sigma^2 v_0)+(1-\eta)\phi(\theta_i-\theta_j;0,\sigma^2v_1)}. \label{eq:E-tree}
\end{equation}
Here, $\phi(\cdot;\mu,\sigma^2)$ stands for the density function of $N(\mu,\sigma^2)$.

To derive the M-step, we introduce the following function
\begin{equation}
F(\alpha,\theta;q)=\|y-X(\alpha w+\theta)\|^2+\nu\alpha^2+\theta^TL_{q}\theta, \label{eq:F()}
\end{equation}
where $L_{q}$ is obtained by replacing $\gamma$ with $q$ in the definition of the graph Laplacian $L_{\gamma}$.
The M-step consists of the following three updates,
\begin{eqnarray}
\label{eq:M1} (\alpha^{\rm new},\theta^{\rm new}) &=& \argmin_{\alpha,\theta\in\Theta_w}F(\alpha,\theta;q^{\rm new}), \\
\label{eq:M2}(\sigma^2)^{\rm new} &=& \argmin_{\sigma^2}\left[\frac{F(\alpha^{\rm new},\theta^{\rm new};q^{\rm new})+b}{2\sigma^2} + \frac{p+n+a+2}{2}\log(\sigma^2)\right], \\
\label{eq:M3}\eta^{\rm new} &=& \argmax_{\eta}\left[\left(A-1+q^{\rm new}_{\rm sum}\right)\log\eta + \left(B-1+p-1-q^{\rm new}_{\rm sum}\right)\log(1-\eta)\right],
\end{eqnarray}
where the notation $q^{\rm new}_{\rm sum}$ stands for $\sum_{(i,j)\in E}q_{ij}^{\rm new}$. While (\ref{eq:M1}) is a simple quadratic programming, (\ref{eq:M2}) and (\ref{eq:M3}) have closed forms, which are given by
\begin{equation}
(\sigma^2)^{\rm new}=\frac{F(\alpha^{\rm new},\theta^{\rm new};q^{\rm new})+b}{p+n+a+2}\quad\text{and}\quad \eta^{\rm new}=\frac{A-1+q^{\rm new}_{\rm sum}}{A+B+p-3}.\label{eq:M4}
\end{equation}

We remark that the EMVS algorithm \citep{rovckova2014emvs} is a special case for the sparse linear regression problem discussed in Example \ref{ex:slr}. When $G$ is a tree, the spike-and-slab graph Laplacian prior \eqref{eq:prior-Laplacian-form} is proportional to the product of individual spike-and-slab priors
\ba
p(\theta\,|\,\gamma,\sigma^2)\propto \prod_{(i,j) \in E}\exp\left(-\frac{(\theta_i-\theta_j)^2}{2\sigma^2[v_0\gamma_{ij}+v_1(1-\gamma_{ij})]}\right),
\ea
supported on $\Theta_w$, as we have seen in Example~\ref{ex:slr} and \ref{ex:cpm}. In this case, the above EM algorithm we have developed can also be extended to models with alternative prior distributions, such as the spike-and-slab Lasso prior \citep{rovckova2018spike} and the finite normal mixture prior \citep{stephens2016false}.

\subsection{General Graphs}\label{sec:alg-non}

When the base graph $G$ is not a tree, the E-step becomes computationally infeasible due to the lack of separability of $p(\theta|\gamma,\sigma^2)$ in $\gamma$. In fact, given the form of the density function in Proposition \ref{prop:Laplacian}, the main problem lies in the term $\sqrt{\det_w(L_{\gamma})}$, which cannot be factorized over $(i,j)\in E$ when the base graph $G=(V,E)$ is not a tree (Lemma \ref{lem:matrix-tree}). To overcome the difficulty, we consider optimizing a lower bound of the objective function (\ref{eq:EM-equivalent}). This means we need to find a good lower bound for $\log\det_w(L_{\gamma})$. Similar techniques are also advocated in the context of learning exponential family graphical models \citep{wainwright2008graphical}.

By Lemma \ref{lem:matrix-tree}, we can write
\begin{equation}
\log\text{det}_w(L_{\gamma})=\log\sum_{T\in\text{spt}(G)}\prod_{(i,j)\in T}\left[v_0^{-1}\gamma_{ij}+v_1^{-1}(1-\gamma_{ij})\right]+\log \frac{(\mathds{1}_p^Tw)^2}{\|w\|^2}.
\end{equation}
We only need to lower bound the first term on the right hand side of the equation above, because the second term is independent of $\gamma$. By Jensen's inequality, for any non-negative sequence $\{\lambda(T)\}_{T\in\text{spt}(G)}$ such that $\sum_{T\in\text{spt}(G)}\lambda(T)=1$, we have
\ba
	{}& \log\sum_{T\in\text{spt}(G)}\prod_{(i,j)\in T}\left[v_0^{-1}\gamma_{ij}+v_1^{-1}(1-\gamma_{ij})\right] \\
	\geq {}& \sum_{T\in\text{spt}(G)}\lambda(T)\log \prod_{(i,j)\in T}\left[v_0^{-1}\gamma_{ij}+v_1^{-1}(1-\gamma_{ij})\right] - \sum_{T\in\text{spt}(G)}\lambda(T)\log\lambda(T) \\
	= {}& \sum_{(i,j)\in E}\left(\sum_{T\in\text{spt}(G)}\lambda(T)\mathbb{I}\{(i,j)\in T\}\right)\log\left[v_0^{-1}\gamma_{ij}+v_1^{-1}(1-\gamma_{ij})\right] - \sum_{T\in\text{spt}(G)}\lambda(T)\log\lambda(T).
\ea
One of the most natural choices of the weights $\{\lambda(T)\}_{T\in\text{spt}(G)}$ is the uniform distribution
\ba
\lambda(T)=\frac{1}{|\text{spt}(G)|}.
\ea
This leads to the following lower bound
\begin{eqnarray}
\nonumber && \log\sum_{T\in\text{spt}(G)}\prod_{(i,j)\in T}\left[v_0^{-1}\gamma_{ij}+v_1^{-1}(1-\gamma_{ij})\right] \\
&\geq& \sum_{(i,j)\in E}r_{ij}\log\left[v_0^{-1}\gamma_{ij}+v_1^{-1}(1-\gamma_{ij})\right] + \log|\text{spt}(G)|, \label{eq:lower}
\end{eqnarray}
where
\begin{equation}
r_{ij} = \frac{1}{|\text{spt}(G)|}\sum_{T\in\text{spt}(G)}\mathbb{I}\{(i,j)\in T\}.\label{eq:effective-res}
\end{equation}

The quantity $r_{ij}$ defined in (\ref{eq:effective-res}) is recognized as the \textit{effective resistance} between the $i$th and the $j$th nodes \citep{lovasz1993random,ghosh2008minimizing}. Given a graph, we can treat each edge as a resistor with resistance $1$. Then, the effective resistance between the $i$th and the $j$th nodes is the resistance between $i$ and $j$ given by the whole graph. That is, if we treat the entire graph as a resistor. Let $L$ be the (unweighted) Laplacian matrix of the base graph $G=(V,E)$, and $L^+$ its pseudo-inverse. Then, an equivalent definition of (\ref{eq:effective-res}) is given by the formula
$$
r_{ij} = (e_i-e_j)^TL^+(e_i-e_j),
$$
where $e_j$ is the basis vector with the $i$th entry $1$ and the remaining entries $0$. Therefore, computation of the effective resistance can leverage fast Laplacian solvers in the literature \citep{spielman2004nearly,livne2012lean}. Some important examples of effective resistance are listed below:
\begin{itemize}
	\item When $G$ is the complete graph of size $p$, then $r_{ij} = 2/p$ for all $(i,j)\in E$.
	\item When $G$ is the complete bipartite graph of sizes $p$ and $k$, then $r_{ij} = \frac{p+k-1}{pk}$ for all $(i,j)\in E$.
	\item When $G$ is a tree, then $r_{ij} = 1$ for all $(i,j)\in E$.
	\item When $G$ is a two-dimensional grid graph of size $n_1\times n_2$, then $r_{ij} \in [0.5,0.75]$ depending on how close the edge $(i,j)$ is from its closest corner.
	\item When $G$ is a lollipop graph, the conjunction of a linear chain with size $p$ and a complete graph with size $k$, then $r_{ij} = 1$ or $2/k$ depending on whether the edge $(i,j)$ belongs to the chain or the complete graph.
\end{itemize}

By (\ref{eq:lower}), we obtain the following lower bound for the objective function (\ref{eq:EM-equivalent}),
\begin{equation} \label{eq:elbo}
\max_q\max_{\alpha,\theta\in\Theta_w,\eta,\sigma^2}\sum_{\gamma}q(\gamma)\log\frac{p(y\,|\,\alpha,\theta,\sigma^2)p(\alpha\,|\,\sigma^2)\wt{p}(\theta\,|\,\gamma,\sigma^2)p(\gamma\,|\,\eta)p(\eta)p(\sigma^2)}{q(\gamma)},
\end{equation}
where the formula of $\wt{p}(\theta\,|\,\gamma,\sigma^2)$ is obtained by applying the lower bound (\ref{eq:lower}) in the formula of $p(\theta\,|\,\gamma,\sigma^2)$ in Proposition \ref{prop:Laplacian}. Since $\wt{p}(\theta\,|\,\gamma,\sigma^2)$ can be factorized over $(i,j)\in E$, the E-step is given by $q^{\rm new}(\gamma)=\prod_{(i,j)\in E}(q_{ij}^{\rm new})^{\gamma_{ij}}(1-q_{ij}^{\rm new})^{1-\gamma_{ij}}$, where
\begin{equation}
q_{ij}^{\rm new} = \frac{\eta v_0^{-{r_{ij}}/{2}}e^{-(\theta_i-\theta_j)^2/2\sigma^2v_0}}{\eta v_0^{-{r_{ij}}/{2}}e^{-(\theta_i-\theta_j)^2/ 2\sigma^2v_0}+(1-\eta) v_1^{-{r_{ij}}/{2}}e^{-(\theta_i-\theta_j)^2/ 2\sigma^2v_1}}.\label{eq:E-general}
\end{equation}
Observe that the lower bound (\ref{eq:lower}) is independent of $\alpha,\theta,\eta,\sigma^2$, and thus the M-step remains the same as in the case of a tree base graph. The formulas are given by (\ref{eq:M1})-(\ref{eq:M3}), except that (\ref{eq:M3}) needs to be replaced by
$$\eta^{\rm new}=\frac{A-1+q^{\rm new}_{\rm sum}}{A+B+m-2}.$$

The EM algorithm for a general base graph can be viewed as a natural extension of that of a tree base graph. When $G=(V,E)$ is a tree, it is easy to see from the formula (\ref{eq:effective-res}) that $r_{ij}=1$ for all $(i,j)\in E$. In this case, the E-step (\ref{eq:E-general}) is reduced to (\ref{eq:E-tree}), and the inequality (\ref{eq:lower}) becomes an equality.

\subsection{Bayesian Model Selection}\label{sec:model-selection}

The output of the EM algorithm $\wh{q}(\gamma)$ can be understood as an estimator of the posterior distribution $p(\gamma|\wh{\alpha},\wh{\theta},\wh{\sigma}^2,\wh{\eta})$, where $\wh{\alpha},\wh{\theta},\wh{\sigma}^2,\wh{\eta}$ are obtained from the M-step. Then, we get a subgraph according to the thresholding rule $\wh{\gamma}_{ij}=\mathbb{I}\{\wh{q}_{ij}\geq 1/2\}$. It can be understood as a model learned from the data. The sparsity of the model critically depends on the values of $v_0$ and $v_1$ in the spike-and-slab Laplacian prior. With a fixed large value of $v_1$, we can obtain the solution path of $\wh{\gamma}=\wh{\gamma}(v_0)$ by varying $v_0$ from $0$ to $v_1$. The question then is how to select the best model along the solution path of the EM algorithm.

The strategy suggested by \cite{rovckova2014emvs} is to calculate the posterior score $p(\gamma|y)$ with respect to the Bayesian model of $v_0=0$. While the meaning of $p(\gamma|y)$ corresponding to $v_0=0$ is easily understood for the sparse linear regression setting in \cite{rovckova2014emvs}, it is less clear for a general base graph $G=(V,E)$.

In order to define a version of (\ref{eq:joint}) for $v_0=0$, we need to introduce the concept of \textit{edge contraction}. Given a $\gamma\in\{0,1\}^m$, the graph corresponding to the adjacency matrix $\gamma$ induces a partition of disconnected components $\{\mathcal{C}_1,...,\mathcal{C}_s\}$ of $[p]$. In other words, $\{i,j\}\subset \mathcal{C}_l$ for some $l\in[s]$ if and only if there is some path between $i$ and $j$ in the graph $\gamma$. For notational convenience, we define a vector $z\in[s]^n$ so that $z_i=l$ if and only if $i\in\mathcal{C}_l$. A membership matrix $Z_{\gamma}\in\{0,1\}^{p\times s}$ is defined with its $(i,l)$th entry being the indicator $\mathbb{I}\{z_i=l\}$.

We let $\wt{G}=(\wt{V},\wt{E})$ be a graph obtained from the base graph $G=(V,E)$ after the operation of edge contraction. In other words, every node in $\wt{G}$ is obtained by combining nodes in $G$ according to the partition of $\{\mathcal{C}_1,...,\mathcal{C}_s\}$. To be specific, $\wt{V}=[s]$, and $(k,l)\in \wt{E}$ if and only if there exists some $i\in\mathcal{C}_k$ and some $j\in\mathcal{C}_l$ such that $(i,j)\in E$.

Now we are ready to define a limiting version of (\ref{eq:Laplacian-prior}) as $v_0\rightarrow 0$. Let $\wt{L}_{\gamma}=D^T\diag(v_1^{-1}(1-\gamma))D$, which is the graph Laplacian of the weighted graph with adjacency matrix $\{v_1^{-1}(1-\gamma_{ij})\}$. Then, define
\begin{equation} 
p(\wt{\theta}\,|\,\gamma,\sigma^2) = \frac{1}{(2\pi\sigma^2)^{(s-1)/2}}\sqrt{\text{det}_{Z_{\gamma}^Tw}(Z_{\gamma}^T\wt{L}_{\gamma}Z_{\gamma})} \exp\left(-\frac{\wt{\theta}^TZ_{\gamma}^T\wt{L}_{\gamma}Z_{\gamma}\wt{\theta}}{2\sigma^2}\right)\mathbb{I}\{\wt{\theta}\in\Theta_{Z_{\gamma}^T w}\}.\label{eq:Laplacian-prior-low-d}
\end{equation}

With $\wt{G}=(\wt{V},\wt{E})$ standing for the contracted base graph, the prior distribution (\ref{eq:Laplacian-prior-low-d}) can also be written as
\begin{equation}
p(\wt{\theta}\,|\,\gamma,\sigma^2)\propto \exp\left(-\sum_{(k,l)\in\wt{E}}\frac{\omega_{kl}(\wt{\theta}_k-\wt{\theta}_l)^2}{2\sigma^2v_1}\right)\mathbb{I}\{\wt{\theta}\in\Theta_{Z_{\gamma}^Tw}\},\label{eq:Laplacian-prior-low-d-explicit}
\end{equation}
where $\omega_{kl}=\sum_{(i,j)\in E}\mathbb{I}\{z(i)=k,z(j)=l\}$, which means that the edges $\{(i,j)\}_{z(i)=k,z(j)=l}$ in the base graph $G=(V,E)$ are contracted as a new edge $(k,l)$ in $\wt{G}=(\wt{V},\wt{E})$ with $\omega_{kl}$ as the weight.

\begin{proposition}\label{prop:reduced-model}
	Suppose $G=(V,E)$ is connected and $\mathds{1}_p^Tw\neq 0$.  Let $Z_{\gamma}$ be the membership matrix defined as above. Then for any $\gamma\in\{0,1\}^m$, (\ref{eq:Laplacian-prior-low-d}) is a well-defined density function on the $(s-1)$-dimensional subspace $\{\wt{\theta}\in\mathbb{R}^s: w^TZ_{\gamma}\wt{\theta}=0\}$. Moreover, for an arbitrary design matrix $X\in\mathbb{R}^{n\times p}$, the distribution of $\theta$ that follows (\ref{eq:Laplacian-prior}) weakly converges to that of $Z_{\gamma}\wt{\theta}$ as $v_0\rightarrow 0$.
\end{proposition}

Motivated by Proposition \ref{prop:reduced-model}, a limiting version of (\ref{eq:joint}) for $v_0=0$ is defined as follows,
\begin{equation} \label{eqn:model_selection_y}
y\,|\,\alpha,\wt{\theta}, \gamma, \sigma^2 \sim N(X(\alpha w+Z_{\gamma}\wt{\theta}),\sigma^2 I_n).
\end{equation}
Then, $p(\wt{\theta}\,|\,\gamma,\sigma^2)$ is given by (\ref{eq:Laplacian-prior-low-d}), and $p(\alpha|\sigma^2)$, $p(\gamma|\eta)$, $p(\eta)$, $p(\sigma^2)$ are specified in (\ref{eq:alpha-prior}) and (\ref{eq:prior-gamma})-(\ref{eq:prior-sigma}). The posterior distribution of $\gamma$ has the formula
\begin{eqnarray*}
	p(\gamma\,|\,y) &\propto& \int\int\int\int p(y,\alpha,\wt{\theta},\gamma,\eta,\sigma^2)\,d\alpha\, d\wt{\theta}\,d\eta\, d\sigma^2 \\
	&=& \int p(\sigma^2)\int p(\alpha\,|\,\sigma^2)\int p(y\,|\,\alpha,\wt{\theta},\gamma,\sigma^2)p(\wt{\theta}\,|\,\gamma,\sigma^2)\,d\wt{\theta}\,d\alpha\, d\sigma^2 \int p(\gamma\,|\,\eta)p(\eta)d\eta.
\end{eqnarray*}
A standard calculation using conjugacy gives
\baa \label{model_selection}
p(\gamma\,|\,y) \propto& \left(\frac{\text{det}_{Z_\gamma^Tw}(Z_{\gamma}^T\wt{L}_{\gamma}Z_{\gamma})}{\text{det}_{Z_\gamma^T w}(Z_{\gamma}^T(X^TX+\wt{L}_{\gamma})Z_{\gamma})}\right)^{1/2} \left(\frac{\nu}{\nu+w^TX^T(I_n-R_{\gamma})Xw}\right)^{1/2} \\
& \times \left(y^T(I_n-R_{\gamma})y-\frac{|w^TX^T(I_n-R_{\gamma})y|^2}{\nu+w^TX^T(I_n-R_{\gamma})Xw}+b\right)^{-\frac{n+a}{2}} \\
& \times \frac{\text{Beta}\left(\sum_{(i,j)\in E}\gamma_{ij}+A-1,\,\sum_{(i,j)\in E}(1-\gamma_{ij})+B-1)\right)}{\text{Beta}(A,B)},
\eaa
where
$$R_{\gamma}=XZ_{\gamma}(Z_{\gamma}^T(X^TX+\wt{L}_{\gamma})Z_{\gamma})^{-1}Z_{\gamma}^TX^T.$$
This defines the model selection score $g(\gamma)=\log p(\gamma\,|\,y)$ up to a universal additive constant. The Bayesian model selection procedure evaluates $g(\gamma)$ on the solution path $\{\wh{\gamma}(v_0)\}_{0<v_0\leq v_1}$ and selects the best model with the highest value of $g(\gamma)$.

\subsection{Summary of Our Approach}
\label{sec:summary-of-the-approach}

The parameter $v_0$ plays a critical role in our Bayesian model selection procedure.
Recall that the joint distribution of our focus has the expression
\begin{equation}
p(y\,|\,\theta)p_{v_0}(\theta\,|\,\gamma)p(\gamma),\label{eq:simplified-model}
\end{equation}
where $p(y\,|\,\theta)$ is the linear model parametrized by $\theta$, $p_{v_0}(\theta\,|\,\gamma)$ is the spike-and-slab prior with tuning parameter $v_0$, and $p(\gamma)$ is the prior on the graph\footnote{We have ignored other parameters such as $\alpha,\eta,\sigma^2$ in order to make the discussion below clear and concise.}. The tuning parameter $v_0$ is the variance of the spike component of the prior. Different choices of $v_0$ are used as different components in our entire model selection procedure.
\begin{enumerate}
	\item The ideal choice of $v_0=0$ models exact sparsity in the sense that $\gamma_{ij}=1$ implies $\theta_i=\theta_j$. In this case, the exact posterior
	$$p_{v_0=0}(\gamma\,|\,y)\propto \int p(y\,|\,\theta)p_{v_0=0}(\theta\,|\,\gamma)p(\gamma)d\theta$$
	can be calculated according to the formulas that we derive in Section~\ref{sec:model-selection}. Then the ideal Bayes model selection procedure would be the one that maximizes the posterior $p_{v_0=0}(\gamma\,|\,y)$ over all $\gamma$. However, since this would require evaluating $p_{v_0=0}(\gamma\,|\,y)$ for exponentially many $\gamma$'s, it is sensible to maximize $p_{v_0=0}(\gamma\,|\,y)$ only over a carefully chosen subset of $\gamma$ that has a reasonable size for computational efficiency.
	\item The choice of $v_0>0$ models approximate sparsity in the sense that $\gamma_{ij}$ implies $\theta_i\approx\theta_j$. Though $v_0>0$ does not offer interpretation of exact sparsity, a nonzero $v_0$ leads to efficient computation via the EM algorithm. That is, for a $v_0>0$, we can maximize
	$$\max_q\max_{\theta}\sum_{\gamma}q(\gamma)\log\frac{p(y\,|\,\theta)p_{v_0}(\theta\,|\,\gamma)p(\gamma)}{q(\gamma)},$$
	which is the objective function of EM. Denote the output of the algorithm by $q_{v_0}(\gamma)=\prod_{ij}q_{ij,v_0}$, we then obtain our model by $\widehat{\gamma}_{ij}(v_0)=\mathbb{I}\{q_{ij,v_0}>0.5\}$. As we vary $v_0$ on a grid from $0$ to $v_1$, we obtain a path of models $\{\widehat{\gamma}(v_0)\}_{0<v_0\leq v_1}$. It covers models that ranges from very parsimonious ones to the fully saturated one.
\end{enumerate}

The proposed model selection procedure in Section~\ref{sec:model-selection} is
\begin{equation}
\max_{\gamma\in \{\widehat{\gamma}(v_0)\}_{0<v_0\leq v_1}}p_{v_0=0}(\gamma\,|\,y),\label{eq:proposed-msg}
\end{equation}
where $p_{v_0=0}(\gamma\,|\,y)\propto \int p(y\,|\,\theta)p_{v_0=0}(\theta\,|\,\gamma)p(\gamma)d\theta$.
That is, we optimize the posterior distribution $p_{v_0=0}(\gamma\,|\,y)$ \textit{only} over the EM solution path. The best one among all the candidate models will be selected according to $p_{v_0=0}(\gamma\,|\,y)$, which is the exact/full posterior of $\gamma$. There are two ways to interpret our model selection procedure \eqref{eq:proposed-msg}:
\begin{enumerate}
	\item The procedure (\ref{eq:proposed-msg}) can be understood as a computationally efficient approximation strategy to the ideal Bayes model selection procedure $\max_{\gamma}p_{v_0=0}(\gamma\,|\,y)$ that is infeasible to compute. The EM algorithm with various choices of $v_0$ simply provides a short list of candidate models. From this perspective, the proposed procedure (\ref{eq:proposed-msg}) is fully Bayes.
	\item The procedure (\ref{eq:proposed-msg}) can be also understood as a method for selecting the tuning parameter $v_0$, because the maximizer of (\ref{eq:proposed-msg}) must be in the form of $\widehat{\gamma}(\widehat{v}_0)$ for some data-driven $\widehat{v}_0$. In this way, the solution $\widehat{\gamma} = \widehat{\gamma}(\widehat{v}_0)$ also has a non-Bayesian interpretation since it is obtained by post-processing the EM solution with the tuning parameter $\widehat{v}_0$ selected by (\ref{eq:proposed-msg}). In this regard, $\widehat\gamma$ also can be thought of as an empirical Bayes estimator.
\end{enumerate}
In summary, our proposed procedure (\ref{eq:proposed-msg}) is motivated by both statistical and computational considerations.

\section{Clustering: A New Deal}\label{sec:clustering}

\subsection{A Multivariate Extension} \label{sec:multivariate}

Before introducing our new Bayesian clustering model, we need a multivariate extension of the general framework (\ref{eq:joint}) to model a matrix observation $y\in\mathbb{R}^{n\times d}$. With a design matrix $X\in\mathbb{R}^{n\times p}$, the dimension of $\theta$ is now $p\times d$. We denote the $i$th row of $\theta$ by $\theta_i$. With the grounding vector $w\in\mathbb{R}^p$, the distribution $p(y|\alpha,\theta,\sigma^2)p(\alpha|\sigma^2)p(\theta|\gamma,\sigma^2)$ is given by
\begin{eqnarray}
\label{eq:likelihood-multi} y\,|\,\alpha,\theta,\sigma^2 &\sim& N(X(w\alpha^T+\theta),\sigma^2 I_n\otimes I_d), \\
\label{eq:prior-alpha-multi} \alpha\,|\,\sigma^2 &\sim& N\left(0,\frac{\sigma^2}{\nu}I_d\right), \\
\label{eq:prior-multi-theta} \theta\,|\,\gamma,\sigma^2 &\sim& p(\theta|\gamma,\sigma^2) \propto \prod_{(i,j)\in E}\exp\left(-\frac{\|\theta_i-\theta_j\|^2}{2\sigma^2[v_0\gamma_{ij}+v_1(1-\gamma_{ij})]}\right)\mathbb{I}\{\theta\in\Theta_w\},
\end{eqnarray}
where $\Theta_w=\{\theta\in\mathbb{R}^{p\times d}:w^T\theta=0\}$. The prior distributions on $\gamma,\eta,\sigma^2$ are the same as (\ref{eq:prior-gamma})-(\ref{eq:prior-sigma}). Moreover, the multivariate spike-and-slab Laplacian prior (\ref{eq:prior-multi-theta}) is supported on a $d(p-1)$-dimensional subspace $\Theta_w$, and is well-defined as long as $\mathds{1}_p^Tw\neq 0$ for the same reason stated in Proposition \ref{prop:Laplacian}.

The multivariate extension can be understood as the task of learning $d$ individual graphs for each column of $\theta$. Instead of modeling the $d$ graph separately by $\gamma^{(1)},...,\gamma^{(d)}$ using (\ref{eq:joint}), we assume the $d$ columns of $\theta$ share the same structure by imposing the condition $\gamma^{(1)}=...=\gamma^{(d)}$.

An immediate example is a Bayesian multitask learning problem with group sparsity. It can be viewed as a multivariate extension of Example \ref{ex:slr}. With the same argument in Example \ref{ex:slr}, (\ref{eq:likelihood-multi})-(\ref{eq:prior-multi-theta}) is specialized to
\begin{eqnarray*}
	y\,|\,\theta,\sigma^2 &\sim& N(X\theta,\sigma^2 I_n\otimes I_d), \\
	\theta\,|\,\gamma,\sigma^2 &\sim& p(\theta\,|\,\gamma,\sigma^2) \propto \prod_{i=1}^p\exp\left(-\frac{\|\theta_i\|^2}{2\sigma^2[v_0\gamma_i+v_1(1-\gamma_i)]}\right).
\end{eqnarray*}

To close this subsection, let us mentions that the model (\ref{eq:likelihood-multi})-(\ref{eq:prior-multi-theta}) can be easily modified to accommodate a heteroscedastic setting. For example, one can replace the $\sigma^2 I_n\otimes I_d$ in (\ref{eq:likelihood-multi}) by a more general $I_n\otimes \diag(\sigma_1^2,...,\sigma_d^2)$, and then make corresponding changes to (\ref{eq:prior-alpha-multi}) and (\ref{eq:prior-multi-theta}) as well.

\subsection{Model Description}\label{sec:clustering-model}

Consider the likelihood
\begin{equation}
y\,|\,\alpha,\theta,\sigma^2\sim N(\mathds{1}_n\alpha^T+\theta,\sigma^2I_n\otimes I_d),\label{eq:likelihood-clustering}
\end{equation}
with the prior distribution of $\alpha|\sigma^2$ specified by  (\ref{eq:prior-alpha-multi}).
The clustering model uses the following form of (\ref{eq:prior-multi-theta}),
\begin{equation}
p(\theta,\mu\,|\,\gamma,\sigma^2) \propto \prod_{i=1}^n\prod_{j=1}^k\exp\left(-\frac{\|\theta_i-\mu_j\|^2}{2\sigma^2[v_0\gamma_{ij}+v_1(1-\gamma_{ij})]}\right)\mathbb{I}\{\mathds{1}_n^T\theta=0\}. \label{eq:spike-and-slab-clustering}
\end{equation}
Here, both vectors $\theta_i$ and $\mu_j$ are in $\mathbb{R}^d$. The prior distribution (\ref{eq:spike-and-slab-clustering}) can be derived from (\ref{eq:prior-multi-theta}) by replacing $\theta$ in (\ref{eq:prior-multi-theta}) with $(\theta,\mu)$ and specifying the base graph as a complete bipartite graph between $\theta$ and $\mu$. We impose the restriction that
\begin{equation}
\sum_{j=1}^k\gamma_{ij}=1,\label{eq:gamma-clustering-condition}
\end{equation}
for all $i\in[n]$. Then, $\mu_1,...,\mu_k$ are latent variables that can be interpreted as the clustering centers, and each $\theta_i$ is connected to one of the clustering centers.

To fully specify the clustering model, the prior distribution of $\gamma$ is given by (\ref{eq:prior-gamma}) with $\Gamma$ being the set of all $\{\gamma_{ij}\}$ that satisfies (\ref{eq:gamma-clustering-condition}). Equivalently,
\begin{equation}
(\gamma_{i1},...,\gamma_{ik})\sim \text{Uniform}(\{e_j\}_{j=1}^k),\label{eq:clustering-prior-gamma-uniform}
\end{equation}
independently for all $i\in[n]$, where $e_j$ is a vector with $1$ on the $j$th entry and $0$ elsewhere. Finally, the prior of $\sigma^2$ is given by (\ref{eq:prior-sigma}).

\subsection{EM Algorithm}\label{sec:clustering-EM}

The EM algorithm can be derived by following the idea developed in Section \ref{sec:alg-non}. In the current setting, the lower bound (\ref{eq:lower}) becomes
\begin{eqnarray}
\nonumber && \log\sum_{T\in\text{spt}(K_{n,k})}\prod_{i=1}^n\prod_{j=1}^k[v_0^{-1}\gamma_{ij}+v_1^{-1}(1-\gamma_{ij})] \\
\label{eq:lower-clustering} &\geq& \sum_{i=1}^n\sum_{j=1}^k r_{ij}\log\left[v_0^{-1}\gamma_{ij}+v_1^{-1}(1-\gamma_{ij})\right] + \log|\text{spt}(K_{n,k})|,
\end{eqnarray}
where $K_{n,k}$ is the complete bipartite graph. By symmetry, the effective resistance $r_{ij}=r$ is a constant independent of $(i,j)$. Thus, (\ref{eq:lower-clustering}) can be written as
\begin{eqnarray}
\label{eq:audi-r8} && r\sum_{i=1}^n\sum_{j=1}^k\log\left[v_0^{-1}\gamma_{ij}+v_1^{-1}(1-\gamma_{ij})\right] + \log|\text{spt}(K_{n,k})| \\
\nonumber &=& r\log(v_0^{-1})\sum_{i=1}^n\sum_{j=1}^k\gamma_{ij} + r\log(v_1^{-1})\sum_{i=1}^n\sum_{j=1}^k(1-\gamma_{ij}) + \log|\text{spt}(K_{n,k})| \\
\label{eq:amg-gt} &=& rn\log(v_0^{-1}) + rn(k-1)\log(v_1^{-1}) + \log|\text{spt}(K_{n,k})|,
\end{eqnarray}
where the last equality is derived from (\ref{eq:gamma-clustering-condition}). Therefore, for the clustering model, the lower bound (\ref{eq:lower}) is a constant independent of $\{\gamma_{ij}\}$. As a result, the lower bound of the objective function of the EM algorithm becomes
\begin{equation}
\sum_{\gamma}q(\gamma)\log\frac{p(y\,|\,\alpha,\theta,\sigma^2)p(\alpha\,|\,\sigma^2)\wt{p}(\theta,\mu\,|\,\gamma,\sigma^2)p(\gamma)p(\sigma^2)}{q(\gamma)},\label{eq:obj-EM-clustering}
\end{equation}
with
$$\wt{p}(\theta,\mu\,|\,\gamma,\sigma^2)=\text{const}\times \frac{1}{(2\pi\sigma^2)^{(n+k-1)d/2}}\prod_{i=1}^n\prod_{j=1}^k\exp\left(-\frac{\|\theta_i-\mu_j\|^2}{2\sigma^2[v_0\gamma_{ij}+v_1(1-\gamma_{ij})]}\right),$$
and the algorithm is to maximize (\ref{eq:obj-EM-clustering}) over $q,\alpha,\theta\in\{\theta:\mathds{1}_n^T\theta=0\}$ and $\sigma^2$.

Maximizing (\ref{eq:obj-EM-clustering}) over $q$, we obtain the E-step as
\begin{equation}
q_{ij}^{\rm new}=\frac{\exp\left(-\frac{\|\theta_i-\mu_j\|^2}{2\sigma^2\bar{v}}\right)}{\sum_{l=1}^k\exp\left(-\frac{\|\theta_i-\mu_l\|^2}{2\sigma^2\bar{v}}\right)},\label{eq:E-clustering}
\end{equation}
independently for all $i\in[n]$,
where $\bar{v}^{-1}=v_0^{-1}-v_1^{-1}$, and we have used the notation
$$q_{ij}=q\left((\gamma_{i1},...,\gamma_{ik})=e_j\right).$$
It is interesting to note that the E-step only depends on $v_0$ and $v_1$ through $\bar{v}$. Maximizing (\ref{eq:obj-EM-clustering}) over $\alpha,\theta\in\{\theta:\mathds{1}_n^T\theta=0\},\mu,\sigma^2$, we obtain the M-step as
\begin{eqnarray}
\label{eq:M-clustering} (\alpha^{\rm new}, \theta^{\rm new},\mu^{\rm new}) &=& \argmin_{\alpha, \mathds{1}_n^T\theta=0,\mu}F(\alpha,\theta,\mu;q^{\rm new}), \\
\nonumber (\sigma^2)^{\rm new} &=& \frac{F(\alpha^{\rm new},\theta^{\rm new},\mu^{\rm new};q^{\rm new})+b}{(2n+k)d+a+2},
\end{eqnarray}
where
$$F(\alpha,\theta,\mu;q)=\fnorm{y-\mathds{1}_n\alpha^T-\theta}^2+\nu\|\alpha\|^2+\sum_{i=1}^n\sum_{j=1}^k\left(\frac{q_{ij}}{v_0}+\frac{1-q_{ij}}{v_1}\right)\|\theta_i-\mu_j\|^2.$$
Note that for all $\theta$ such that $\mathds{1}_n^T\theta=0$, we have
$$\fnorm{y-\mathds{1}_n\alpha^T-\theta}^2=\fnorm{n^{-1}\mathds{1}_n\mathds{1}_n^Ty-\mathds{1}_n\alpha^T}^2+\fnorm{y-n^{-1}\mathds{1}_n\mathds{1}_n^Ty-\theta}^2,$$
and thus the M-step (\ref{eq:M-clustering}) can be solved separately for $\alpha$ and $(\theta,\mu)$.

\subsection{A Connection to Bipartite Graph Projection}

The clustering model (\ref{eq:spike-and-slab-clustering}) involves latent variables $\mu$ that do not appear in the likelihood (\ref{eq:likelihood-clustering}). This allows us to derive an efficient EM algorithm in Section~\ref{sec:clustering-EM}. To better understand (\ref{eq:spike-and-slab-clustering}), we connect the bipartite graphical structure between $\theta$ and $\mu$ to a graphical structure on $\theta$ alone. Given a $\gamma=\{\gamma_{ij}\}$ that satisfies (\ref{eq:gamma-clustering-condition}), we call it non-degenerate if $\sum_{i=1}^n\gamma_{ij}>0$ for all $j\in[k]$. In other words, none of the $k$ clusters is empty.

\begin{proposition}\label{prop:projection-clustering}
	Let the conditional distribution of $\theta,\mu\,|\,\gamma,\sigma^2$ be specified by (\ref{eq:spike-and-slab-clustering}) with some non-degenerate $\gamma$. Then, the distribution of $\theta\,|\,\gamma,\sigma^2$ weakly converges to
	\begin{equation}
	p(\theta\,|\,\gamma,\sigma^2)\propto \prod_{1\leq i<l\leq n}\exp\left(-\frac{\lambda_{il}\|\theta_i-\theta_l\|^2}{2\sigma^2v_0}\right)\mathbb{I}\{\mathds{1}_n^T\theta=0\},\label{eq:projection-clustering}
	\end{equation}
	as $v_1\rightarrow\infty$, where $\lambda_{il}=\sum_{j=1}^k\gamma_{ij}\gamma_{lj}/n_j$ with $n_j=\sum_{i=1}^n\gamma_{ij}$ being the size of the $j$th cluster.
\end{proposition}
The formula (\ref{eq:projection-clustering}) resembles (\ref{eq:Laplacian-prior}), except that $\lambda$ encodes a clustering structure.
By the definition of $\lambda_{il}$, if $\theta_i$ and $\theta_l$ are in different clusters, $\lambda_{il}=0$, and otherwise, $\lambda_{il}$ takes the inverse of the size of the cluster that both $\theta_i$ and $\theta_l$ belong to. The relation between (\ref{eq:spike-and-slab-clustering}) and (\ref{eq:projection-clustering}) can be understood from the operations of graph projection and graph lift, in the sense that the weighted graph $\lambda$, with nodes $\theta$, is a projection of $\gamma$, a bipartite graph between nodes $\theta$ and nodes $\mu$. Conversely, $\gamma$ is said to be a graph lift of $\lambda$. Observe that the clustering structure of (\ref{eq:projection-clustering}) is combinatorial, and therefore it is much easier to work with the bipartite structure in (\ref{eq:spike-and-slab-clustering}) with latent variables.

\subsection{A Connection to Gaussian Mixture Models}\label{sec:gmm}

We establish a connection to Gaussian mixture models. We first give the following result.
\begin{proposition}\label{prop:gmm}
	Let the conditional distribution of $\theta,\mu\,|\,\gamma,\sigma^2$ be specified by (\ref{eq:spike-and-slab-clustering}). Then, as $v_0\rightarrow 0$, this conditional distribution weakly converges to the distribution specified by the following sampling process: $\theta=\gamma\mu$ and $\mu\,|\,\gamma,\sigma^2$ is sampled from
	\begin{equation}
	p(\mu\,|\,\gamma,\sigma^2)\propto \prod_{1\leq j<l\leq k}\exp\left(-\frac{(n_j+n_l)\|\mu_j-\mu_l\|^2}{2\sigma^2 v_1}\right)\mathbb{I}\{\mathds{1}_n^T\gamma\mu=0\},\label{eq:clustering-marginal-mu}
	\end{equation}
	where $n_j=\sum_{i=1}^n\gamma_{ij}$ being the size of the $j$th cluster.
\end{proposition}
With this proposition, we can see that as $v_0\rightarrow 0$, the clustering model specified by (\ref{eq:likelihood-clustering}), (\ref{eq:prior-alpha-multi}) and (\ref{eq:spike-and-slab-clustering}) becomes
\begin{equation}
y\,|\,\alpha,\mu,\gamma,\sigma^2 \sim N(\mathds{1}_n\alpha^T+\gamma\mu,\sigma^2I_n\otimes I_d), \label{eq:gmm-likelihood}
\end{equation}
with $\alpha\,|\,\sigma^2$ distributed by (\ref{eq:prior-alpha-multi}) and $\mu\,|\,\gamma,\sigma^2$ distributed by (\ref{eq:clustering-marginal-mu}). The likelihood function (\ref{eq:gmm-likelihood}) is commonly used in Gaussian mixture models, which encodes an exact clustering structure. Therefore, with a finite $v_0$, the model specified by (\ref{eq:likelihood-clustering}), (\ref{eq:prior-alpha-multi}) and (\ref{eq:spike-and-slab-clustering}) can be interpreted as a relaxed version of the Gaussian mixture models that leads to an approximate clustering structure.

\subsection{Adaptation to the Number of Clusters}\label{sec:no-cluster}

The number $k$ in (\ref{eq:spike-and-slab-clustering}) should be understood as an upper bound of the number of clusters. Even though the EM algorithm outputs $k$ cluster centers $\{\wh{\mu}_1,...,\wh{\mu}_k\}$, these $k$ cluster centers will be automatically grouped according to their own closedness as we vary the value of $v_0$. Generally speaking, for a very small $v_0$ (the Gaussian mixture model in Section \ref{sec:gmm}, for example), $\{\wh{\mu}_1,...,\wh{\mu}_k\}$ will take $k$ vectors that are not close to each other. As we increase $v_0$, the clustering centers $\{\wh{\mu}_1,...,\wh{\mu}_k\}$ start to merge, and eventually for a sufficiently large $v_0$, they will all converge to a single vector. In short, $v_0$ parametrizes the solution path of our clustering algorithm, and on this solution path, the \textit{effective} number of clusters increases as the value of $v_0$ increases.

We illustrate this point by a simple numerical example. Consider the observation $y=(4,2,-2,4)^T\in\mathbb{R}^{4\times 1}$. We fit our clustering model with $k\in\{2,3,4\}$. Figure \ref{fig:clustering-example} visualizes the output of the EM algorithm $(\wh{\mu},\wh{\theta})$ as $v_0$ varies. It is clear that the solution path always starts at $\{\wh{\mu}_1,...,\wh{\mu}_k\}$ of different values. Then, as $v_0$ increases, the solution path has various phase transitions where the closest two $\wh{\mu}_j$'s merge. In the end, for a sufficiently large $v_0$, the clustering centers $\{\wh{\mu}_1,...,\wh{\mu}_k\}$ all merge to a common value.
\begin{figure}[!ht]
	\centering
	\includegraphics[width=1.0\textwidth]{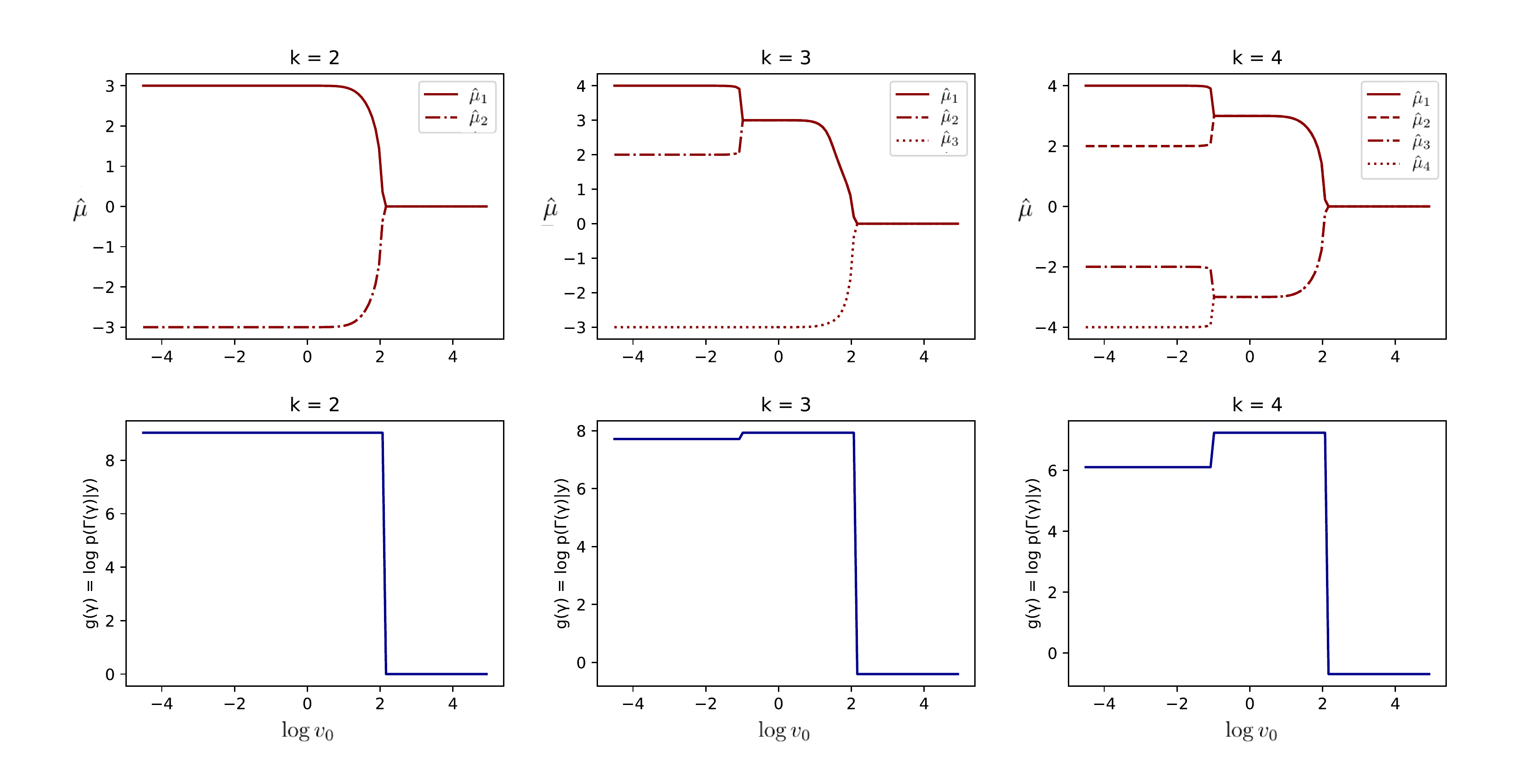}
	\caption{(\textrm{Top}) Solution paths of $\wh{\mu}$ with different choices of $k$; (\textrm{Bottom}) Model selection scores on the solution paths.}\label{fig:clustering-example} 
\end{figure}

To explain this phenomenon, it is most clear to investigate the case $k=2$. Then, the M-step (\ref{eq:M-clustering}) updates $\mu$ according to
\begin{eqnarray*}
	\mu_1^{\rm new} &=& \argmin_{\mu_1}\sum_{i=1}^n\left(\frac{q_{i1}}{v_0}+\frac{1-q_{i1}}{v_1}\right)\|\theta_i-\mu_1\|^2 \\
	\mu_1^{\rm new} &=& \argmin_{\mu_2}\sum_{i=1}^n\left(\frac{q_{i2}}{v_0}+\frac{1-q_{i2}}{v_1}\right)\|\theta_i-\mu_2\|^2.
\end{eqnarray*}
Observe that both $\mu_1^{\rm new}$ and $\mu_1^{\rm new}$ are weighted averages of $\{\theta_1,...,\theta_n\}$, and the only difference between $\mu_1^{\rm new}$ and $\mu_1^{\rm new}$ lies in the weights. According to the E-step (\ref{eq:E-clustering}),
\begin{eqnarray*}
	q_{i1} &=& \frac{\exp\left(-\frac{\|\theta_i-\mu_1\|^2}{2\sigma^2\bar{v}}\right)}{\exp\left(-\frac{\|\theta_i-\mu_1\|^2}{2\sigma^2\bar{v}}\right)+\exp\left(-\frac{\|\theta_i-\mu_2\|^2}{2\sigma^2\bar{v}}\right)} \\
	q_{i2} &=& \frac{\exp\left(-\frac{\|\theta_i-\mu_2\|^2}{2\sigma^2\bar{v}}\right)}{\exp\left(-\frac{\|\theta_i-\mu_1\|^2}{2\sigma^2\bar{v}}\right)+\exp\left(-\frac{\|\theta_i-\mu_2\|^2}{2\sigma^2\bar{v}}\right)},
\end{eqnarray*}
and we recall that $\bar{v}^{-1}=v_0^{-1}-v_1^{-1}$. Therefore, as $v_0\rightarrow\infty$, $q_{i1}\rightarrow 1/2$ and $q_{i2}\rightarrow 1/2$, which results in the phenomenon that $\mu_1^{\rm new}$ and $\mu_2^{\rm new}$ merge to the same value. The same reasoning also applies to $k\geq 2$.

\subsection{Model Selection} \label{sec:modelselectioncluster}

In this section, we discuss how to select a clustering structure from the solution path of the EM algorithm. According to the discussion in Section \ref{sec:no-cluster}, we should understand $k$ as an upper bound of the number of clusters, and the estimator of the number of clusters will be part of the output of the model selection procedure. The general recipe of our method follows the framework discussed in Section \ref{sec:model-selection}, but some nontrivial twist is required in the clustering problem. To make the presentation clear, the model selection procedure will be introduced in two parts. We will first propose our model selection score, and then we will describe a method that extracts a clustering structure from the output of the EM algorithm.

\subsubsection{The model selection score}
For any $\gamma\in\{0,1\}^{n\times k}$ that satisfies (\ref{eq:gamma-clustering-condition}), we can calculate the posterior probability $p(\gamma\,|\,y)$ with $v_0=0$. This can be done by the connection to Gaussian mixture models discussed in Section \ref{sec:gmm}. To be specific, the calculation follows the formula
$$p(\gamma\,|\,y) = \int\int\int p(y\,|\,\alpha,\mu,\gamma,\sigma^2)p(\alpha\,|\,\sigma^2) p(\mu\,|\,\gamma,\sigma^2) p(\sigma^2) p(\gamma)\,d\alpha\, d\mu\, d\sigma^2,$$
where $p(y\,|\,\alpha,\mu,\gamma,\sigma^2)$, $p(\alpha|\sigma^2)$, $p(\mu|\gamma,\sigma^2)$, $p(\sigma^2)$, and $p(\gamma)$ are specified by (\ref{eq:likelihood-clustering}), (\ref{eq:prior-alpha-multi}), (\ref{eq:clustering-marginal-mu}), (\ref{eq:prior-sigma}) and (\ref{eq:clustering-prior-gamma-uniform}). A standard calculation gives the formula
\baa \label{ms2}
p(\gamma|y)  \propto {}& \left(\frac{\nu}{\nu+n} \times \frac{\text{det}_{\gamma^T\mathds{1}_n}(\bar{L}_{\gamma})}{\text{det}_{\gamma^T\mathds{1}_n}(\bar{L}_{\gamma}+\gamma^T\gamma)}\right)^{d/2} \times \Bigg[\frac{\nu n}{\nu+n}\|n^{-1}\mathds{1}_ny\|^2 \\
{}& +\Tr\left((y-n^{-1}\mathds{1}_n\mathds{1}_n^Ty)^T(I_n-\gamma(\bar{L}_{\gamma}+\gamma^T\gamma)\gamma^T) (y-n^{-1}\mathds{1}_n\mathds{1}_n^Ty)\right)\Bigg]^{-\frac{nd+a}{2}},
\eaa
where $\bar{L}_{\gamma} = (\mathds{1}_{k\times n}\gamma+\gamma^T\mathds{1}_{n\times k}-2\gamma^T\gamma)/v_1$ is the graph Laplacian of the weighted adjacency matrix, which satisfies
$$\Tr(\mu^T\bar{L}_{\gamma}\mu)=\sum_{1\leq j<l\leq k}\frac{n_j+n_l}{v_1}\|\mu_j-\mu_l\|^2.$$
Recall that $n_j=\sum_{i=1}^n\gamma_{ij}$ is the size of the $j$th cluster.

However, the goal of the model selection is to select a clustering structure, and it is possible that different $\gamma$'s may correspond to the same clustering structure due to label permutation. To overcome this issue, we need to sum over all equivalent $\gamma$'s. Given a $\gamma\in\{0,1\}^{n\times k}$ that satisfies (\ref{eq:gamma-clustering-condition}), define a symmetric membership matrix $\Gamma(\gamma)\in\{0,1\}^{n\times n}$ by
$$\Gamma_{il}(\gamma) = \mathbb{I}\left\{\sum_{j=1}^k\gamma_{ij}\gamma_{lj}=1\right\}.$$
In other words, $\Gamma_{il}(\gamma)=1$ if and only if $i$ and $l$ are in the same cluster. It is easy to see that every clustering structure can be uniquely represented by a symmetric membership matrix. We define the posterior probability of a clustering structure $\Gamma$ by
$$p(\Gamma\,|\,y) = \sum_{\gamma\in\{\gamma:\Gamma(\gamma)=\Gamma\}}p(\gamma\,|\,y).$$
The explicit calculation of the above summation is not necessary. A shortcut can be derived from the fact that $p(\gamma\,|\,y)=p(\gamma'\,|\,y)$ if $\Gamma(\gamma)=\Gamma(\gamma')$. This immediately implies that $p(\Gamma\,|\,y)=|\Gamma|\cdot p(\gamma\,|\,y)$ for any $\gamma$ that satisfies $\Gamma(\gamma)=\Gamma$. Suppose $\Gamma$ encodes a clustering structure with $\wt{k}$ nonempty clusters, and then we have $|\Gamma|={k \choose \widetilde{k}} \wt{k}!$. This leads to the model selection score
\begin{equation}
g(\gamma) = \log p(\Gamma(\gamma)\,|\,y) = \log p(\gamma\,|\,y) + \log\left[ {k\choose \widetilde{k}} \wt{k}!\right], \label{eq:ms-score-clustering}
\end{equation}
for any $\gamma\in\{0,1\}^{n\times k}$ that satisfies (\ref{eq:gamma-clustering-condition}), and $\wt{k}$ above is calculated by
$\wt{k}=\sum_{j=1}^k\max_{1\leq i\leq n}\gamma_{ij}$,
the effective number of clusters.

\subsubsection{Extraction of clustering structures from the EM algorithm}
Let $\wh{\mu}$ and $\wh{q}$ be outputs of the EM algorithm, and we discuss how to obtain $\wh{\gamma}$ that encodes a meaningful clustering structure to be evaluated by the model selection score (\ref{eq:ms-score-clustering}). It is very tempting to directly threshold $\wh{q}$ as is done in Section \ref{sec:model-selection}. However, as has been discussed in Section \ref{sec:no-cluster}, the solution paths of $\{\mu_1,...,\mu_k\}$ merge at some values of $v_0$. Therefore, we should treat the clusters whose clustering centers merge together as a single cluster.

Given $\wh{\mu}_1,...,\wh{\mu}_k$ output by the M-step, we first merge $\wh{\mu}_j$ and $\wh{\mu}_l$ whenever $\|\wh{\mu}_j-\wh{\mu}_l\|\leq \epsilon$. The number $\epsilon$ is taken as $10^{-8}$, the square root of the machine precision, in our code. This forms a partition $[k]=\cup_{l=1}^{\wh{k}}\mathcal{G}_l$ for some $\wh{k}\leq k$. Then, by taking average within each group, we obtain a reduced collection of clustering centers $\wt{\mu}_1,...,\wt{\mu}_{\wh{k}}$. In other words, $\wt{\mu}_l=|\mathcal{G}_l|^{-1}\sum_{j\in \mathcal{G}_l}\wh{\mu}_j$.

The $\wh{q}\in[0,1]^{n\times k}$ output by the E-step should also be reduced to $\wt{q}\in[0,1]^{n\times \wh{k}}$ as well. Note that $\wh{q}_{ij}$ is the estimated posterior probability that the $i$th node belongs to the $j$th cluster. This means that $\wt{q}_{il}$ is the estimated posterior probability that the $i$th node belongs to the $l$th reduced cluster. An explicit formua is given by $\wt{q}_{il}=\sum_{j\in \mathcal{G}_l}\wh{q}_{ij}$.

With the reduced posterior probability $\wt{q}$, we simply apply thresholding to obtain $\wh{\gamma}$. We have $(\wh{\gamma}_{i1},...,\wh{\gamma}_{ik})=e_j$ if $j=\argmax_{1\leq l\leq\wh{k}}\wt{q}_{il}$. Recall that $e_j$ is a vector with $1$ on the $j$th entry and $0$ elsewhere. Note that according to this construction, we always have $\wh{\gamma}_{ij}=0$ whenever $j> \wh{k}$. This does not matter, because the model selection score (\ref{eq:ms-score-clustering}) does not depend on the clustering labels. Finally, the $\wh{\gamma}$ constructed according to the above procedure will be evaluated by $g(\wh{\gamma})$ defined by (\ref{eq:ms-score-clustering}).

In the toy example with four data points $y=(4,2,-2,4)^T$, the model selection score is computed along the solution path. According to Figure \ref{fig:clustering-example}, the model selection procedure suggests that a clustering structure with two clusters $\{4,2\}$ and $\{-2,-4\}$ is the most plausible one. We also note that the curve of $g(\gamma)$ has sharp phase transitions whenever the solution paths $\mu$ merge.

\section{Extensions with Graph Algebra}\label{sec:algebra}

In many applications, it is useful to have a model that imposes both row and column structures on a high-dimensional matrix $\theta\in\mathbb{R}^{p_1\times p_2}$. We list some important examples below.
\begin{enumerate}
	\item \textit{Biclustering.} In applications such as gene expression data analysis, one needs to cluster both samples and features. This task imposes a clustering structure for both rows and columns of the data matrix \citep{hartigan1972direct,cheng2000biclustering}.
	\item \textit{Block sparsity.} In problems such as planted clique detection \citep{feige2000finding} and submatrix localization \citep{hajek2017submatrix}, the matrix can be viewed as the sum of a noise background plus a submatrix of signals with unknown locations. Equivalently, it can be modeled by simultaneous row and column sparsity \citep{ma2015volume}.
	\item \textit{Sparse clustering.} Suppose the data matrix exhibits a clustering structure for its rows and a sparsity structure for its columns, then we have a sparse clustering problem \citep{witten2010framework}. For this task, we need to select nonzero column features in order to accurately cluster the rows.
\end{enumerate}
For the problems listed above, the row and column structures can be modeled by graphs $\gamma_1$ and $\gamma_2$. Then, the structure of the matrix $\theta$ is induced by a notion of graph product of $\gamma_1$ and $\gamma_2$. In this section, we introduce tools from graph algebra including Cartesian product and Kronecker product to build complex structure from simple components.

We first introduce the likelihood of the problem. To cope with many useful models, we assume that the observation can be organized as a matrix $y\in\mathbb{R}^{n_1\times n_2}$. Then, the specific setting of a certain problem can be encoded by design matrices $X_1\in\mathbb{R}^{n_1\times p_1}$ and $X_2\in\mathbb{R}^{n_2\times p_2}$. The likelihood is defined by
\begin{equation}
y\,|\,\alpha,\theta,\sigma^2\sim N(X_1(\alpha w+\theta)X_2^T,\sigma^2 I_{n_1}\otimes I_{n_2}).\label{eq:likelihood-product2}
\end{equation}
The matrix $w\in\mathbb{R}^{p_1\times p_2}$ is assumed to have rank one, and can be decomposed as $w=w_1w_2^T$ for some $w_1\in\mathbb{R}^{p_1}$ and $w_2\in\mathbb{R}^{p_2}$. The prior distribution of the scalar is simply given by
\begin{equation}
\alpha\,|\,\sigma^2\sim N(0,\sigma^2/\nu).\label{eq:alpha-product}
\end{equation}
We then need to build prior distributions of $\theta$ that is supported on $\Theta_w=\{\theta\in\mathbb{R}^{p_1\times p_2}:\Tr(w\theta^T)=0\}$ using Cartesian and Kronecker products.

\subsection{Cartesian Product}\label{sec:cartesian}

We start with the definition of the Cartesian product of two graphs.
\begin{definition}\label{def:cartesian}
	Given two graphs $G_1=(V_1,E_1)$ and $G_2=(V_2,E_2)$, their Cartesian product $G=G_1\,\square\, G_2$ is defined with the vertex set $V_1\times V_2$. Its edge set contains $((x_1,x_2),(y_1,y_2))$ if and only if $x_1=y_1$ and $(x_2,y_2)\in E_2$ or $(x_1,y_1suno)\in E_1$ and $x_2=y_2$.
\end{definition}

According to the definition, it can be checked that for two graphs of sizes $p_1$ and $p_2$, the adjacency matrix, the Laplacian and the incidence matrix of the Cartesian product enjoy the relations
\begin{eqnarray*}
	A_{1\,\square\, 2} &=& A_2\otimes I_{p_1} + I_{p_2}\otimes A_1,\\
	L_{1\,\square\, 2} &=& L_2\otimes I_{p_1} + I_{p_2}\otimes L_1,\\
	D_{1\,\square\, 2} &=& [D_2\otimes I_{p_1};I_{p_2}\otimes D_1].
\end{eqnarray*}

Given graphs $\gamma_1$ and $\gamma_2$ that encode row and column structures of $\theta$, we introduce the following prior distribution
\begin{eqnarray}
\label{eq:cartesian-laplacian} p(\theta\,|\,\gamma_1,\gamma_2,\sigma^2) &\propto& \prod_{(i,j)\in E_1}\exp\left(-\frac{\|\theta_{i*}-\theta_{j*}\|^2}{2\sigma^2[v_0\gamma_{1,ij}+v_1(1-\gamma_{1,ij})]}\right) \\
\nonumber && \times \prod_{(k,l)\in E_2}\exp\left(-\frac{\|\theta_{*k}-\theta_{*l}\|^2}{2\sigma^2[v_0\gamma_{2,kl} +v_1(1-\gamma_{2,kl})]}\right) \mathbb{I}\{\theta\in\Theta_w\}.
\end{eqnarray}
Here, $E_1$ and $E_2$ are the base graphs of the row and column structures.
According to its form, the prior distribution (\ref{eq:cartesian-laplacian}) models both pairwise relations of rows and those of columns based on $\gamma_1$ and $\gamma_2$, respectively. To better understand (\ref{eq:cartesian-laplacian}), we can write it in the following equivalent form,
\begin{equation}
p(\theta\,|\,\gamma_1,\gamma_2,\sigma^2)\propto \exp\left(-\frac{1}{2\sigma^2}{\sf vec}(\theta)^T\left(L_{\gamma_2}\otimes I_{p_1} + I_{p_2}\otimes L_{\gamma_1}\right){\sf vec}(\theta)\right)\mathbb{I}\{\theta\in\Theta_w\},\label{eq:cartesian-laplacian-vec}
\end{equation}
where $L_{\gamma_1}\in\mathbb{R}^{p_1\times p_1}$ and $L_{\gamma_2}\in\mathbb{R}^{p_2\times p_2}$ are Laplacian matrices of the weighted graphs $\{v_0\gamma_{1,ij}+v_1(1-\gamma_{1,ij})\}$ and $\{v_0\gamma_{2,kl} +v_1(1-\gamma_{2,kl})\}$, respectively. Therefore, by Definition \ref{def:cartesian}, $p(\theta\,|\,\gamma_1,\gamma_2,\sigma^2)$ is a spike-and-slab Laplacian prior $p(\theta\,|\,\gamma,\sigma^2)$ defined in (\ref{eq:Laplacian-prior}) with $\gamma=\gamma_1\,\square\,\gamma_2$, and the well-definedness is guaranteed by Proposition \ref{prop:Laplacian}.

To complete the Bayesian model, the distribution of $(\gamma_1,\gamma_2,\sigma^2)$ are specified by
\begin{eqnarray}
\label{eq:cart-gamma}\gamma_1,\gamma_2\,|\,\eta_1,\eta_2 &\sim& \prod_{(i,j)\in E_1}\eta_1^{\gamma_{1,ij}}(1-\eta_1)^{1-\gamma_{1,ij}}\prod_{(i,j)\in E_2}\eta_2^{\gamma_{2,kl}}(1-\eta_2)^{1-\gamma_{2,kl}}, \\
\label{eq:cart-eta}\eta_1,\eta_2 &\sim& \text{Beta}(A_1,B_1)\bigotimes \text{Beta}(A_2,B_2), \\
\label{eq:cart-sigma}\sigma^2 &\sim& \text{InvGamma}(a/2,b/2).
\end{eqnarray}
We remark that it is possible to constrain $\gamma_1$ and $\gamma_2$ in some subsets $\Gamma_1$ and $\Gamma_2$ like (\ref{eq:prior-gamma}). This extra twist is useful for a biclustering model that will be discussed in Section \ref{sec:biclust}.

Note that in general the base graph $G=G_1\,\square\, G_2$ is not a tree, and the derivation of the EM algorithm follows a similar argument in Section \ref{sec:alg-non}. Using the same argument in (\ref{eq:lower}), we lower bound $\log\sum_{T\in\text{spt}(G)}\sum_{e\in T}[v_0^{-1}\gamma_e+v_1^{-1}(1-\gamma_e)]$ by
\begin{equation}
\sum_{e\in E_1\,\square\, E_2}r_e\log[v_0^{-1}\gamma_e+v_1^{-1}(1-\gamma_e)]+\log|\text{spt}(G)|. \label{eq:lower-cartesion}
\end{equation}
Since
$$E_1\,\square\, E_2=\left\{((i,k),(j,k)):(i,j)\in E_1, k\in V_2\right\}\cup\left\{((i,k),(i,l)): i\in V_1,(k,l)\in E_2\right\},$$
and $\gamma=\gamma_1\,\square\, \gamma_2$,
we can write (\ref{eq:lower-cartesion}) as
\begin{eqnarray*}
	&& \sum_{(i,j)\in E_1}\sum_{k=1}^{p_2}r_{(i,k),(j,k)}\log[v_0^{-1}\gamma_{1,ij}+v_1^{-1}(1-\gamma_{1,ij})] \\
	&& + \sum_{(k,l)\in E_2}\sum_{i=1}^{p_1}r_{(i,k),(i,l)}\log[v_0^{-1}\gamma_{2,kl}+v_1^{-1}(1-\gamma_{2,kl})] \\
	&=& \sum_{(i,j)\in E_1}r_{1,ij}\log[v_0^{-1}\gamma_{1,ij}+v_1^{-1}(1-\gamma_{1,ij})]  + \sum_{(k,l)\in E_2}r_{2,kl}\log[v_0^{-1}\gamma_{2,kl}+v_1^{-1}(1-\gamma_{2,kl})],
\end{eqnarray*}
where
$$r_{1,ij}=\sum_{k=1}^{p_2}r_{(i,k),(j,k)}=\frac{1}{|\text{spt}(G)|}\sum_{k=1}^{p_2}\sum_{T\in\text{spt}(G)}\mathbb{I}\{((i,k),(j,k))\in T\},$$
and $r_{2,kl}$ is similarly defined.

Using the lower bound derived above, it is direct to derive the an EM algorithm, which consists of the following iterations,
\begin{eqnarray}
\label{eq:cartesian-E1} q_{1,ij}^{\rm new} &=& \frac{\eta_1 v_0^{-{r_{1,ij}}/{2}}e^{-\|\theta_{i*}-\theta_{j*}\|^2/2\sigma^2v_0 } }{\eta_1 v_0^{-{r_{1,ij}}/{2}}e^{-\|\theta_{i*}-\theta_{j*}\|^2 / 2\sigma^2v_0}+(1-\eta_1) v_1^{-{r_{1,ij}}/{2}}e^{-\|\theta_{i*}-\theta_{j*}\|^2 / 2\sigma^2v_1}}, \\
\label{eq:cartesian-E2} q_{2,kl}^{\rm new} &=& \frac{\eta_2 v_0^{-{r_{2,kl}}/{2}}e^{-\|\theta_{*k}-\theta_{*l}\|^2 / 2\sigma^2v_0}}{\eta_2 v_0^{-{r_{2,kl}}/{2}}e^{-\|\theta_{*k}-\theta_{*l}\|^2 / 2\sigma^2v_0}+(1-\eta_2) v_1^{-{r_{2,kl}}/{2}}e^{-\|\theta_{*k}-\theta_{*l}\|^2 / 2\sigma^2v_1}}, \\
\label{eq:cartesian-M} (\alpha^{\rm new}, \theta^{\rm new}) &=& \argmin_{\alpha,\theta\in\Theta_w}F(\alpha,\theta;q_1^{\rm new},q_2^{\rm new}), \\
\nonumber (\sigma^2)^{\rm new} &=& \frac{F(\alpha^{\rm new},\theta^{\rm new};q_1^{\rm new},q_2^{\rm new})+b}{n_1n_2+p_1p_2+a+2}, \\
\nonumber \eta_1^{\rm new} &=& \frac{A_1-1+\sum_{(i,j)\in E_1}q_{1,ij}^{\rm new}}{A_1+B_1-2+m_1}, \\
\label{eq:cartesian-M-eta2} \eta_2^{\rm new} &=& \frac{A_2-1+\sum_{(k,l)\in E_2}q_{2,ij}^{\rm new}}{A_2+B_2-2+m_2}.
\end{eqnarray}
The definition of the function $F(\alpha,\theta;q_1,q_2)$ is given by
\ba
F(\alpha,\theta;q_1,q_2)=\fnorm{y-X_1(\alpha w+\theta)X_2^T}^2+\nu\alpha^2+{\sf vec}(\theta)^T\left(L_{q_2}\otimes I_{p_1} + I_{p_2}\otimes L_{q_1}\right){\sf vec}(\theta)
\ea

Though the E-steps (\ref{eq:cartesian-E1}) and (\ref{eq:cartesian-E2}) are straightforward, the M-step (\ref{eq:cartesian-M}) is a quadratic programming of dimension $p_1p_2$, which may become the computational bottleneck of the EM algorithm when the size of the problem is large. We will introduce a Dykstra-like algorithm to solve (\ref{eq:cartesian-M}) in Appendix \ref{app:imp}.

The Cartesian product model is useful for simultaneous learning the row structure $\gamma_1$ and the column structure and $\gamma_2$ of the coefficient matrix $\theta$. Note that when $X_1 = X$, $X_2 = I_d$, and $E_2=\varnothing$, the Cartesian product model becomes the multivariate regression model described in Section~\ref{sec:multivariate}. In this case, the model only regularizes the row structure of $\theta$. Another equally interesting example is obtained when $X_1 = X$, $X_2 = I_d$, and $E_1=\varnothing$. In this case, the model only regularizes the column structure of $\theta$, and can be interpreted as multitask learning with task clustering.

\subsection{Kronecker Product}\label{sec:kronecker}

The Kronecker product of two graphs is defined below.
\begin{definition}\label{def:kronecker}
	Given two graphs $G_1=(V_1,E_1)$ and $G_2=(V_2,E_2)$, their Kronecker product $G=G_1\otimes G_2$ is defined with the vertex set $V_1\times V_2$. Its edge set contains $((x_1,x_2),(y_1,y_2))$ if and only if $(x_1,y_1)\in E_1$ and $(x_2,y_2)\in E_2$.
\end{definition}

It is not hard to see that the adjacency matrix of two graphs has the formula $A_{1\otimes 2}=A_1\otimes A_2$, which gives the name of Definition \ref{def:kronecker}. The prior distribution of $\theta$ given row and column graphs $\gamma_1$ and $\gamma_2$ that we discuss in this subsection is
\begin{equation}
p(\theta\,|\,\gamma_1,\gamma_2,\sigma^2)\propto \prod_{(i,j)\in E_1}\prod_{(k,l)\in E_2}\exp\left(-\frac{(\theta_{ik}-\theta_{jl})^2}{2\sigma^2[v_0\gamma_{1,ij}\gamma_{2,kl}+v_1(1-\gamma_{1,ij}\gamma_{2,kl})]}\right)\mathbb{I}\{\theta\in\Theta_w\}.\label{eq:kronecker-laplacian}
\end{equation}
Again, $E_1$ and $E_2$ are the base graphs of the row and column structures.
According to its form, the prior imposes a nearly block structure on $\theta$ based on the graphs $\gamma_1$ and $\gamma_2$. Moreover, $p(\theta\,|\,\gamma_1,\gamma_2,\sigma^2)$ can be viewed as a spike-and-slab Laplacian prior $p(\theta\,|\,\gamma,\sigma^2)$ defined in (\ref{eq:Laplacian-prior}) with $\gamma=\gamma_1\otimes\gamma_2$. The distribution of $(\gamma_1,\gamma_2,\sigma^2)$ follows the same specification in (\ref{eq:cart-gamma})-(\ref{eq:cart-sigma}). 

To derive an EM algorithm, we follow the strategy in Section \ref{sec:alg-non} and lower bound $\log\sum_{T\in\text{spt}(G)}\sum_{e\in T}[v_0^{-1}\gamma_e+v_1^{-1}(1-\gamma_e)]$ by
\begin{equation}
\sum_{(i,j)\in E_1}\sum_{(k,l)\in E_2}r_{(i,k),(j,l)}\log[v_0^{-1}\gamma_{1,ij}\gamma_{2,kl}+v_1^{-1}(1-\gamma_{1,ij}\gamma_{2,kl})].\label{eq:lower-kronecker}
\end{equation}
Unlike the Cartesian product, the Kronecker product structure has a lower bound (\ref{eq:lower-kronecker}) that is not separable with respect to $\gamma_1$ and $\gamma_2$. This makes the E-step combinatorial, and does not apply to a large-scale problem. To alleviate this computational barrier, we consider a variational EM algorithm that finds the best posterior distribution of $\gamma_1,\gamma_2$ that can be factorized. In other words, instead of maximizing over all possible distribution $q$, we  maximize over the mean-filed class $q\in\mathcal{Q}$, with $\mathcal{Q}=\{q(\gamma_1,\gamma_2)=q_1(\gamma_1)q_2(\gamma_2):q_1,q_2\}$. Then, the objective becomes
\ba
\max_{q_1,q_2}\max_{\alpha,\theta\in\Theta_2,\delta,\eta,\sigma^2}\sum_{\gamma_1,\gamma_2}q_1(\gamma_1)q_2(\gamma_2)\log\frac{\wt{p}(y,\alpha,\theta,\delta,\gamma_1,\gamma_2,\eta,\sigma^2)}{q_1(\gamma_1)q_2(\gamma_2)},
\ea
where $\wt{p}(y,\alpha,\theta,\delta,\gamma_1,\gamma_2,\eta,\sigma^2)$ is obtained by replacing $p(\theta\,|\,\gamma_1,\gamma_2,\sigma^2)$ with $\wt{p}(\theta\,|\,\gamma_1,\gamma_2,\sigma^2)$ in the joint distribution ${p}(y,\alpha,\theta,\delta,\gamma_1,\gamma_2,\eta,\sigma^2)$. Here, $\log\wt{p}(\theta\,|\,\gamma_1,\gamma_2,\sigma^2)$ is a lower bound for $\log{p}(\theta\,|\,\gamma_1,\gamma_2,\sigma^2)$ with (\ref{eq:lower-kronecker}).
The E-step of the variational EM is
\ba
q_1^{\rm new}(\gamma_1) {}& \propto \exp\left(\sum_{\gamma_2}q_2(\gamma_2)\log \wt{p}(y,\alpha,\theta,\delta,\gamma_1,\gamma_2,\eta,\sigma^2)\right), \\
q_2^{\rm new}(\gamma_2) {}& \propto \exp\left(\sum_{\gamma_1}q_1^{\rm new}(\gamma_1)\log \wt{p}(y,\alpha,\theta,\delta,\gamma_1,\gamma_2,\eta,\sigma^2)\right).
\ea
After some simplification, we have
\begin{align}
\label{eq:kronecker-E1} q_{1,ij}^{\rm new} &=& \left( {1+\frac{(1-\eta_1)\prod_{(k,l)\in E_2}\left(v_1^{-{r_{(i,k),(j,l)}}/{2}}e^{-(\theta_{ik}-\theta_{jl})^2 / 2\sigma^2 v_1}\right)^{q_{2,kl}}}{\eta_1\prod_{(k,l)\in E_2}\left(v_0^{-{r_{(i,k),(j,l)}}/{2}}e^{-(\theta_{ik}-\theta_{jl})^2 / 2\sigma^2 v_0}\right)^{q_{2,kl}}}} \right)^{-1}, \\
\label{eq:kronecker-E2} q_{2,kl}^{\rm new} &=& \left( {1+\frac{(1-\eta_2)\prod_{(i,j)\in E_1}\left(v_1^{-{r_{(i,k),(j,l)}}/{2}}e^{-(\theta_{ik}-\theta_{jl})^2 / 2\sigma^2 v_1}\right)^{q^{\rm new}_{1,kl}}}{\eta_2\prod_{(i,j)\in E_1}\left(v_0^{-{r_{(i,k),(j,l)}}/{2}}e^{-(\theta_{ik}-\theta_{jl})^2 / 2\sigma^2 v_0}\right)^{q^{\rm new}_{1,kl}}}} \right)^{-1}.
\end{align}

The M-step can be derived in a standard way, and it has the same updates as in (\ref{eq:cartesian-M})-(\ref{eq:cartesian-M-eta2}), with a new definition of $F(\alpha,\theta;q_1,q_2)$ given by
\ba
F(\alpha,\theta;q_1,q_2) = {}&  \|y-X_1(\alpha w+\theta)X_2\|^2+\nu\alpha^2 \\
{}&  +\sum_{(i,j)\in E_1}\sum_{(k,l)\in E_2}\left(\frac{q_{1,ij}q_{2,kl}}{v_0}+\frac{1-q_{1,ij}q_{2,kl}}{v_1}\right)(\theta_{ik}-\theta_{jl})^2.
\ea

\subsection{Applications in Biclustering} \label{sec:biclust}

When both row and column graphs encode clustering structures discussed in Section \ref{sec:clustering-model}, we have the biclustering model. In this section, we discuss both biclustering models induced by Kronecker and Cartesian products. We start with a special form of the likelihood (\ref{eq:likelihood-product2}), which is given by
\ba
y\,|\,\alpha,\theta,\sigma^2\sim N(\alpha \mathds{1}_{n_1}\mathds{1}_{n_2}^T+ \theta, \sigma^2 I_{n_1}\otimes I_{n_2}),
\ea
and the prior distribution on $\alpha$ is given by (\ref{eq:alpha-product}). The prior distribution on $\theta$ will be discussed in two cases.

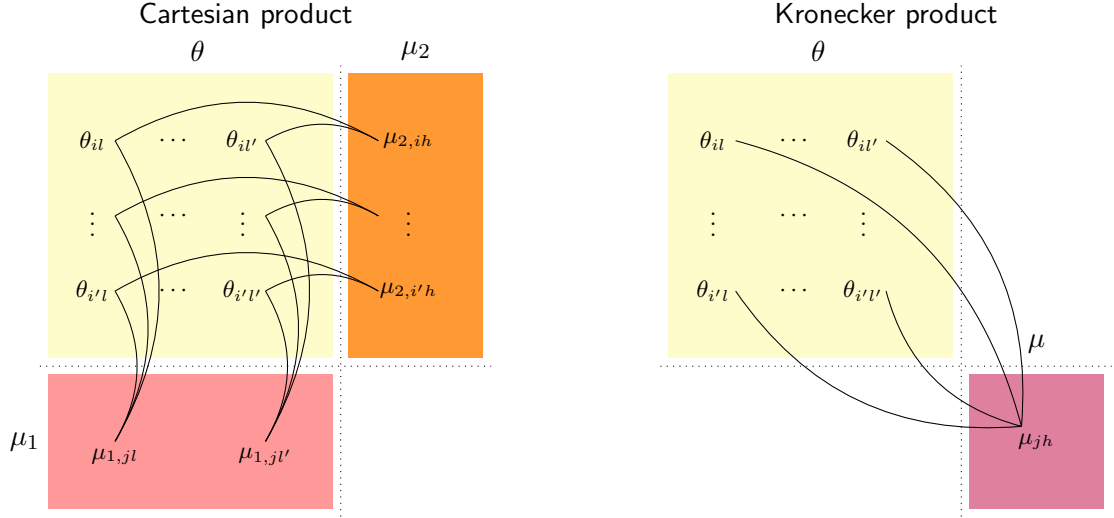
\begin{figure}[!h]
	\centering
	\centerline{
		\begin{minipage}{0.55\textwidth}
			\centering
			{\sf Cartesian product}
			\begin{tikzpicture}[every node/.style={minimum size=0.1cm-\pgflinewidth, outer sep=0pt}]
			\draw[dotted, step=10cm,color=black] (-4,4) grid (2,-2);
			\node at (-1.9,4.2) {$\bf \theta$};
			\node at (1.0,4.2) {$\mu_2$};
			\node at (-4.2,-1.0) {$\mu_1$};
			\node[fill=yellow, opacity = 0.2, text opacity = 1, minimum size = 3.8cm-\pgflinewidth] at (-2,2) {};
			\node[fill=orange,  opacity = 0.8, text opacity = 1, minimum size = 3.8cm-\pgflinewidth, minimum width = 1.8cm-\pgflinewidth] at (1,2) {};
			\node[fill=red,  opacity = 0.4, text opacity = 1, minimum size = 3.8cm-\pgflinewidth, minimum height = 1.8cm-\pgflinewidth] at (-2,-1) {};
			\node at (-3.3, 3) {\footnotesize $\theta_{il}$}; \node at (-2.2, 3) {\footnotesize $\cdots$}; 
			\node at (-1.3, 3) {\footnotesize $\theta_{il'}$}; \node at (0.9, 3) {\footnotesize $\mu_{2,ih}$}; 
			\node at (-3.3, 2) {\footnotesize $\vdots$}; \node at (-2.2, 2) {\footnotesize $\cdots$}; 
			\node at (-1.3, 2) {\footnotesize $\vdots$}; \node at (0.9, 2) {\footnotesize $\vdots$}; 
			\node at (-3.3, 1) {\footnotesize $\theta_{i'l}$}; \node at (-2.2, 1) {\footnotesize $\cdots$}; 
			\node at (-1.3, 1) {\footnotesize $\theta_{i'l'}$}; \node at (0.9, 1) {\footnotesize $\mu_{2,i'h}$};
			\node at (-3, -1.2) {\footnotesize $\mu_{1,jl}$}; \node at (-1, -1.2) {\footnotesize $\mu_{1,jl'}$};
			\path[bend left]  (-3, 3) edge (1/2,3); \path[bend left]  (-1, 3) edge (1/2,3);
			\path[bend left]  (-3, 2) edge (1/2,2); \path[bend left]  (-1, 2) edge (1/2,2);
			\path[bend left]  (-3, 1) edge (1/2,1); \path[bend left]  (-1, 1) edge (1/2,1);
			\path[bend left]  (-3, 3) edge (-3,-1); \path[bend left]  (-3, 1) edge (-3,-1);
			\path[bend left]  (-1, 3) edge (-1,-1); \path[bend left]  (-1, 1) edge (-1,-1);
			\path[bend left]  (-3, 2) edge (-3,-1); \path[bend left]  (-1, 2) edge (-1,-1);
			\end{tikzpicture}
		\end{minipage}
		\begin{minipage}{0.55\textwidth}
			\centering
			{\sf Kronecker product}
			\begin{tikzpicture}[every node/.style={minimum size=0.1cm-\pgflinewidth, outer sep=0pt}]
			\draw[dotted, step=10cm,color=black] (-4,4) grid (2,-2);
			\node at (-1.9,4.2) {$\theta$};
			\node at (1,0.3) {$\mu$};
			\node[fill=yellow, opacity = 0.2, text opacity = 1, minimum size = 3.8cm-\pgflinewidth] at (-2,2) {};
			\node[fill=purple,  opacity = 0.5, text opacity = 1, minimum size = 1.8cm-\pgflinewidth] at (1,-1) {};
			\node at (-3.3, 3) {\footnotesize $\theta_{il}$}; \node at (-1.3, 3) {\footnotesize $\theta_{il'}$};
			\node at (-3.3, 1) {\footnotesize $\theta_{i'l}$}; \node at (-1.3, 1) {\footnotesize $\theta_{i'l'}$};
			\node at (-2.2, 3) {\footnotesize $\cdots$}; \node at (-2.2, 2) {\footnotesize $\cdots$}; 
			\node at (-3.3, 2) {\footnotesize $\vdots$};\node at (-1.3, 2) {\footnotesize $\vdots$};
			\node at (1, -1) {\footnotesize $\mu_{jh}$}; \node at (-2.2, 1) {\footnotesize $\cdots$};
			\path[bend left]  (-3, 3) edge (0.8,-0.8); \path[bend left]  (-1, 3) edge (0.8,-0.8);
			\path[bend right]  (-3, 1) edge (0.8,-0.8); \path[bend right]  (-1,1) edge (0.8,-0.8);
			\end{tikzpicture}
		\end{minipage}
	}
	\caption{Structure diagrams for the two biclustering methods. The Cartesian product biclustering model \textrm{(Left)} and the Kronecker product biclustering model \textrm{(Right)} have different latent variables and base graphs. While the Cartesian product models the row and column clustering structures by separate latent variable matrices $\mu_1\in\mathbb{R}^{k_1\times n_2}$ and $\mu_2\in\mathbb{R}^{n_1\times k_2}$, the Kronecker product directly models the checkerboard structure by a single latent matrix $\mu\in\mathbb{R}^{k_1\times k_2}$.}
	\label{fig:cart_vs_kron}
\end{figure}

\subsubsection{Cartesian product biclustering model} Let $k_1\in[n_1]$ and $k_2\in[n_2]$ be upper bounds of the numbers of row and column clusters, respectively. We introduce two latent matrices $\mu_1\in\mathbb{R}^{k_1\times n_2}$ and $\mu_2\in\mathbb{R}^{n_1\times k_2}$ that serve as row and column clustering centers. The prior distribution is then specified by
\ba
p(\theta,\mu_1,\mu_2\,|\,\gamma_1,\gamma_2,\sigma^2) \propto {}& \prod_{i=1}^{n_1}\prod_{j=1}^{k_1}\exp\left(-\frac{\|\theta_{i*}-\mu_{1,j*}\|^2}{2\sigma^2[v_0\gamma_{1,ij}+v_1(1-\gamma_{1,ij})]}\right) \\
{}& \times \prod_{l=1}^{n_2}\prod_{h=1}^{k_2}\exp\left(-\frac{\|\theta_{*l}-\mu_{2,*h}\|^2}{2\sigma^2[v_0\gamma_{2,lh}+v_1(1-\gamma_{2,lh})]}\right)\mathbb{I}\{\mathds{1}_{n_1}^T\theta\mathds{1}_{n_2}=0\},
\ea
which can be regarded as an extension of (\ref{eq:spike-and-slab-clustering}) in the form of (\ref{eq:cartesian-laplacian}). The prior distributions on $\gamma_1$ and $\gamma_2$ are independently specified by (\ref{eq:clustering-prior-gamma-uniform}) with $(k,n)$ replaced by $(k_1,n_1)$ and $(k_2,n_2)$. Finally, $\sigma^2$ follows the inverse Gamma prior (\ref{eq:prior-sigma}).

We follow the framework of Section \ref{sec:alg-non}. The derivation of the EM algorithm requires lower bounding $\log\sum_{T\in\text{spt}(G)}\sum_{e\in T}[v_0^{-1}\gamma_e+v_1^{-1}(1-\gamma_e)]$. Using the same argument in Section \ref{sec:cartesian}, we have the following lower bound
\baa
\sum_{i=1}^{n_1}\sum_{j=1}^{k_1}r_{1,ij}\log[v_0^{-1}\gamma_{1,ij}+v_1^{-1}(1-\gamma_{1,ij})]  + \sum_{l=1}^{n_2}\sum_{h=1}^{k_2}r_{2,lh}\log[v_0^{-1}\gamma_{2,lh}+v_1^{-1}(1-\gamma_{2,lh})].\label{eq:lower-cartesian-biclustering}
\eaa
By the symmetry of the complete bipartite graph, $r_{1,ij}$ is a constant that does not depend on $(i,j)$. Then use the same argument in (\ref{eq:audi-r8})-(\ref{eq:amg-gt}), and we obtain the fact that $\sum_{i=1}^{n_1}\sum_{j=1}^{k_1}r_{1,ij}\log[v_0^{-1}\gamma_{1,ij}+v_1^{-1}(1-\gamma_{1,ij})]$ is independent of $\{\gamma_{1,ij}\}$, and the same conclusion also applies to the second term of (\ref{eq:lower-cartesian-biclustering}).

Since the lower bound (\ref{eq:lower-cartesian-biclustering}) does not dependent on $\gamma_1,\gamma_2$, the determinant factor in the density function $p(\theta,\mu_1,\mu_2\,|\,\gamma_1,\gamma_2,\sigma^2)$ does not play any role in the derivation of the EM algorithm. With some standard calculations, the E-step is given by
\ba
q_{1,ij}^{\rm new} {}&= \frac{\exp\left(-\frac{\|\theta_{i*}-\mu_{1,j*}\|^2}{2\sigma^2\bar{v}}\right)}{\sum_{u=1}^{k_1}\exp\left(-\frac{\|\theta_{i*}-\mu_{1,u*}\|^2}{2\sigma^2\bar{v}}\right)}, \quad q_{1,lh}^{\rm new} = \frac{\exp\left(-\frac{\|\theta_{*l}-\mu_{2,*h}\|^2}{2\sigma^2\bar{v}}\right)}{\sum_{v=1}^{k_2}\exp\left(-\frac{\|\theta_{*l}-\mu_{2,*v}\|^2}{2\sigma^2\bar{v}}\right)},
\ea
where $\bar{v}^{-1}=v_0^{-1}-v_1^{-1}$. The M-step is given by
\ba
(\alpha^{\rm new}, \theta^{\rm new},\mu_1^{\rm new},\mu_2^{\rm new}) {}&= \argmin_{\alpha, \mathds{1}_{n_1}^T\theta\mathds{1}_{n_2}=0,\mu_1,\mu_2}F(\alpha,\theta,\mu_1,\mu_2;q_1^{\rm new},q_2^{\rm new}), \\
\nonumber (\sigma^2)^{\rm new}{}&= \frac{F(\alpha^{\rm new},\theta^{\rm new},\mu_1^{\rm new},\mu_2^{\rm new};q_1^{\rm new},q_2^{\rm new})+b}{2n_1n_2+n_1k_2+n_2k_1+a+2},
\ea
where
\ba
F(\alpha,\theta,\mu_1,\mu_2;q_1,q_2) ={}& \fnorm{y-\alpha\mathds{1}_{n_1}\mathds{1}_{n_2}^T-\theta}^2+\nu\|\alpha\|^2 \\
{}& +\sum_{i=1}^{n_1}\sum_{j=1}^{k_1}\left(\frac{q_{1,ij}}{v_0}+\frac{1-q_{1,ij}}{v_1}\right)\|\theta_{i*}-\mu_{1,j*}\|^2 \\
{}& +\sum_{l=1}^{n_2}\sum_{h=1}^{k_2}\left(\frac{q_{2,lh}}{v_0}+\frac{1-q_{2,lh}}{v_1}\right)\|\theta_{*l}-\mu_{2,*h}\|^2.
\ea

\subsubsection{Kronecker product biclustering model} For the Kronecker product structure, we introduce a latent matrix $\mu\in\mathbb{R}^{k_1\times k_2}$. Since the biclustering model implies a block-wise constant structure for $\theta$. Each entry of $\mu$ serves as a center for a block of the matrix $\theta$. The prior distribution is defined by
\ba
{}& p(\theta,\mu\,|\,\gamma_1,\gamma_2,\sigma^2) \\
\propto {}& \prod_{i=1}^{n_1}\prod_{j=1}^{k_1}\prod_{l=1}^{n_2}\prod_{h=1}^{k_2}\exp\left(-\frac{(\theta_{il}-\mu_{jh})^2}{2\sigma^2[v_0\gamma_{1,ij}\gamma_{2,lh}+v_1(1-\gamma_{1,ij}\gamma_{2,lh})]}\right)\mathbb{I}\{\mathds{1}_{n_1}^T\theta\mathds{1}_{n_2}=0\}.
\ea
The prior distribution is another extension of (\ref{eq:spike-and-slab-clustering}), and it is in a similar form of (\ref{eq:kronecker-laplacian}). To finish the Bayesian model specification, we consider the same priors for $\gamma_1,\gamma_2,\sigma^2$ as in the Cartesian product case.

Recall that the lower bound of $\log\sum_{T\in\text{spt}(G)}\sum_{e\in T}[v_0^{-1}\gamma_e+v_1^{-1}(1-\gamma_e)]$ is given by (\ref{eq:lower-kronecker}) for a general Kronecker product structure. In the current setting, a similar argument gives the lower bound 
\ba
\sum_{i=1}^{n_1}\sum_{j=1}^{k_1}\sum_{l=1}^{n_2}\sum_{h=1}^{k_2}r_{(i,l),(j,h)}\log[v_0^{-1}\gamma_{1,ij}\gamma_{2,lh}+v_1^{-1}(1-\gamma_{1,ij}\gamma_{2,lh})].
\ea
Since $r_{(i,l),(j,h)} \equiv r$ is independent of $(i,l),(j,h)$ by the symmetry of the complete bipartite graph, the above lower bound can be written as
\ba
\nonumber {}& r\sum_{i=1}^{n_1}\sum_{j=1}^{k_1}\sum_{l=1}^{n_2}\sum_{h=1}^{k_2}\log[v_0^{-1}\gamma_{1,ij}\gamma_{2,lh}+v_1^{-1}(1-\gamma_{1,ij}\gamma_{2,lh})] \\
\nonumber {}&= r\log(v_0^{-1})\sum_{i=1}^{n_1}\sum_{j=1}^{k_1}\sum_{l=1}^{n_2}\sum_{h=1}^{k_2}\gamma_{1,ij}\gamma_{2,lh} + r\log(v_1^{-1})\sum_{i=1}^{n_1}\sum_{j=1}^{k_1}\sum_{l=1}^{n_2}\sum_{h=1}^{k_2}(1-\gamma_{1,ij}\gamma_{2,lh}) \\
\label{eq:biclustering-constant} {}&= rn_1n_2\log(v_0^{-1}) + rn_1n_2(k_1k_2-1)\log(v_1^{-1}),
\ea
which is independent of $\gamma_1,\gamma_2$. The inequality (\ref{eq:biclustering-constant}) is because both $\gamma_1$ and $\gamma_2$ satisfy (\ref{eq:gamma-clustering-condition}).

Again, the determinant factor in the density function $p(\theta,\mu\,|\,\gamma_1,\gamma_2,\sigma^2)$ does not play any role in the derivation of the EM algorithm, because the lower bound (\ref{eq:biclustering-constant}) does not depend on $(\gamma_1,\gamma_2)$. Since we are working with the Kronecker product, we will derive a variational EM algorithm with the E-step finding the posterior distribution in the mean filed class $\mathcal{Q}=\{q(\gamma_1,\gamma_2)=q_1(\gamma_1)q_2(\gamma_2):q_1,q_2\}$. By following the same argument in Section \ref{sec:kronecker}, we obtain the E-step as
\ba
q_{1,ij}^{\rm new} {}&= \frac{\exp\left(-\sum_{l=1}^{n_2}\sum_{h=1}^{k_2}\frac{q_{2,lh}(\theta_{il}-\mu_{jh})^2}{2\sigma^2\bar{v}}\right)}{\sum_{u=1}^{k_1}\exp\left(-\sum_{l=1}^{n_2}\sum_{h=1}^{k_2}\frac{q_{2,lh}(\theta_{il}-\mu_{uh})^2}{2\sigma^2\bar{v}}\right)}, \\
q_{2,lh}^{\rm new} {}&= \frac{\exp\left(-\sum_{i=1}^{n_1}\sum_{j=1}^{k_1}\frac{q_{1,ij}^{\rm new}(\theta_{il}-\mu_{jh})^2}{2\sigma^2\bar{v}}\right)}{\sum_{v=1}^{k_2}\exp\left(-\sum_{i=1}^{n_1}\sum_{v=1}^{k_1}\frac{q_{1,ij}^{\rm new}(\theta_{il}-\mu_{lv})^2}{2\sigma^2\bar{v}}\right)}
\ea
where $\bar{v}^{-1}=v_0^{-1}-v_1^{-1}$. The M-step is given by
\ba
(\alpha^{\rm new}, \theta^{\rm new},\mu^{\rm new}) {}&= \argmin_{\alpha, \mathds{1}_{n_1}^T\theta\mathds{1}_{n_2}=0,\mu}F(\alpha,\theta,\mu;q_1^{\rm new},q_2^{\rm new}), \\
\nonumber (\sigma^2)^{\rm new} {}&= \frac{F(\alpha^{\rm new},\theta^{\rm new},\mu^{\rm new};q_1^{\rm new},q_2^{\rm new})+b}{2n_1n_2+n_1k_2+n_2k_1+a+2},
\ea
where
\ba
F(\alpha,\theta,\mu_1,\mu_2;q_1,q_2) = {}& \fnorm{y-\alpha\mathds{1}_{n_1}\mathds{1}_{n_2}^T-\theta}^2+\nu\|\alpha\|^2 \\
{}& +\sum_{i=1}^{n_1}\sum_{j=1}^{k_1}\sum_{l=1}^{n_2}\sum_{h=1}^{k_2}\left(\frac{q_{1,ij}q_{2,lh}}{v_0}+\frac{1-q_{1,ij}q_{2,lh}}{v_1}\right)(\theta_{il}-\mu_{jh})^2.
\ea

\section{Reduced Isotonic Regression} \label{sec:iso}

The models that we have discussed so far in our general framework all involve Gaussian likelihood functions and Gaussian priors. It is important to develop a natural extension of the framework to include non-Gaussian models. In this section, we discuss a reduced isotonic regression problem with a non-Gaussian prior distribution, while a full extension to non-Gaussian models will be considered as a future project.

Given a vector of observation $y\in\mathbb{R}^n$, the reduced isotonic regression seeks the best piecewise constant fit that is nondecreasing \citep{schell1997reduced,gao2017minimax}. It is an important model that has applications in problems with natural monotone constraint on the signal. With the likelihood $y|\alpha,\theta,\sigma^2\sim N(\alpha\mathds{1}_n+\theta,\sigma^2 I_n)$, we need to specify a prior distribution on $\theta$ that induces both piecewise constant and isotonic structures. We propose the following prior distribution,
\baa
\theta\,|\,\gamma,\sigma^2 \sim p(\theta\,|\,\gamma,\sigma^2)\propto \prod_{i=1}^{n-1}\exp\left(-\frac{(\theta_{i+1}-\theta_i)^2}{2\sigma^2[v_0\gamma_i+v_1(1-\gamma_i)]}\right)\mathbb{I}\{\theta_i\leq\theta_{i+1}\}\mathbb{I}\{\mathds{1}_n^T\theta=0\}. \label{eq:Laplacian-half}
\eaa
We call (\ref{eq:Laplacian-half}) the spike-and-slab half-Gaussian distribution. Note that the support of the distribution is the intersection of the cone $\{\theta:\theta_1\leq\theta_2\leq...\leq\theta_n\}$ and the subspace $\{\theta:\mathds{1}_n^T\theta=0\}$. The parameters $v_0$ and $v_1$ play similar roles as in (\ref{eq:Laplacian-prior}), which model the closedness between $\theta_i$ and $\theta_{i+1}$ depending on the value of $\gamma_i$.
\begin{proposition}\label{prop:half-gaussian}
	For any $\gamma\in\{0,1\}^{n-1}$ and $v_0,v_1\in(0,\infty)$, the spike-and-slab half-Gaussian prior (\ref{eq:Laplacian-half}) is well defined on $\{\theta:\theta_1\leq\theta_2\leq...\leq\theta_n\}\cap \{\theta:\mathds{1}_n^T\theta=0\}$, and its density function with respect to the Lebesgue measure restricted on the support is given by
	\ba
	p(\theta\,|\,\gamma,\sigma^2) ={}& 2^{n-1}\frac{1}{(2\pi\sigma^2)^{(n-1)/2}}\sqrt{n\prod_{i=1}^{n-1}[v_0^{-1}\gamma_i+v_1^{-1}(1-\gamma_i)]} \\
	{}& \times \exp\left(-\sum_{i=1}^{n-1}\frac{(\theta_{i+1}-\theta_i)^2}{2\sigma^2[v_0\gamma_i+v_1(1-\gamma_i)]}\right)\mathbb{I}\{\theta_1\leq \theta_2\leq...\leq\theta_n\}\mathbb{I}\{\mathds{1}_n^T\theta=0\}.
	\ea
\end{proposition}
Note that the only place that Proposition \ref{prop:half-gaussian} deviates from Proposition \ref{prop:Laplacian} is the extra factor $2^{n-1}$ due to the isotonic constraint $\{\theta:\theta_1\leq\theta_2\leq...\leq\theta_n\}$ and the symmetry of the density. We complete the model specification by put priors on $\alpha,\gamma,\eta,\sigma^2$ that are given by (\ref{eq:alpha-prior}), (\ref{eq:prior-gamma}), (\ref{eq:prior-eta}) and (\ref{eq:prior-sigma}).

Now we are ready to derive the EM algorithm. Since the base graph is a tree, the EM algorithm for reduced isotonic regression is exact. The E-step is given by
\ba
q^{\rm new}_i=\frac{\eta\phi(\theta_i-\theta_{i-1};0,\sigma^2 v_0)}{\eta\phi(\theta_i-\theta_{i-1};0,\sigma^2 v_0)+(1-\eta)\phi(\theta_i-\theta_{i-1};0,\sigma^2v_1)}.
\ea
The M-step is given by
\baa
(\alpha^{\rm new},\theta^{\rm new}) = \argmin_{\alpha,\theta_1\leq\theta_2\leq...\leq\theta_n, \mathds{1}_n^T\theta=0}F(\alpha,\theta;q^{\rm new}),\label{eq:isotonic-M}
\eaa
where
\ba
F(\alpha,\theta;q)=\|y-\alpha\mathds{1}_n-\theta\|^2+\nu\alpha^2+\sum_{i=1}^{n-1}\left(\frac{q_i}{v_0}+\frac{1-q_i}{v_1}\right)(\theta_i-\theta_{i-1})^2,
\ea
and the updates of $\sigma^2$ and $\eta$ are given by (\ref{eq:M4}) with $p=n$. The M-step (\ref{eq:isotonic-M}) can be solved by a very efficient optimization technique. Since $\|y-\alpha\mathds{1}_n-\theta\|^2=\|(\bar{y}-\alpha)\mathds{1}_n\|^2+\|y-\bar{y}\mathds{1}_n-\theta\|^2$ by $\mathds{1}_n^T\theta=0$, $\alpha$ and $\theta$ can be updated independently. It is easy to see that $\alpha^{\rm new}=\frac{n}{n+\nu}\bar{y}$. The update of $\theta$ can be solved by SPAVA \citep{burdakov2017dual}.

Similar to the Gaussian case, the parameter $v_0$ determines the complexity of the model. For each $v_0$ between $0$ and $v_1$, we apply the EM algorithm above to calculate $\wh{q}$, and then let $\wh{\gamma}_{i}=\wh{\gamma}_i(v_0)=\mathbb{I}\{\wh{q}_{i}\geq 1/2\}$ form a solution path. The best model will be selected from the EM-solution path by the limiting version of the posterior distribution as $v_0\rightarrow 0$.

Given a $\gamma\in\{0,1\}^{n-1}$, we write $s=1+\sum_{i=1}^{n-1}(1-\gamma_i)$ to be the number of pieces, and $Z_{\gamma}\in\{0,1\}^{n\times s}$ is the membership matrix defined in Section \ref{sec:model-selection}. As $v_0\rightarrow 0$, a slight variation of Proposition \ref{prop:reduced-model} implies that $\theta$ that follows (\ref{eq:Laplacian-half}) weakly converges to $Z_{\gamma}\wt{\theta}$, where $\wt{\theta}$ is distributed by
\begin{equation}
p(\wt{\theta}\,|\,\gamma,\sigma^2)\propto\exp\left(-\sum_{l=1}^s\frac{(\wt{\theta}_l-\wt{\theta}_{l+1})^2}{2\sigma^2v_1}\right)\mathbb{I}\{\wt{\theta}_1\leq\wt{\theta}_2\leq...\leq\wt{\theta}_{s}\}\mathbb{I}\{\mathds{1}_n^TZ_{\gamma}\wt{\theta}=0\}. \label{eq:half-limit-iso}
\end{equation}
The following proposition determines the normalizing constant of the above distribution.
\begin{proposition}\label{prop:half-gaussian-limit}
	The density function of (\ref{eq:half-limit-iso}) is given by
	\baa \label{eqn:half-gaussian-limit}
	p(\wt{\theta}\,|\,\gamma,\sigma^2) ={}& 2^{s-1}(2\pi\sigma^2)^{-(s-1)/2} \sqrt{\text{det}_{Z_{\gamma}^T \mathds{1}_n}(Z_{\gamma}^T\wt{L}_{\gamma}Z_{\gamma})} \times \\
	{}& \exp\left(-\sum_{l=1}^s\frac{(\wt{\theta}_l-\wt{\theta}_{l+1})^2}{2\sigma^2v_1}\right)\mathbb{I}\{\wt{\theta}_1\leq\wt{\theta}_2\leq...\leq\wt{\theta}_{s}\}\mathbb{I}\{\mathds{1}_n^TZ_{\gamma}\wt{\theta}=0\},
	\eaa
	where $Z_{\gamma}$ and $\wt{L}_{\gamma}$ are defined in Section \ref{sec:model-selection}.
\end{proposition}
Interestingly, compared with the formula (\ref{eq:Laplacian-prior-low-d}), (\ref{eqn:half-gaussian-limit}) has an extra $2^{s-1}$ due to the isotonic constraint $\{\wt{\theta}_1\leq...\leq\wt{\theta}_s\}$.

Following Section \ref{sec:model-selection}, we consider a reduced version of the likelihood $y\,|\,\alpha,\wt{\theta},\gamma,\sigma^2\sim N(\alpha \mathds{1}_n + Z_\gamma \wt\theta, \sigma^2 I_n)$. Then, with the prior distributions on $\alpha,\wt{\theta},\gamma,\sigma^2$ specified by (\ref{eq:alpha-prior}), (\ref{eqn:half-gaussian-limit}), (\ref{eq:prior-gamma}), (\ref{eq:prior-eta}) and (\ref{eq:prior-sigma}), we obtain the joint posterior distribution $p(\alpha,\wt{\theta},\gamma,\sigma^2\,|\,y)$. Ideally, we would like to integrate out $\alpha,\wt{\theta},\sigma^2$ and use $p(\gamma\,|\,y)$ for model selection. However, the integration with respect to $\wt{\theta}$ is intractable due to the isotonic constraint. Therefore, we propose to maximize out $\alpha,\wt{\theta},\sigma^2$, and then the model selection score for reduced isotonic regression is given by
\ba
g(\gamma)=\max_{\alpha,\wt{\theta}_1\leq...\leq\wt{\theta}_s,\mathds{1}_n^TZ_{\gamma}\wt{\theta}=0,\sigma^2}\log p(\alpha,\wt{\theta},\gamma,\sigma^2\,|\,y).
\ea
For each $\gamma$, the optimization involved in the evaluation of $g(\gamma)$ can be done efficiently, which is very similar to the M-step updates.



\section{Numerical Results}\label{sec:num}

In this section, we test the performance of the methods proposed in the paper and compare the accuracy in terms of sparse signal recovery and graphical structure estimation with existing methods. We name our method BayesMSG (Bayesian Model Selection on Graphs) throughout the section. All simulation studies and real data applications were conduced on a standard laptop (2.6 GHz Intel Core i7 processor and 16GB memory) using R and Julia programming languages. 

Our Bayesian method outputs a subgraph defined by
\ba
\wh{\gamma}=\argmax\left\{g(\gamma): \gamma\in\{\wh{\gamma}(v_0)\}_{0<v_0\leq v_1}\right\},
\ea
which is a sub-model selected by the model selection score $g(\gamma)$ on the EM solution path (see Section \ref{sec:model-selection} for details). Suppose $\gamma^*$ is the underlying true subgraph that generates the data, we measure the performance of $\wh{\gamma}$ by false discovery proportion and power. The definitions are
\ba
{\sf FDP}=\frac{\sum_{(i,j)\in E}(1-\wh{\gamma}_{ij})\gamma_{ij}^*}{\sum_{(i,j)\in E}(1-\wh{\gamma}_{ij})}\quad\text{and}\quad{\sf POW}=1-\frac{\sum_{(i,j)\in E}(1-\gamma_{ij}^*)\wh{\gamma}_{ij}}{\sum_{(i,j)\in E}(1-\gamma_{ij}^*)},
\ea
where we adopt the convention that $0/0=1$. Note that the above ${\sf FDP}$ and ${\sf POW}$ are not suitable for the clustering/biclustering model, because clustering structures are equivalent up to arbitrary clustering label permutations.

The sub-model indexed by $\wh{\gamma}$ also induces a point estimator for the model parameters. This can be done by calculating the posterior mean of the reduced model specified by the likelihood (\ref{eqn:model_selection_y}) and priors (\ref{eq:alpha-prior}) and (\ref{eq:Laplacian-prior-low-d}). With notations $p(y|\alpha,\wt{\theta},\gamma,\sigma^2)$, $p(\alpha|\sigma^2)$ and $p(\wt{\theta}|\gamma,\sigma^2)$ for (\ref{eqn:model_selection_y}), (\ref{eq:alpha-prior}) and (\ref{eq:Laplacian-prior-low-d}), the point estimator is defined by $\wh{\beta}=\alpha^{\rm est}w+ Z_{\wh{\gamma}}\wt{\theta}^{\rm est}$, where $Z_{\gamma}$ is the membership matrix defined in Section \ref{sec:model-selection}, and the definition of $(\alpha^{\rm est},\wt{\theta}^{\rm est})$ is given by
\ba
(\alpha^{\rm est},\wt{\theta}^{\rm est})=\argmax_{\alpha,\wt{\theta}\in\{\wt{\theta}:w^TZ_{\wh{\gamma}}\wt{\theta}=0\}}\log\left[p(y|\alpha,\wt{\theta},\wh{\gamma},\sigma^2)p(\alpha|\sigma^2)p(\wt{\theta}|\wh{\gamma},\sigma^2)\right],
\ea
which is a simple quadratic programming whose solution does not depend on $\sigma^2$. Note that the definition implies that $\wh{\beta}$ is the posterior mode of the reduced model. Since the posterior distribution is Gaussian, $\wh{\beta}$ is also the posterior mean. The performance of $\wh{\beta}$ will be measured by the mean squared error
\ba
{\sf MSE}=\frac{1}{n}\|X(\wh{\beta}-\beta^*)\|^2,
\ea
where $\beta^*$ is the true parameter that generates the data.

The hyper-parameters $a,b,A,B$ in (\ref{eq:prior-eta}) and (\ref{eq:prior-sigma}) are all set as the default value $1$. The same rule is also applied to the extensions in Sections \ref{sec:clustering}-\ref{sec:iso}.

\begin{figure}[!t]
	\centering
	\centerline{\includegraphics[width = 6.2in]{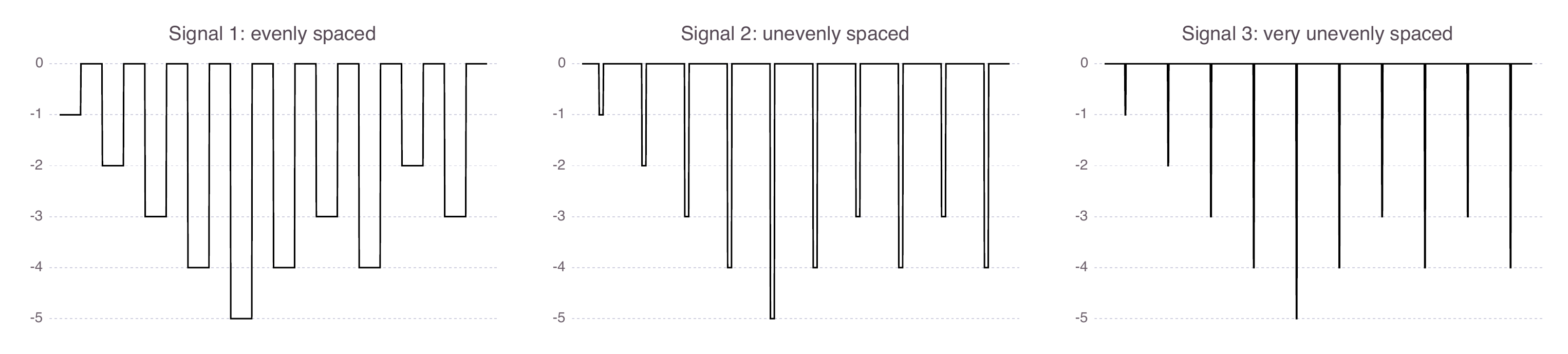}}
	\caption{Three different signals for the linear chain graph. All signals have 20 pieces. Signal 1 has evenly spaced changes (each piece has length $50$), Signal 2 has unevenly spaced changes (a smaller piece has length $10$), and Signal 3 has very unevenly spaced changes (a smaller one has length $2$).}
	\label{fig:linearpath0}
\end{figure}
\begin{figure}[!h]
	\centering
	\centerline{\includegraphics[width = 6.2in]{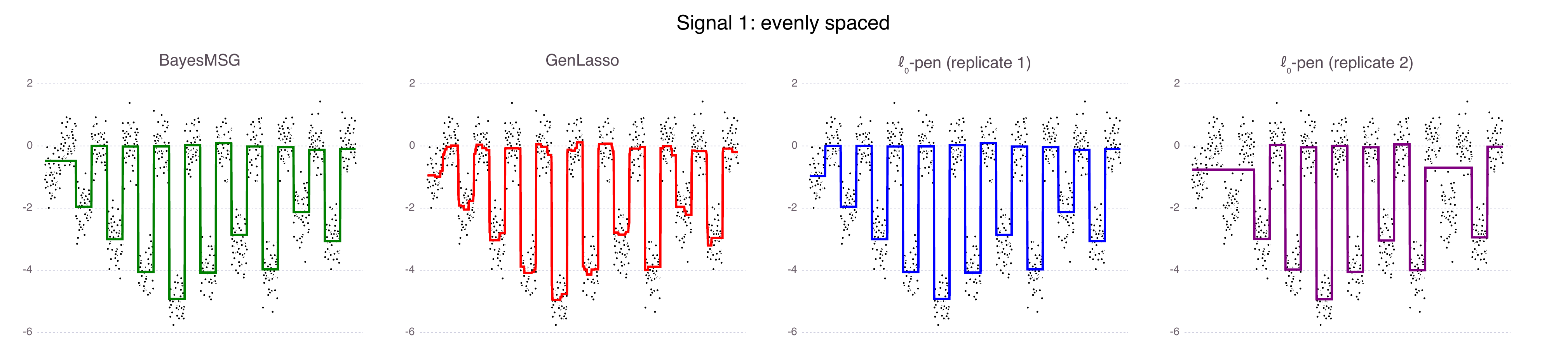}}
	\centerline{\includegraphics[width = 6.2in]{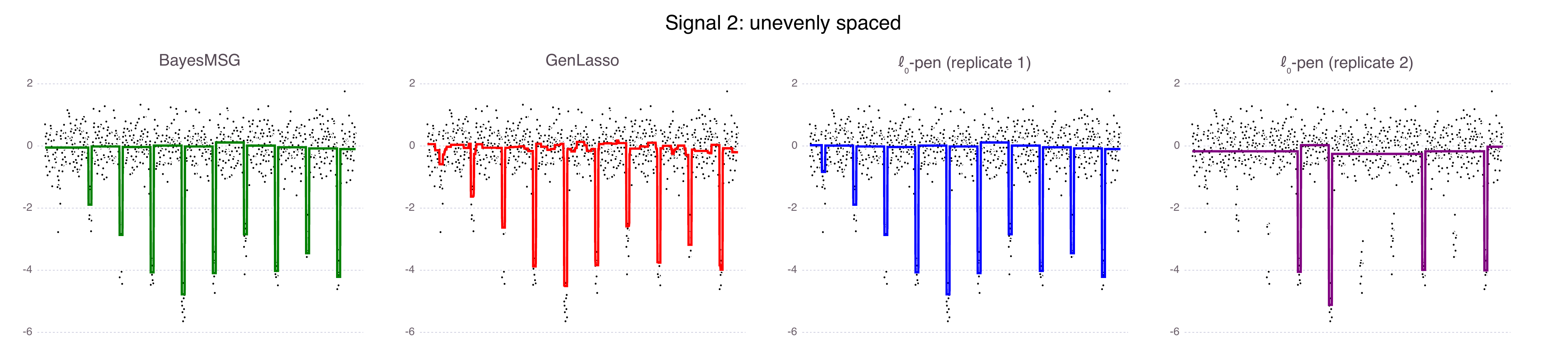}}
	\centerline{\includegraphics[width = 6.2in]{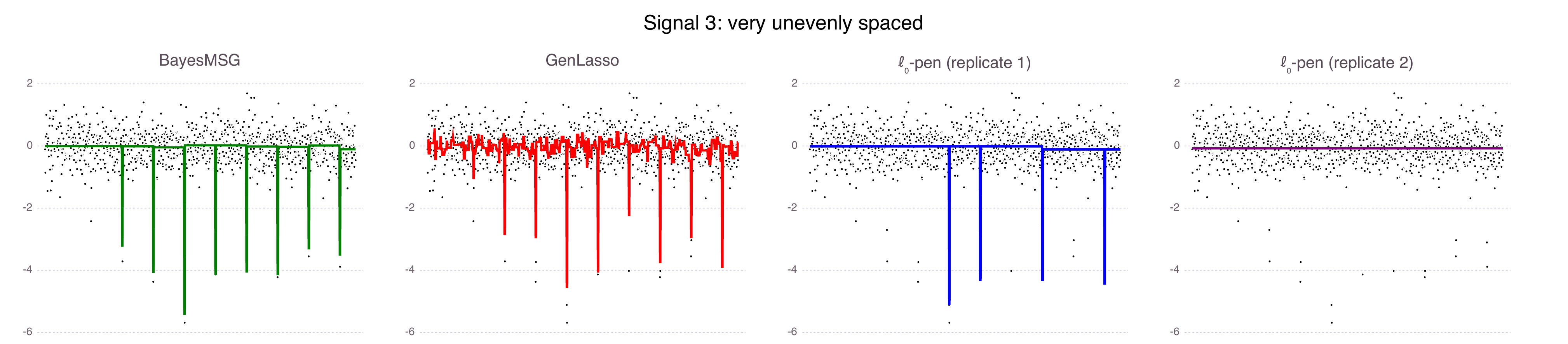}}
	\caption{Visualization of typical solutions of the three methods when $\sigma = 0.5$. Since $\ell_0$-pen is very unstable, we plot contrasting solutions from two independent replicates.
		\textrm{(Top)} Evenly spaced signal; \textrm{(Center)} Unequally spaced signal; \textrm{(Bottom)} Very unevenly spaced signal; \textrm{(Far Left)} BayesMSG; \textrm{(Left)} GenLasso; \textrm{(Right and Far Right)} Two independent replicates of $\ell_0$-pen.}
	\label{fig:linearpath1}
\end{figure}
\subsection{Simulation Studies}

In this section, we compare the proposed Bayesian model selection procedure with existing methods in the literature. There are two popular generic methods for graph-structured model selection in the literature. The first method is the generalized Lasso (or simlply GenLasso henceforth) \citep{tibshirani2005sparsity,she2010sparse,tibshirani2011solution}, defined by
\baa
\wh{\beta}=\frac{1}{2}\|y-X\beta\|^2+\lambda\sum_{(i,j)\in E}|\beta_i-\beta_j|.\label{eq:GenLasso}
\eaa
The second method is the $\ell_0$-penalized least-squares \citep{barron1999risk,friedrich2008complexity,fan2018approximate}, defined by
\baa
\wh{\beta}=\frac{1}{2}\|y-X\beta\|^2+\lambda\sum_{(i,j)\in E}\mathbb{I}\{\beta_i\neq \beta_j\}.\label{eq:l0-pen}
\eaa
For both methods, an estimated subgraph is given by
\ba
\wh{\gamma}_{ij}=\mathbb{I}\{|\wh{\beta}_i-\wh{\beta}_j|\leq\epsilon\},
\ea
for all $(i,j)\in E$. Here, the number $\epsilon$ is taken as $10^{-8}$. The two methods are referred to by GenLasso and $\ell_0$-pen from now on. In addition to GenLasso and $\ell_0$-pen, various other methods \citep{perez2014genome,bondell2008simultaneous,govaert2003clustering,tan2014sparse,gao2016optimal,chi2017convex,mair2009isotone,gao2017minimax,tibshirani2011nearly} that are specific to different models will also be compared in our simulation studies.

\subsubsection{Linear Chain Graph}
\label{sec:linear-chain-sim}

We first consider the simplest linear chain graph, which corresponds to the change-point mdoel explained in Example \ref{ex:cpm}. We generate data according to $y\sim N(\beta^*,\sigma^2I_n)$ with $n=1000$ and $\sigma\in\{0.1,0.2,0.3,0.4,0.5\}$. The mean vector $\beta^*\in\mathbb{R}^n$ is specified in three different cases as shown in Figure~\ref{fig:linearpath0}.

We compare the performances of the proposed Bayesian method, GenLasso and $\ell_0$-pen. For the linear chain graph, GenLasso is the same as fused Lasso \citep{tibshirani2005sparsity}. Its tuning parameter $\lambda$ in (\ref{eq:GenLasso}) is selected by cross validation using the default method of the R package {\tt genlasso} \citep{arnold2014genlasso}. For $\ell_0$-pen, the $\lambda$ in (\ref{eq:l0-pen}) is selected using the method suggested by \cite{fan2018approximate}.
\begin{table}[t]
	\centering
	\small
	\begin{tabular}{l | l || lll | lll | lll  }
		\multirow{2}{*}{}  & \multirow{2}{*}{$\sigma$}& \multicolumn{3}{|c|}{Even} &  \multicolumn{3}{|c|}{Uneven} &   \multicolumn{3}{|c}{Very uneven}  \\
		& & {\sf MSE}  & {\sf FDP} & {\sf POW}  & {\sf MSE}  & {\sf FDP} & {\sf POW} & {\sf MSE}  & {\sf FDP} & {\sf POW} \\
		\hline
		\hline
		\multirow{5}{*}{BayesMSG}
		& $0.1$ & 0.00019 & 0.00 & 1.00 & 0.00949 & 0.00 & 1.00 & 0.00217 & 0.00 & 0.80  \\
		& $0.2$ & 0.00585 & 0.00 & 0.98 & 0.01010 & 0.00 & 0.97 & 0.00279 & 0.00 & 0.81  \\
		& $0.3$ & 0.01620 & 0.01 & 0.96 & 0.01116 & 0.01 & 0.97 & 0.00349 &  0.00 & 0.81 \\
		& $0.4$ & 0.01940 & 0.05 & 0.95 & 0.01693 & 0.02 & 0.96 & 0.00837 &  0.00 & 0.79 \\
		& $0.5$ & 0.04667 & 0.10 & 0.95 & 0.03682 & 0.02 & 0.96 & 0.01803 & 0.05 & 0.78 \\
		\hline
		\multirow{5}{*}{GenLasso}
		& $0.1$ & 0.00094 & 0.81 &  1.00 & 0.00116 & 0.90 & 1.00 & 0.00570 & 0.96 & 1.00  \\
		& $0.2$ & 0.00374 & 0.81 & 1.00 &  0.00458 & 0.90 & 1.00 & 0.01152 & 0.94 & 1.00 \\
		& $0.3$ & 0.00842 & 0.81 & 0.98 & 0.01024& 0.89 & 1.00 & 0.02084 & 0.93 & 0.99 \\
		& $0.4$ & 0.01494 & 0.81 & 0.98 & 0.01813 & 0.88 & 0.98 & 0.03376 & 0.92  & 0.98 \\
		& $0.5$ & 0.02345 & 0.82 & 0.98 & 0.02818& 0.89 & 0.98 & 0.04984 & 0.92 & 0.97 \\
		\hline
		\multirow{5}{*}{$\ell_0$-pen}
		& $0.1$ & 0.00505 & 0.00 & 0.98 & 0.00288 & 0.00 & 0.97 & 0.02042 & 0.00 & 0.81 \\
		& $0.2$ & 0.00545 & 0.00 & 0.98 &  0.00888 & 0.00 & 0.94 & 0.06049 & 0.00 & 0.63 \\
		& $0.3$ & 0.00399 & 0.01 & 0.98 & 0.00918 & 0.02 & 0.94 & 0.06121 & 0.00 & 0.63 \\
		& $0.4$ & 0.00826 & 0.02 & 0.97 & 0.01119 & 0.02 & 0.93 & 0.06250 & 0.00 & 0.63  \\
		& $0.5$ & 0.06512 & 0.03 & 0.92  &0.04627& 0.02 & 0.93 & 0.06452 & 0.00 & 0.63  \\
	\end{tabular}
\normalsize
\caption{Comparisons of the three methods for the linear chain graph.} \label{table:1}
\end{table}

The results are summarized in Table \ref{table:1}. Some typical solutions of the three methods are plotted in Figure~\ref{fig:linearpath1}. In terms of {\sf MSE}, our Bayesian method achieves the smallest error among the three methods when $\sigma$ is small, and GenLasso has the best performance when $\sigma$ is large. For model selection performance measured by {\sf FDP} and {\sf POW}, the Bayesian method is the best, and $\ell_0$-pen is better than GenLasso. We also point out that the solutions of $\ell_0$-pen is highly unstable, as shown in Figure~\ref{fig:linearpath1}. In terms of computational time, BayesMSG, GenLasso and $\ell_0$-pen require 5.2, 11.8 and 19.0 seconds on average.

It is not surprising that GenLasso achieves the lowest {\sf MSE} in the low signal strength regime. This is because Lasso is known to produce estimators with strong bias towards zero, and therefore it is favored when the true parameters are close to zero. The other two methods, BayesMSG and $\ell_0$ are designed to achieve nearly unbiased estimators when the signals are strong, and therefore show their advantages over GenLasso when the signal strength is large.

An interesting question would be if it is possible to design a method that works well in all of the three criteria ({\sf MSE}, {\sf POW}, {\sf FDP}) with both low and strong signal? Unfortunately, a recent paper \citep{song2018optimal} proves that this is impossible. The result of \citet{song2018optimal} rigorously establishes the incompatibility phenomenon between selection consistency and rate-optimality in high-dimensional sparse regression. We therefore believe our proposed BayesMSG, which performs very well in terms of the three criteria ({\sf MSE}, {\sf POW}, {\sf FDP}) except for the {\sf MSE} in the low signal regime is a very good solution in view of this recent impossibility result.

\subsubsection{Regression with Graph-structured Coefficients}
\label{sec:reg-graph-structure-sim}

For this experiment, we consider a linear regression setting with graph structured sparsity on the regression coefficients. We sample random Gaussian measurements $X_{ij} \sim N(0,1)$ and measurement errors $\epsilon_i \sim N(0,1)$. Then we construct a design matrix $X \in \mathbb{R}^{n \times p}$ and a response vector $y = X\theta + \epsilon \in \mathbb{R}^n$, where $\theta$ is a vector of node attributes of an underlying graph $G$. 

We fix $n = 500$ throughout this simulation study, and consider three different graph settings listed in Table~\ref{table:r1-appendix}. When $G$ is a star graph with the center node at $0$, our proposed model in Section \ref{sec:model} corresponds to the sparse linear regression problem. We compare our proposed approach BayesMSG with the following baseline methods implemented in R programming language: Lasso ({\tt glmnet} R package) \citep{friedman2010regularization}, and Bayesian Spike-And-Slab linear regression (BSAS) via MCMC ({\tt BGLR} R package) \citep{perez2014genome}. All the R packages listed here are implemented using their default setting and their recommended model selection methods.
\begin{table}[!t]
	\small
	\centering
	\begin{tabular}{l || l | l | l  lll}
		Graph &	\# of nodes & \# of edges & Description \\
		\hline
		\hline
		Star & 1,001 (1 fixed node) & 1,001 &  regression with sparse coefficients \\
		Linear chain 	& 1,000 & 999 &  regression with gradient-sparse signals  \\
		Complete & 200 & 19,900 & regression with clustered coefficients \\
	\end{tabular}
\normalsize
	\caption{Simulation Settings for Gaussian Design.} \label{table:r1-appendix}
	\vspace{0.25in}
	\small
	\centering
	\begin{tabular}{l | c | c | c | c | c | c | c | c | c}
		Graph & \multicolumn{3}{|c|}{Star Graph} & \multicolumn{3}{|c|}{Linear Chain Graph} & \multicolumn{3}{|c}{Complete Graph}\\
		\hline
		Method & BMSG & Lasso & BSAS & BMSG & GLasso & ITALE & BMSG & GLasso &  OSCAR  \\
		\hline
		\hline
		{\sf MSE} & 0.579 & 0.792 & 0.530 & 0.099 & 0.276 & 0.159 & 0.399 & 0.472 & 0.438  \\
		{\sf Time} & 5.425 & 1.663 & 11.32 & 3.161 & 5.791 & 6.543 & 27.52 & 50.62 & 18.18 \\
		{\sf FDP} & 0.324 & 0.662 & - & 0.000 & 0.953 & 0.106 & 0.224 & 0.614 & 0.467 \\
		{\sf PO}W & 0.978 & 0.995 & - & 0.954 & 0.980 & 0.988 & 0.996 & 1.000 & 1.000 
	\end{tabular}\\
	\vspace{0.2in}
	\begin{tabular}{l | c | c | c | c | c | c | c | c | c}
		Graph & \multicolumn{3}{|c|}{Star Graph} & \multicolumn{3}{|c|}{Linear Chain Graph} & \multicolumn{3}{|c}{Complete Graph}\\
		\hline
		Method & BMSG & Lasso & BSAS & BMSG & GLasso & ITALE & BMSG & GLasso &  OSCAR  \\
		\hline
		\hline
		{\sf MSE} & 0.624 & 0.789 & 0.592 & 0.099 & 0.274 & 0.176 & 0.403 & 0.472 & 0.442 \\
		{\sf Time} & 5.396 & 1.721 & 12.75 & 3.436 & 5.904 & 6.554 & 22.85 &  51.42 & 15.80 \\
		{\sf FDP} & 0.319 & 0.658 & -  & 0.000 & 0.930 & 0.032 & 0.205 & 0.290 & 0.208 \\
		{\sf POW} & 0.983 & 0.982 & - & 0.994 & 0.988 & 0.996 & 0.998 & 1.000 & 0.994
	\end{tabular}
\normalsize
\caption{Simulation Results for Gaussian Design $C = 0.5$ (above) $C= 1$ (below).} \label{table:r2}

\end{table}

Next, when $G$ is a linear chain graph, the model corresponds to the linear regression problem with a sparse graph difference vector $(\theta_2 - \theta_1,\cdots,\theta_p - \theta_{p-1})$. This problem setting is particularly studied for fused Lasso \citep{tibshirani2005sparsity} and the approximate $\ell_0$ regression setting (ITALE) \citep{xu2019iterative}. Therefore, we compare BayesMSG with the above baseline approaches: fused Lasso ({\tt genlasso} R package) and ITALE ({\tt ITALE} R package) on the linear chain graph.

Finally, when $G$ is a complete graph, the model corresponds to the linear regression problem with clustered coefficients, i.e. $\beta_j$'s may be clustered together. This particular problem setting is also considered in the studies of GenLasso \citep{tibshirani2011solution} and OSCAR \citep{bondell2008simultaneous}. We compare BayesMSG with the following baseline methods: GenLasso and OSCAR ({\tt lqa} R package). In brief, OSCAR seeks to solve
\begin{equation*}
\textrm{minimize}_{\beta} \quad \frac12 \norm{y - X\theta}_2^2 + \lambda \sum_{j=1}^p (c(j-1) + 1 )|\theta|_{(j)}.
\end{equation*}

A true graph-structured sparse signal is constructed as follows. For the case of star graphs, $\theta_j^* = 0.5C$ for $j = 1,\cdots,40$ and $0$ otherwise. For the cases of linear change graphs and complete graphs, $\theta_j^* = C$ for $j =1,\cdots, 0.4p$, $\theta_j^* =2C$ for $j=0.4p+1,\cdots,0.7p$, $\theta_j^* =3C$ for $j=0.7p+1,\cdots,0.9p$ and $4C$ otherwise. Tables~\ref{table:r2} displays the estimation error ({\sf MSE}s) $\norm{\theta^* - \widehat\theta}_2$ on the test data sets, and computation times. BMSG, BSAS and GenLasso are abbreviations for BayesMSG, Bayesian Spike-And-Slab regression and Generalized Lasso. Each reported error value is averaged across 10 independent simulations with different random seeds.

The results show that our proposed BayesMSG method is the overall winners across all models in both estimation error ({\sf MSE}) and model selection error ({\sf FDP} and {\sf POW}). The advantage is especially obvious for the linear chain graph and the complete graph. The only case where BayesMSG cannot beat its competitor (BSAS) is the estimation error in the star graph case (sparse linear regression). However, we note that BSAS is an MCMC-based method that does not involve model selection but perform Bayesian model averaging. Thus, the solution of BSAS is not sparse. On the other hand, BayesMSG is designed for model selection, and therefore performs much better in terms of model selection error ({\sf FDP} and {\sf POW}).

\begin{figure}[htbp]
	\centering
	\centerline{\includegraphics[width = 6in]{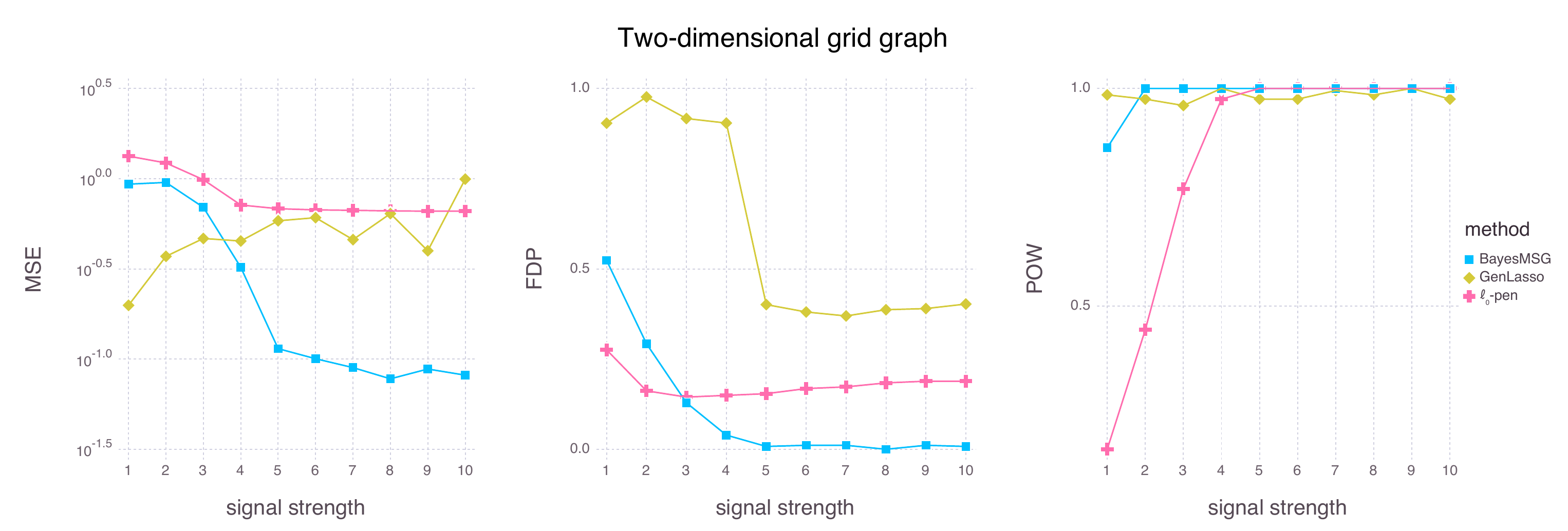}}
	\caption{Comparison of the three methods for the two-dimensional grid graph. (\textrm{Left}) {\sf MSE}; (\textrm{Center}) {\sf FDP}; (\textrm{Right}) {\sf POW}.}
	\label{fig:2dimgrid1}
	\centering
	\centerline{\includegraphics[width = 7.4in]{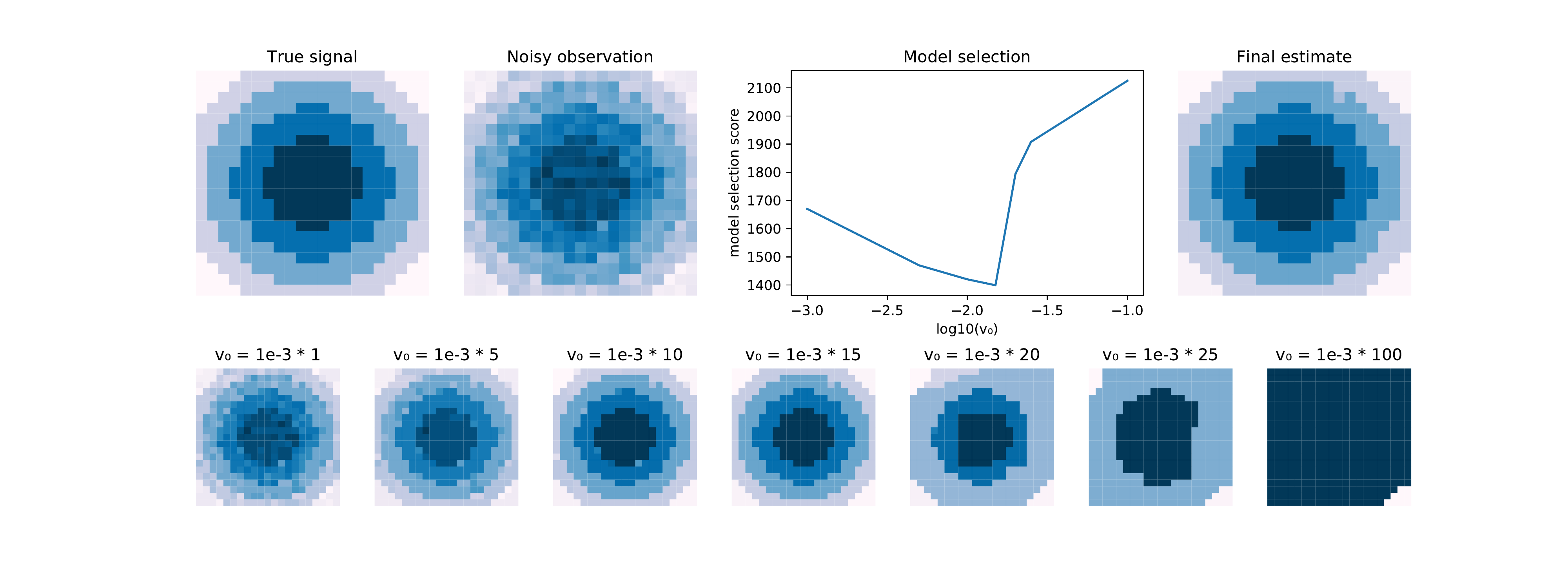}}
	\vspace{-0.35in}
	\caption{(\textrm{Top panels}) True signal, noisy observations, model selection score,  and final estimate; (\textrm{Bottom panels}) A regularization path from $v_0 = 10^{-3}$ to $v_0 = 10^{-1}$.}
	\label{fig:2dimgrid2}
	\centering
	\vspace{0.1in}
	\centerline{\includegraphics[width = 6.6in]{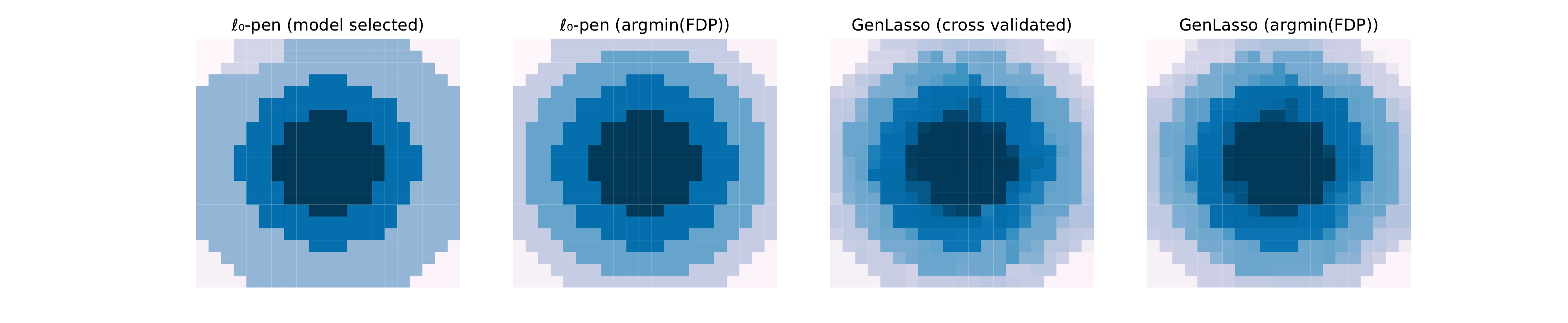}}
	\vspace{-0.03in}
	\caption{\textrm{(Far Left)} $\ell_0$-pen with $\lambda$ selected using the method in \cite{fan2018approximate}; \textrm{(Left)} $\ell_0$-pen with $\lambda$ that minimizes {\sf FDP}; \textrm{(Right)} GenLasso with $\lambda$ selected by cross validation; \textrm{(Far Right)} GenLasso with $\lambda$ that minimizes {\sf FDP}.}
	\label{fig:2dimgrid3}
\end{figure}

\subsubsection{Two-Dimensional Grid Graph}
\label{sec:two-dim-grid-sim}

We consider the two-dimensional grid graph described in Example \ref{ex:2dimgrid}. The data is generated according to $y_{ij}\sim N(\kappa\mu_{ij}^*,1)$ for $i=1,...,21$ and $j=1,...,21$, where
\ba
\mu_{ij}^*=\left\lceil 2.8 \cos \left( \frac{\sqrt{i^2 + j^2}}{2\pi} \right) - 0.2 \right\rceil,
\ea
and $\kappa\in\{1,2,...,10\}$ is used to control the signal strength. Note that $\mu_{ij}^*$ has a piecewise constant structure because of the operation by $\left\lceil\cdot\right\rceil$ that denotes the integer part. In fact, $\mu_{ij}^*$ only takes $5$ possible values as shown in Figure  \ref{fig:2dimgrid2}.

Since the R package {\tt genlasso} does not provide a tuning method for the $\lambda$ in (\ref{eq:GenLasso}) for the two-dimensional grid graph setting, we report {\sf MSE} based on the $\lambda$ selected by cross validation, and {\sf FDP} and {\sf POW} are reported based on the $\lambda$ that minimizes the {\sf FDP}. The $\lambda$ in $\ell_0$-pen is tuned by the method in \cite{fan2018approximate}.

The results are shown in Figure \ref{fig:2dimgrid1}. It is clear that our method outperforms the other two in terms of all the evaluation criteria when the signal strength is not very small. When the signal strength is very small, GenLasso achieves the lowest {\sf MSE} but shows poor model selection performance. We also illustrate the solution path of our method in Figure \ref{fig:2dimgrid2}. Typical solutions of GenLasso and $\ell_0$-pen are visualized in Figure \ref{fig:2dimgrid3}. We observe that $\ell_0$-pen tends to oversmooth the data, while GenLasso tends to undersmooth. In terms of the computational time, BayesMSG, GenLasso and $\ell_0$-pen require 21.2, 26.7 and 8.4 seconds on average.

\subsubsection{Generic Graphs}
\begin{table}[t]
	\centering
	\small
	\begin{tabular}{l || l | l | l | l | l | lll}
		Name &	\# of nodes & \# of edges & mean.ER & sd.ER & diameter & \# of CC \\
		\hline
		\hline
		Chicago roadmap & 4126 & 4308 & 0.9575 & 0.0499 & 324 & 1 \\
		Enron email 	& 4112 & 14520 & 0.2831 & 0.2341 & 14 & 1 \\
		Facebook egonet & 4039 & 88234 & 0.0457 & 0.0608 & 8 & 1 \\
	\end{tabular}
\normalsize
\caption{Graph properties of the three real networks.} \label{table:2}
	\vspace{0.2in}
	\centering
	\small
	\begin{tabular}{l || l | l | l | l l  lll}
		Name & \# of clust &	\# of nodes in each cluster & \# of cuts &  total variation \\
		\hline
		\hline
		Chicago roadmap & 4 & $(576, 678, 835, 2037)$ & 31 & $31 \times \kappa$ \\
		Enron email 	& 4  & $(384, 538, 1531, 1659)$ & 4570 & $5047 \times \kappa$  \\
		Facebook egonet & 4  & $(750, 753, 778, 1758)$ & 651  & $ 1220 \times \kappa$ \\
	\end{tabular}
\normalsize
\caption{Important features of the signals on the three networks.} \label{table:3}
\end{table}
\begin{figure}[!h]
	\centering
	\centerline{\includegraphics[width = 6.2in]{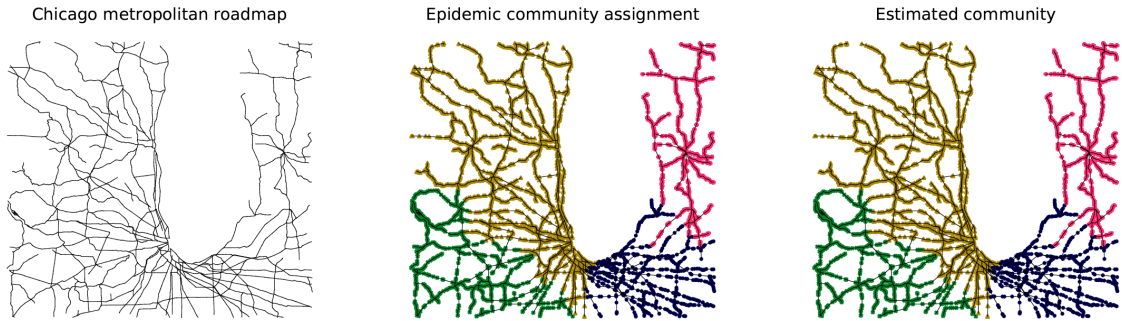}}
	\caption{The Chicago roadmap network with signals exhibiting four clusters.}
	\label{fig:generic0}
\end{figure}

In this section, we consider some graphical structures that naturally arise in real world applications.
The three graphs to be tested are the Chicago metropolitan area road network\footnote{The data set can be retrieved from \url{http://www.cs.utah.edu/~lifeifei/SpatialDataset.htm}.}, the Enron email network\footnote{The data set can be retrieved from \url{http://snap.stanford.edu/data/email-Enron.html}.}, and the Facebook egonet network\footnote{The data set can be retrieved from \url{http://snap.stanford.edu/data/ego-Facebook.html}.}. For all the three networks, we extract induced subgraphs of sizes about $4000$.
Graph properties for the three networks are summarized in Table~\ref{table:2}.
For each network, we calculate its number of nodes, number of edges, mean and standard deviation of effective resistances, diameter, and number of connected components. We observe that the three networks behave very differently.
The Chicago roadmap network is locally and globally tree-like, since its number of edges is very close to its number of nodes, and the distribution of its effective resistances highly concentrates around $1$. The other two networks, the Enron email network and the Facebook egonet, are denser graphs but their effective resistances behave in very different ways. 

For each network, we generate data according to $y_i\sim N(\kappa\mu_i^*,1)$ on its set of nodes, with the signal strength varies according to $\kappa\in\{1,2,...,5\}$. The signal $\mu^*$ for each graph is generated as follows:
\begin{enumerate}
	\item Pick four anchor nodes  from the the set of all nodes uniformly at random.
	\item For each node, compute the the length of the shortest path to each of the four anchor nodes.
	\item Code the $i$th node by $j$ if the $j$th anchor node is the closest one to the $i$th node. This gives four clusters for each graph.
	\item Generate a piecewise constant signal $\mu_i^* = j$.
\end{enumerate}
Some properties of the signals are summarized in Table \ref{table:3}, where the number of cuts of $\mu^*$ with respect to the base graph $G=(V,E)$ is defined by $\sum_{(i,j)\in E}\mathbb{I}\{\mu_i^*\neq \mu_j^*\}$, and the total variation of $\mu^*$ means $\sum_{(i,j)\in E}|\mu_i^*-\mu_j^*|$.
We also plot the signal on the Chicago roadmap network in Figure \ref{fig:generic0}.

Since the R package {\tt genlasso} does not provide a tuning method for the $\lambda$ in (\ref{eq:GenLasso}) for a generic graph, we report {\sf MSE} based on the $\lambda$ selected by cross validation, and {\sf FDP} and {\sf POW} are reported based on the $\lambda$ that minimizes the {\sf FDP}. The $\lambda$ in $\ell_0$-pen is tuned by the method in \cite{fan2018approximate}.

\begin{figure}[!t]
	\centering
	\centerline{\includegraphics[width = 6.2in]{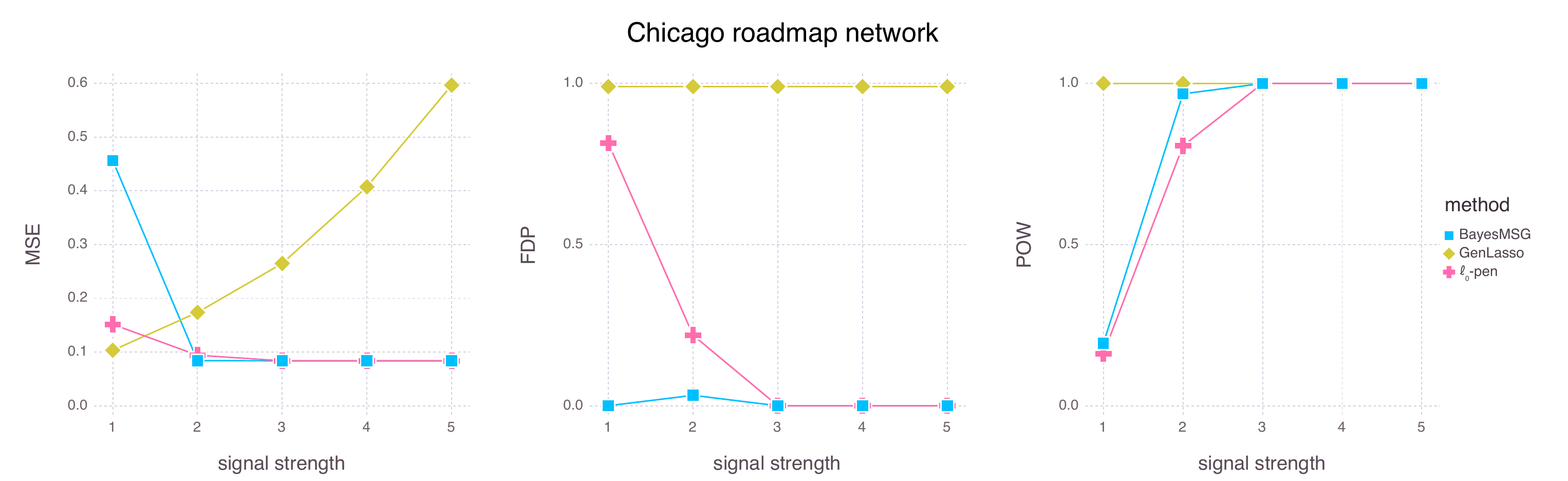}}
	\centerline{\includegraphics[width = 6.2in]{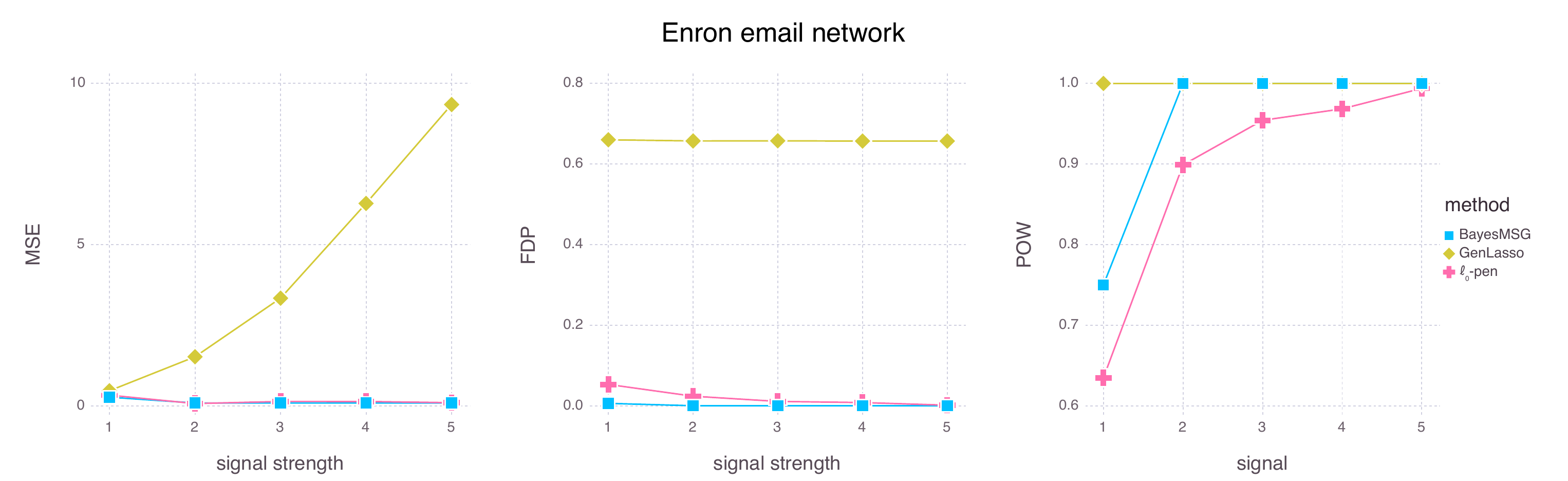}}
	\centerline{\includegraphics[width = 6.2in]{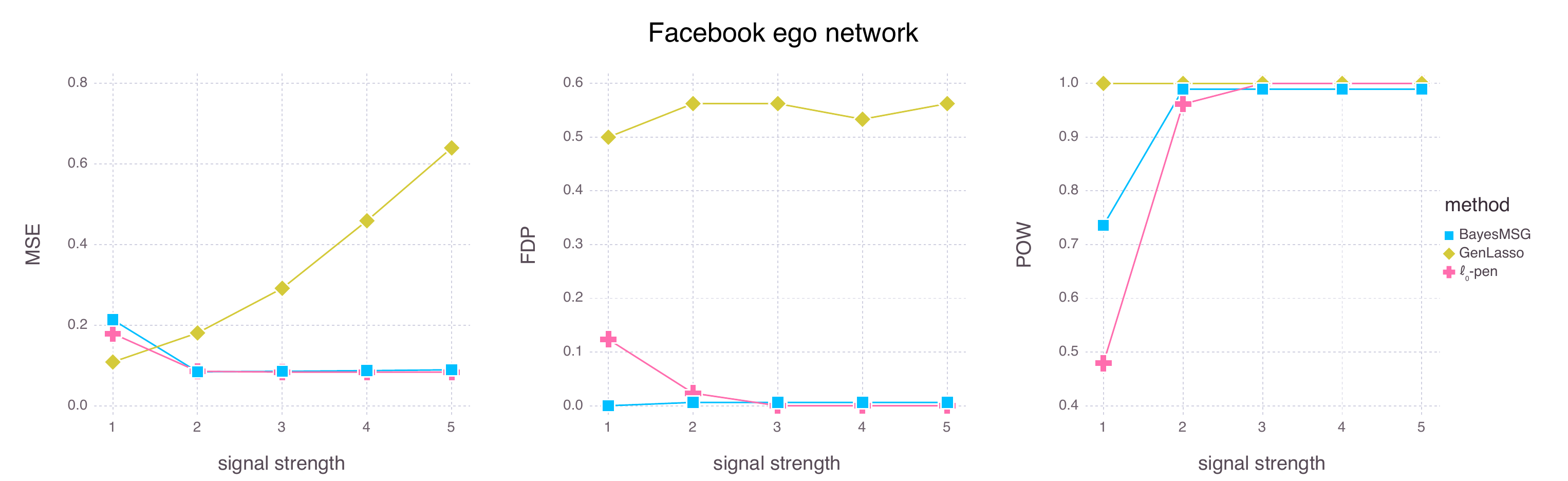}}
	\caption{Comparison of the three methods on generic graphs.
		(\textrm{Top}) Chicago Roadmap network; (\textrm{Center}) Enron Email network; (\textrm{Bottom}) Facebook Ego network.}
	\vspace{-0.3in}
	\label{fig:generic1}
\end{figure}
The results are shown in Figure \ref{fig:generic1}. It is clear that our method outperforms the other two. When the signal strength $\kappa$ is small, we observe that GenLasso sometimes has the smallest {\sf MSE}, but its {\sf MSE} grows very quickly as $\kappa$ increases. For most $\kappa$'s, our method and $\ell_0$-pen are similar in terms of {\sf MSE}. In terms of the model selection performance, GenLasso is not competitive, and our method outperforms $\ell_0$-pen.

\subsubsection{Comparison of Different Base Graphs}
\label{sec:comparision-of-different-base-graphs}
\begin{figure}[!t]
	\centering
	\centerline{\includegraphics[width = 6in]{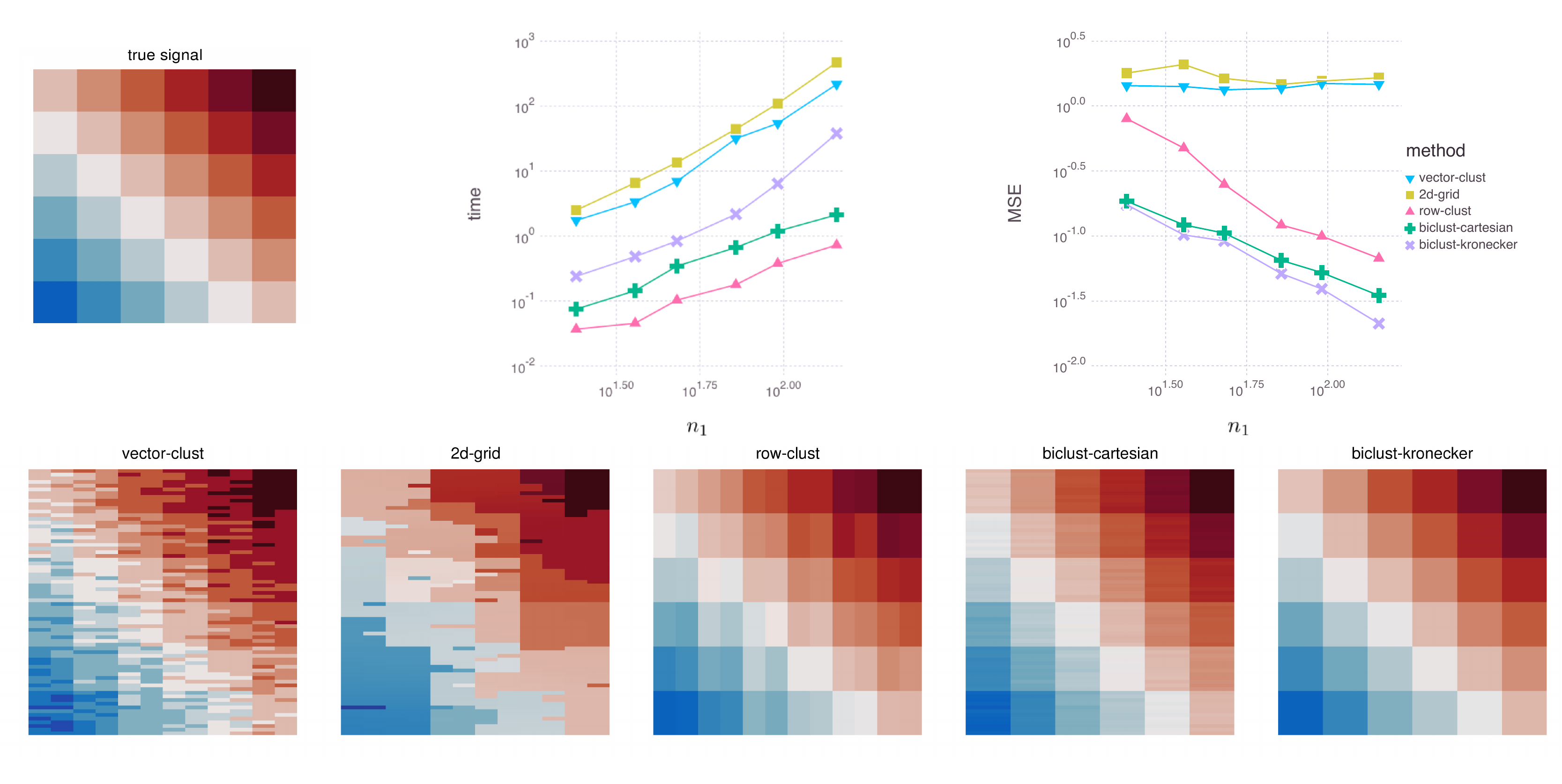}}
	\caption{(\textrm{Top left}) True signal; (\textrm{Top center}) Computational time; (\textrm{Top right}) {\sf MSE}; \textrm{(Bottom)}  Heatmaps of estimators using different models ($n_1=72$).}
	\label{fig:check}
\end{figure}

One key ingredient of our Bayesian model selection framework is the specification of the base graph. For the same problem, there can be multiple ways to specify the base graph that lead to completely different models and methods. In this section, we consider an example and compare the performances of different Bayesian methods with different base graphs.

We consider observations $y_{ij}\sim N(\theta_{ij}^*,1)$ for $i\in[n_1]$ and $j\in[n_2]$. We fix $n_2=12$ and vary $n_1$ from $24$ to $144$. The signal matrix $\theta^*\in\mathbb{R}^{n_1\times n_2}$ has a checkerboard structure as shown in Figure \ref{fig:check}. That is, the $n_1\times n_2$ matrix is divided into $6\times 6$ equal-sized blocks. On the $(u,v)$th block, $\theta_{ij}^*=2(u+v-6)$.

The following models are considered to fit the observations:
\begin{enumerate}
	\item \textit{Vector clustering.} We regard the matrix $y\in\mathbb{R}^{n_1\times n_2}$ as a $n_1n_2$-dimensional vector and apply the clustering model described in Section \ref{sec:clustering-model} with $n=k=n_1n_2$.
	\item \textit{Two-dimensional grid graph.} The two-dimensional image denoising model described in Example \ref{ex:2dimgrid} is fit to the observations.
	\item \textit{Row clustering.} We regard the matrix $y\in\mathbb{R}^{n_1\times n_2}$ as $n=n_1$ observations in $\mathbb{R}^d$ with $d=n_2$, and then fit the clustering model described in Section \ref{sec:clustering-model} to the rows of $y$ with $n=k=n_1$.
	\item \textit{Cartesian product biclustering.} The biclustering model induced by the Cartesian product described in Section \ref{sec:biclust} is fit to the observations.
	\item \textit{Kronecker product biclustering.} The biclustering model induced by the Kronecker product described in Section \ref{sec:biclust} is fit to the observations.
\end{enumerate}

Figure~\ref{fig:check} summarizes the results. In terms of {\sf MSE}, the vector clustering and the two-dimensional grid graph do not fully capture the structure of the data and thus perform worse than all other methods. Both the biclustering models are designed for the checkerboard structure, and they therefore have the best performances. Between the two biclustering models, the one induced by the Kronecker product has a smaller {\sf MSE} at the cost of a higher computational time.

To summarize the comparisons, we would like to emphasize that the right choice of the base graph has an enormous impact to the result. This also highlights the flexibility of our Bayesian model selection framework that is able to capture various degrees of structures of the data.

\subsubsection{Biclustering}
\label{sec:biclust-sim-manu}
\begin{figure}[!t]
	\centering
	\centerline{\includegraphics[width = 6.2in]{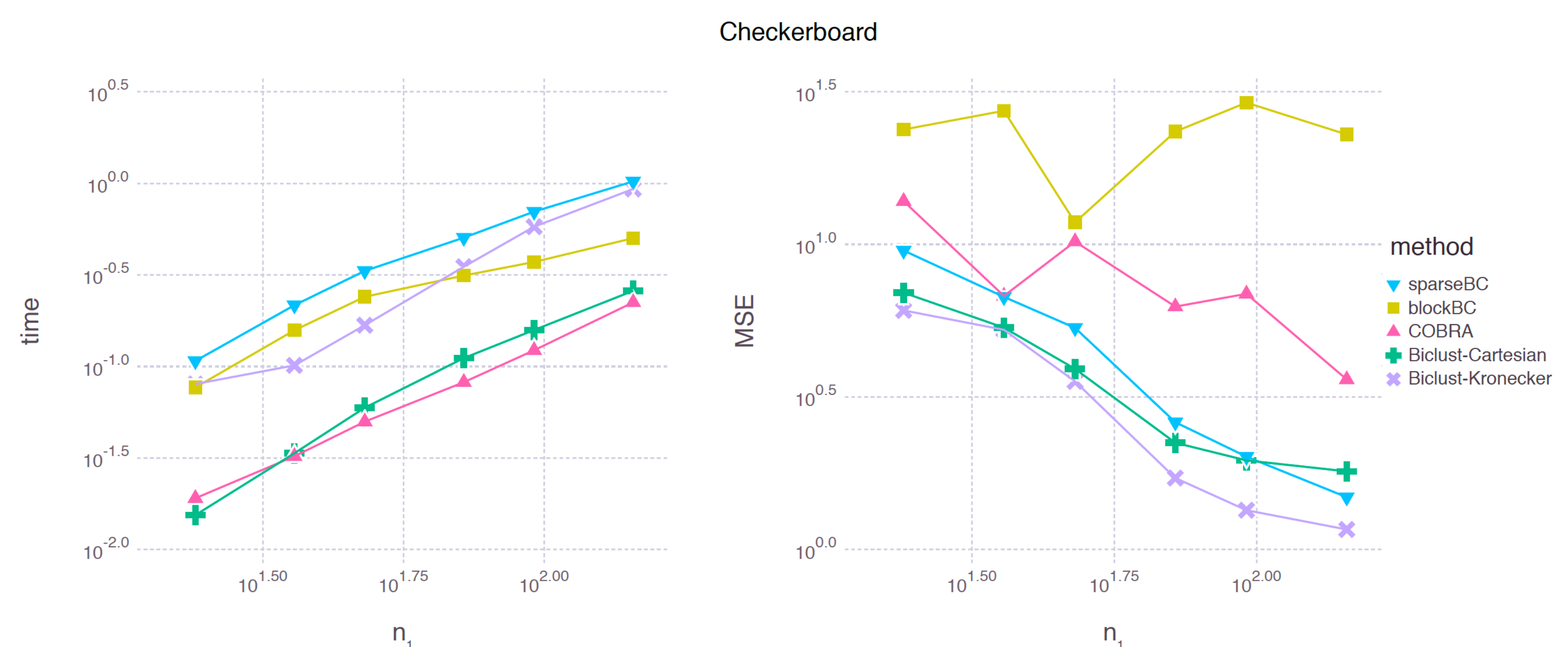}}
	\caption{Comparisons of the 5 biclustering methods for the checkerboard data. (\textrm{Left}) {\sf Time}; (\textrm{Right}) {\sf MSE}.}
	\label{fig:r4}
\end{figure}

To evaluate performance of our biclustering methods in comparison to existing methods, we provide a simulation study using the same data in Section~\ref{sec:comparision-of-different-base-graphs}. The true parameter $\theta \in \mathbb{R}^{n_1\times n_2}$ has a checkerboard structure and is visualized in Figure~\ref{fig:check}. More precisely, we consider observations $y_{ij}\sim N(\theta_{ij}^*,1)$ for $i\in[n_1]$ and $j\in[n_2]$. We fix $n_2=12$ and vary $n_1$ from $24$ to $144$. The signal matrix $\theta^*\in\mathbb{R}^{n_1\times n_2}$ has a checkerboard structure as shown in Figure \ref{fig:check}. That is, the $n_1\times n_2$ matrix is divided into $6\times 6$ equal-sized blocks. On the $(u,v)$th block, $\theta_{ij}^*=2(u+v-6)$.

For comparison, we consider the three competitors, blockBC \citep{govaert2003clustering}, sparseBC \citep{tan2014sparse},  and COBRA \citep{chi2017convex}, with the implementation via R packages {\tt blockcluster}, {\tt sparseBC} and {\tt cvxbiclustr}. In summary, blockBC and sparseBC are based on the minimization of $\sum_{i,j} (y_{ij} - \mu_{z_1(i)z_2(j)})^2$ given the number $k_1$ and $k_2$ of row and column clusters, respectively, where $\mu \in \mathbb{R}^{k_1 \times k_2}$ is a matrix of latent hidden bicluster means and $z_1$ and $z_2$ are row and column cluster assignments, respectively. The methods blockBC and sparseBC use different approaches for the estimation of row and column clusters, i.e. blockBC uses a block EM algorithm. sparseBC solves a penalized linear regression with the $\ell_1$-penalty $\lambda\norm{\mu}_1$. COBRA is an alternating direction method of multipliers (ADMM) algorithm based on the minimization of $\sum_{i,j} (y_{ij} - \theta_{ij})) + \textrm{pen}_{\rm row}(\theta) + \textrm{pen}_{\rm col}(\theta)$ where $\textrm{pen}_{\rm row}(\theta) = \sum_{i<j}w_{ij}\norm{\theta_{i*} - \theta_{j*}}_2$ and $\textrm{pen}_{\rm col}$ is defined similarly.
\begin{figure}[htbp]
	\centering
	\centerline{\includegraphics[width = 6.2in]{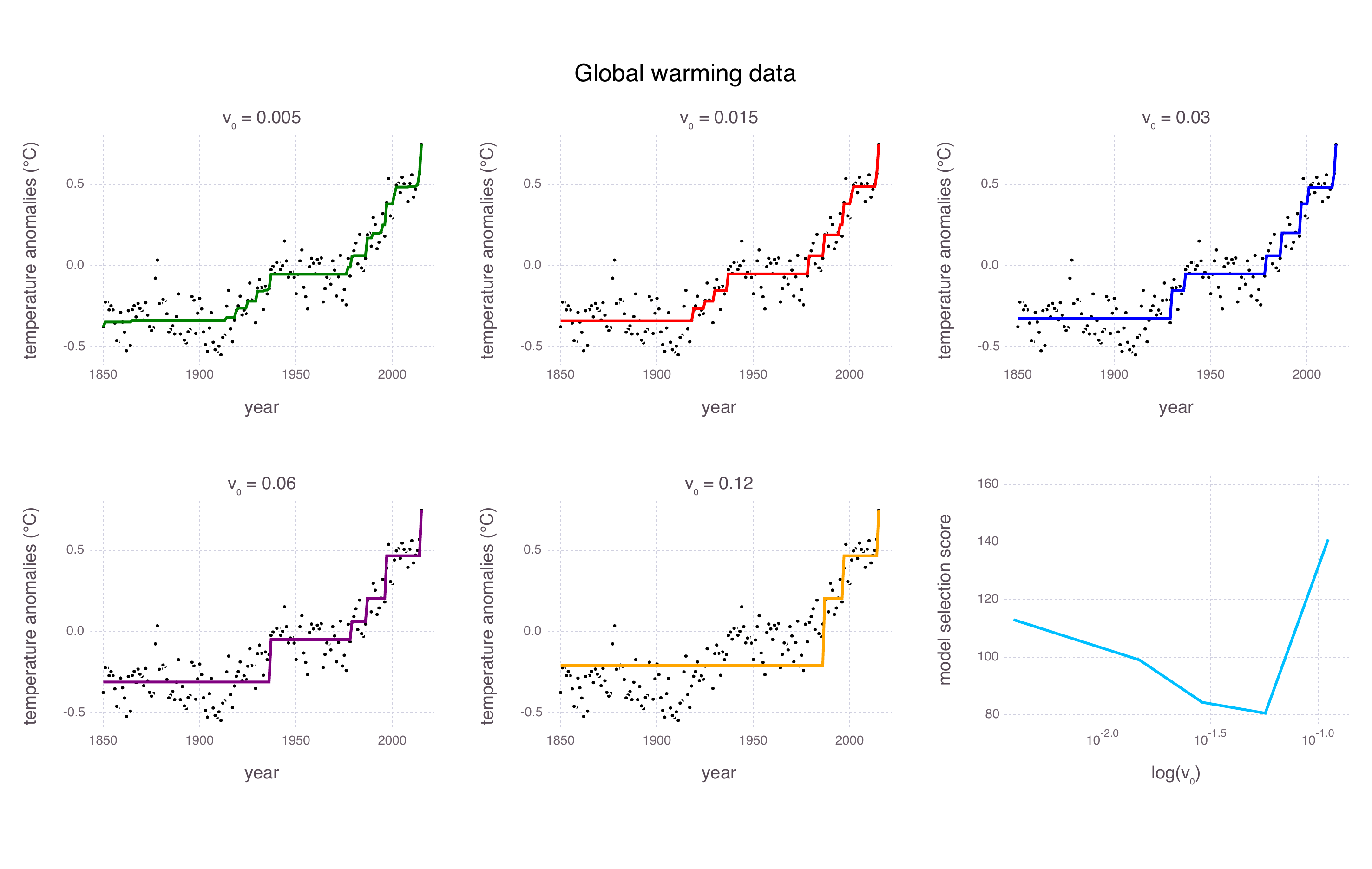}}
	\vspace{-0.3in}
	\caption{The solution path (smaller $v_0$ at top left and larger $v_0$ at bottom right) for Bayesian reduced isotonic regression.}
	\label{fig:iso1}
	\bigskip
	\vspace{0.2in}
	\centerline{\includegraphics[width = 3.2in]{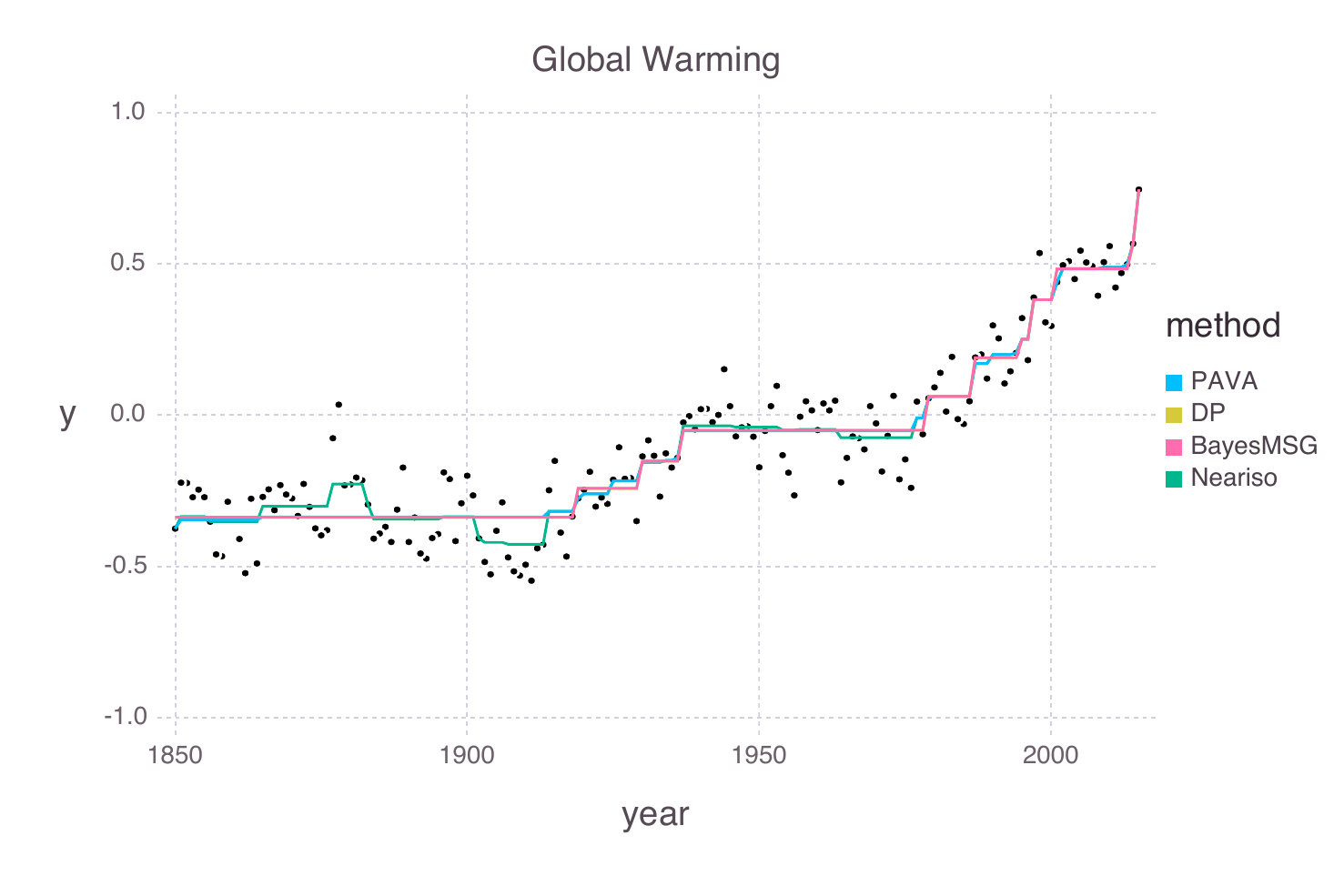} \includegraphics[width = 3.2in]{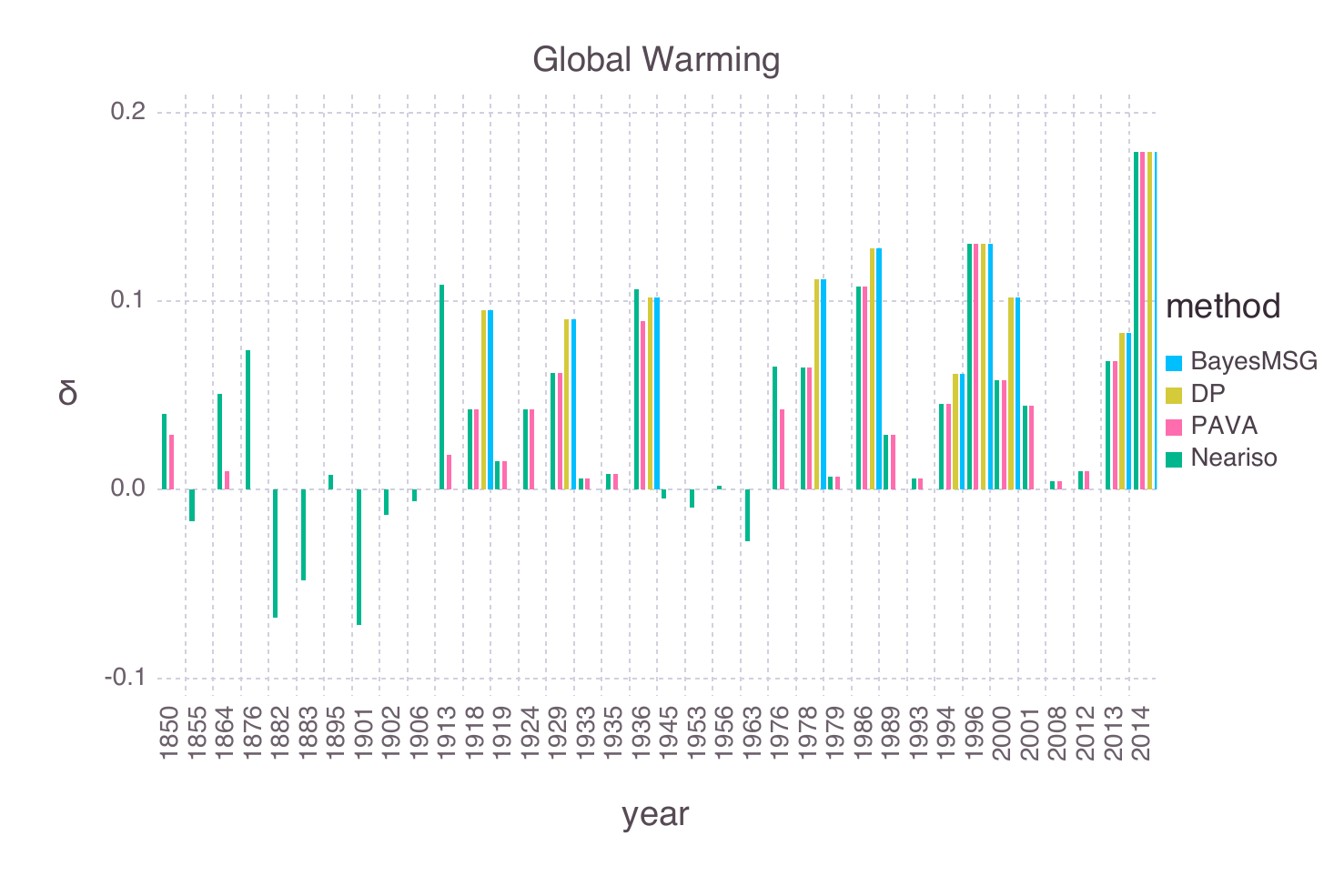}}
	\caption{Comparison of various isotonic regression methods. (\textrm{Left}) The estimated isotonic signals (PAVA, DP, BayesMSG, NearIso); let us mention that DP and BayesMSG signals exactly coincide. (\textrm{Right}) The estimated differences between the two adjacent years; the only years (x-axis) with at least one nonzero differences are reported.}
	\label{fig:r3}
\end{figure}

The result is displayed in Figure~\ref{fig:r4}. The result shows that BayesMSG Kronecker outperforms other methods in terms of {\sf MSE}. BayesMSG Cartesian also produces competitive solutions in terms of {\sf MSE} within short amount of computation time. Note that BayesMSG and COBRA are regularization path based methods, and require less computation time than other approaches (blockBC and sparseBC) based on cross validation.

\subsection{Real Data Applications}
In this section, we apply our methods to three different data sets.

\subsubsection{Global warming data}
\label{sec:global-warming-sim}
The global warming data has been studied previously by \cite{wu2001isotonic,tibshirani2011nearly}. It consists of 166 data points in degree Celsius from 1850 to 2015. Here we fit the Bayesian reduced isotonic regression discussed in Section \ref{sec:iso}. Our results are shown in Figure~\ref{fig:iso1}.

When $v_0$ is nearly zero, the solution is very close to the regular isotonic regression that can be solved efficiently by the pool-adjacent-violators algorithm (PAVA) \citep{mair2009isotone}. When $v_0=0.005$, we obtain a fit with 24 pieces. The PAVA outputs a very similar fit also with 24 pieces. In contrast, the Bayesian model selection procedure suggests a model with $v_0=0.06$, which has only $6$ pieces, a significantly more parsimonious and a more interpretable fit. This may suggest global warming is accelerating faster in recent years. The same conclusion cannot be obtained from the suboptimal fit with $24$ pieces.

To compare with existing methods, we have implemented reduced isotonic regression with dynamic programming (DP) algorithm introduced by \citet{gao2017minimax} and the near isotonic (NearIso) regression \citep{tibshirani2011nearly}. 
Figure~\ref{fig:r3} shows the estimated signals (left panel) and the estimated differences $\{\theta_{i+1} -\theta_{i}: i=1,\cdots,n-1 \}$ between the two adjacent years (right panel). Since the DP algorithm does not have a practical model selection procedure, we use the number of pieces $k = 9$ selected by BayesMSG. NearIso uses Mallow's $C_p$ for model selection  \citep{tibshirani2011nearly}. The left panel of Figure~\ref{fig:r3} shows that BayesMSG outputs a more sparse but still reasonable isotonic fit. 
On the other hand, NearIso relaxes the isotonic constraint, allowing non-monotone signals but penalizing the decreasing portion of variations. Indeed, one can see from the right panel of Figure~\ref{fig:r3} that the changepoints of NearIso contain those of PAVA. The PAVA solution is the least parsimonious isotonic fit by definition, and thus we can conclude that NearIso does not seem to find a more parsimonious fit. 
\begin{figure}[htbp]
	\centering
	\centerline{\includegraphics[width = 6.2in]{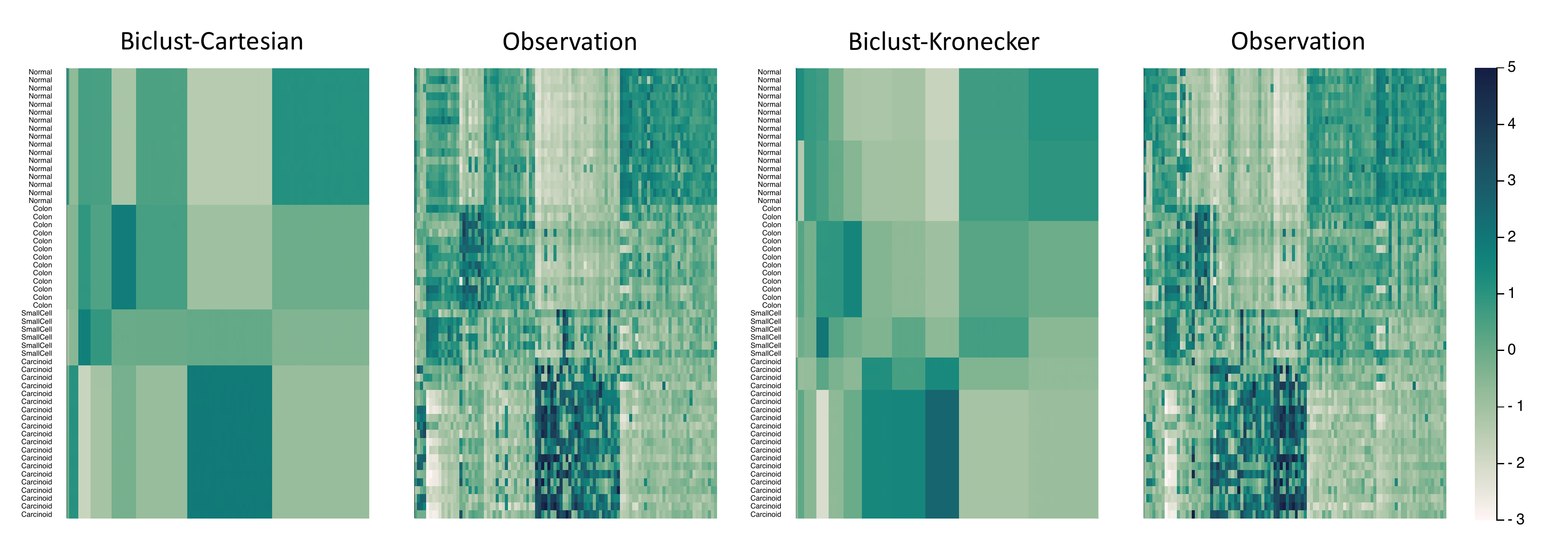}}
	\vspace{-0.1in}
	\caption{Results of biclustering for the lung cancer data\protect\footnotemark.}
	\label{fig:biclust1}
	\bigskip
	\centering
	\centerline{\includegraphics[width = 6.2in]{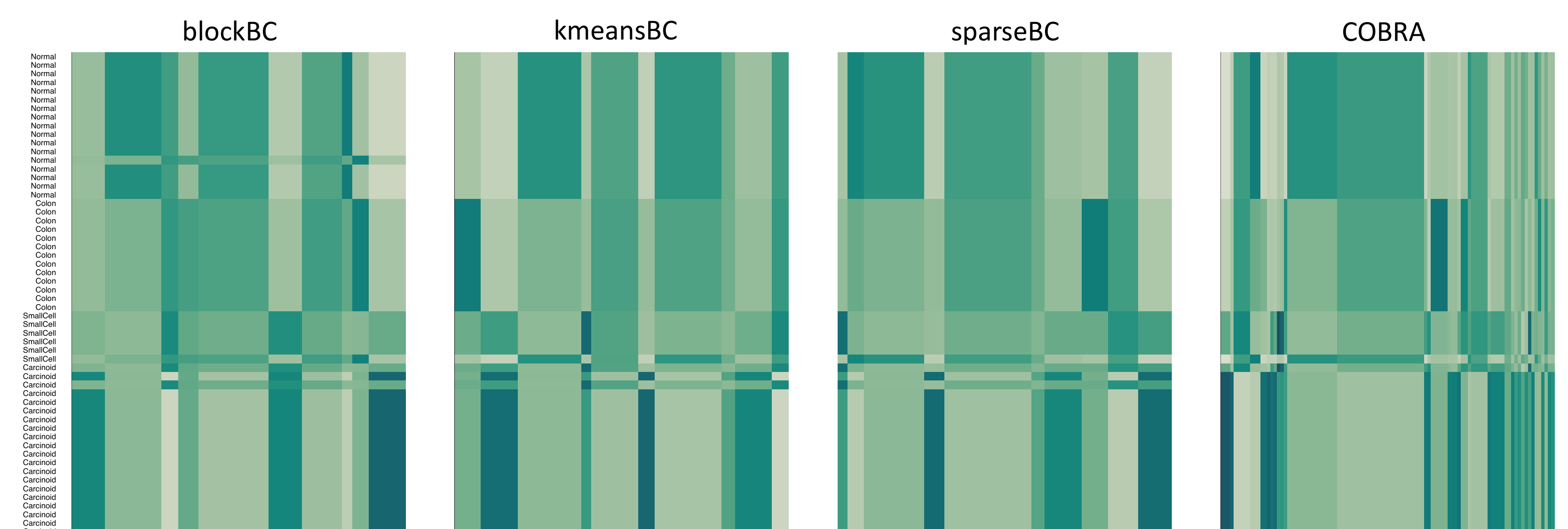}}
	\caption{Comparison of existing biclustering methods.}
	\label{fig:r5}
	\bigskip
	\centering
	\centerline{\includegraphics[width = 3.3in]{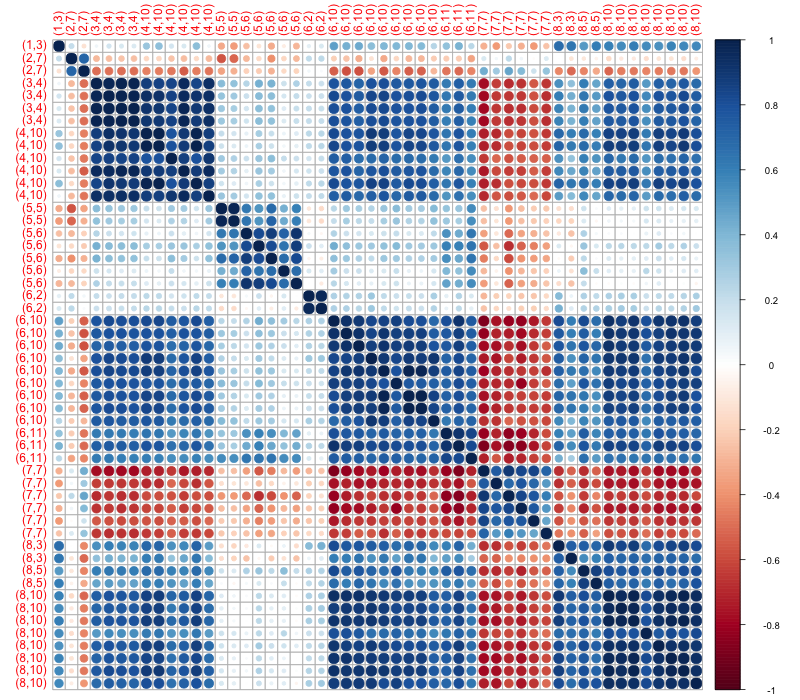}}
	\caption{A correlation plot of the selected genes.}
	\label{fig:biclust2}
\end{figure}
\subsubsection{Lung cancer data}
\label{sec:lung-cancer-sim}
We illustrate the Bayesian biclustering models by a gene expression data set from a lung cancer study. The same data set has also been used by \cite{bhattacharjee2001classification, lee2010biclustering,sill2011robust,chi2017convex}. Following \cite{chi2017convex}, we study a subset with $56$ samples and $100$ genes.  The $56$ samples comprise $20$ pulmonary carcinoid samples (Carcinoid), $13$ colon cancer metastasis samples (Colon), $17$ normal lung samples (Normal) and $6$ small cell lung carcinoma samples (SmallCell). We also apply the row and column normalizations as has been done in \cite{chi2017convex}.

Our goal is to identify sets of biologically relevant genes, for example, that are significantly expressed for certain cancer types. We fit both Bayesian biclustering models (Section \ref{sec:biclust}) induced by the Cartesian and Kronecker products to the data with $n_1=56$, $n_2=100$, $k_2=10$, and $k_2=20$. Recall that $k_1$ and $k_2$ are upper bounds of the numbers of row and column clusters, and the actual numbers of row and column clusters will be learned through Bayesian model selection. To pursue a more flexible procedure of model selection, we use two independent pairs of $(v_0,v_1)$ for the row structure and the column structure. To be specific, let $(v_0,v_1)$ be the parameters for the row structure, and the parameters for the column structure are set as $(cv_0,cv_1)$ with some $c\in \{1/10, 1/5, 1/2, 1,2,5,10\}$. Then, the model selection scores are computed with both $v_0$ and $c$ varying in their ranges.
\begin{table}[!t]
	\footnotesize
	\centering
	\begin{tabular}{c | l l lll}
		\footnotesize Label &  \footnotesize Gene description/GenBank ID \\
		\hline\hline
		\footnotesize $(2,7)$ & ``proteoglycan 1, secretory granule'', ``AI932613: Homo sapiens cDNA, 3 end" \\
		\hline
		\multirow{3}{*}{\footnotesize $(3,4)$}& ``AI147237: Homo sapiens cDNA, 3 end'', ``S71043: Ig A heavy chain allotype 2'', \\
		& ``advanced glycosylation end product-specific receptor'', \\
		& ``leukocyte immunoglobulin-like receptor, subfamily B'' \\
		\hline
		\footnotesize $(5,5)$& `immunoglobulin lambda locus", ``glypican 3'' \\
		\hline
		\multirow{3}{*}{\footnotesize$(5,6)$}  &``glutamate receptor, ionotropic, AMPA 2'', ``small inducible cytokine subfamily A'', \\ &``W60864: Homo sapiens cDNA, 3 end'', ``secreted phosphoprotein 1'',\\
		&``LPS-induced TNF-alpha factor''\\
		\hline
		\footnotesize $(6,2)$& ``interleukin 6", ``carcinoembryonic antigen-related cell adhesion molecule 5'' \\
		\hline
		\multirow{2}{*}{\footnotesize $(6,11)$}& ``secretory granule, neuroendocrine protein 1", ``alcohol dehydrogenase 2'', \\ & ``neurofilament, light polypeptide'' \\
		\hline
		\footnotesize $(8,3)$ & ``fmajor histocompatibility complex, class II",  ``glycoprotein (transmembrane) nmb" \\
		\hline
		\footnotesize $(8,5)$ & ``N90866: Homo sapiens cDNA, 3 end'', `receptor (calcitonin) activity modifying protein 1''
	\end{tabular}
\normalsize
	\caption{Description or GenBank ID of the selected gene clusters of size at least $2$ and at most $5$.} \label{table:4}
\end{table}
\footnotetext{The rows of the four heatmaps are ordered in the same way according to the labels of tumor types.}

The results are shown in Figure~\ref{fig:biclust1}. The two methods select different models with different interpretations. The Cartesian product fit gives $4$ row clusters and $8$ column clusters, while the Kroneker product fit gives $6$ row clusters and $11$ column clusters. Even though we have not used the information of the row labels for both biclustering methods, the row clustering structure output by the Cartesian product model almost coincides with these labels except one. On the other hand, the Kronecker product model leads to a finer row clustering structure, with potential discoveries of subtypes of both normal lung samples and pulmonary carcinoid samples.

Using the same lung cancer data set, we have also compared the BayesMSG Cartesian product method with several existing biclustering methods implemented via R packages {\tt blockcluster}, {\tt sparseBC} and {\tt cvxbiclustr}. Figure~\ref{fig:r4} displays the final estimates of the four competitors, blockBC \citep{govaert2003clustering}, sparseBC \citep{tan2014sparse}, kmeansBC \citep{gao2016optimal} and COBRA \citep{chi2017convex}. One can qualitatively compare these results with the BayesMSG solutions in Figure~\ref{fig:r5}.

The results show that BayesMSG and all the competitors select models with $\widehat{k}_1 = 4$ row clusters. However, BayesMSG selects a smallest number of column clusters. This implies that BayesMSG favors a more parsimonious model compared to the others. Also, the selected BayesMSG Cartesian model achieves a lowest misclassification error for the row (cancer type).

An important goal of biclustering is to simultaneously identify gene and tumor types. To be specific, we seek to find genes that show different expression levels for different types of samples. To this end, we report those genes that are clustered together by both the Cartesian and Kronecker product structures. Groups of genes with size between $2$ and $5$ are reported in Table~\ref{table:4}. Note that our gene clustering is assisted by the sample clustering in the biclustering framework, which is different from gene clustering methods that are only based on the correlation structure \citep{bhattacharjee2001classification}. As a sanity check, the correlation matrix of the subset of the selected genes is plotted in Figure \ref{fig:biclust2}, and we can observe a clear pattern of block structure.

\subsubsection{Chicago crime data}
\label{sec:chicago-crime-sim}

\begin{figure}[!t]
	\centering
	\centerline{\includegraphics[width = 6.2in]{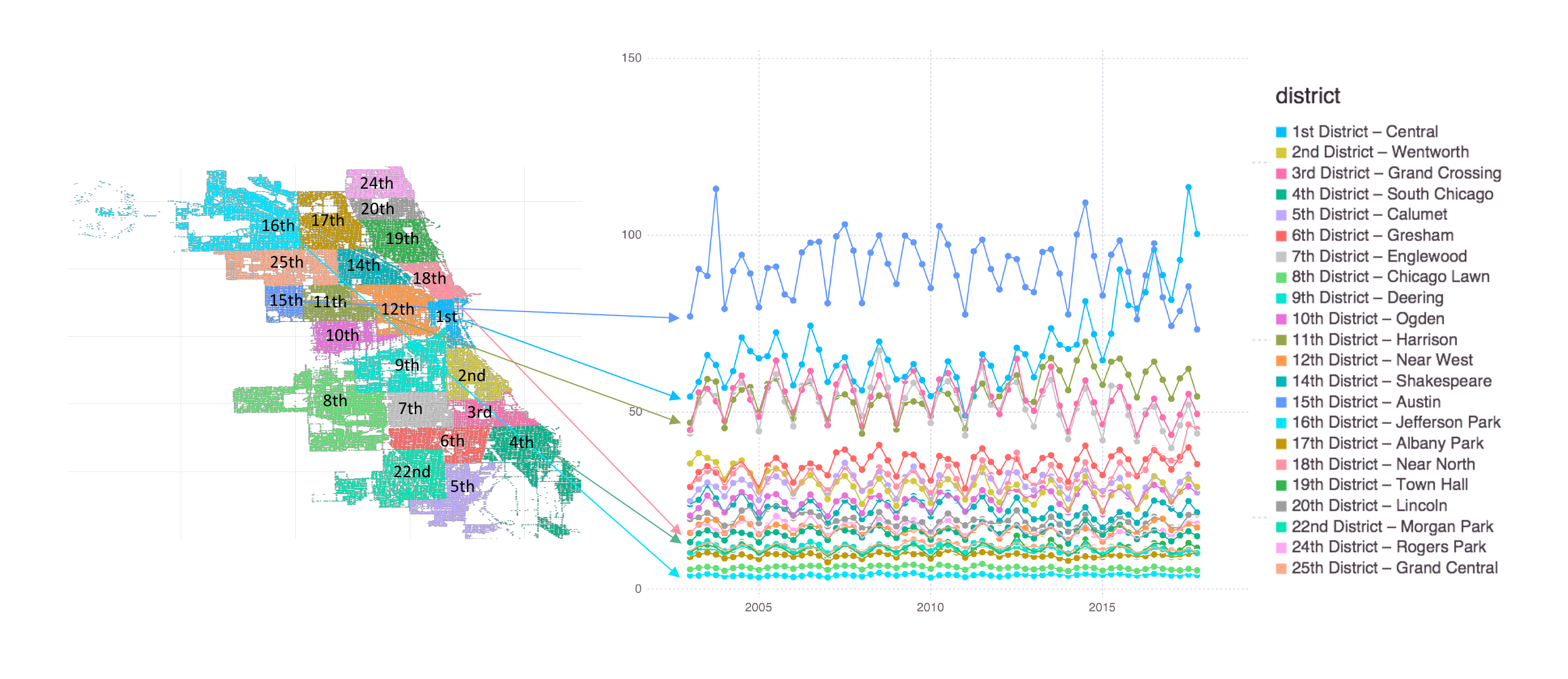}}
	\vspace{-0.15in}
	\caption{Visualization of the Chicago crime data after preprocessing.}
	\label{fig:chicago1}
	\vspace{0.3in}
	\centering
	\centerline{\includegraphics[width = 6.5in]{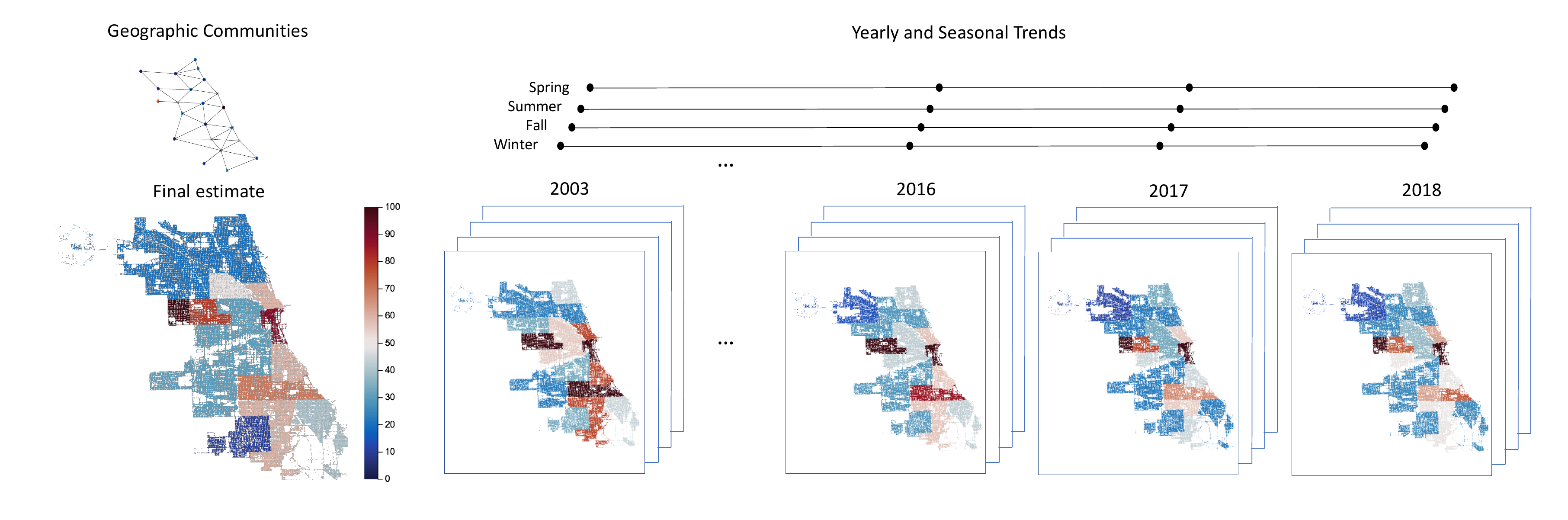}}
	\caption{Visualization of Bayesian model selection for the Chicago crime data. (\textrm{Left}) The overall geographical pattern; (\textrm{Right}) Four different patterns from $2003$ to $2018$.}
	\label{fig:chicago2}
\end{figure}

The Chicago crime data is publicly available at Chicago Police Department website\footnote{The Chicago crime data set can be retrieved from \url{https://data.cityofchicago.org/Public-Safety/Crimes-2001-to-present}.}. The report in the website contains time, type, district, community area, latitude and longitude of each crime occurred. After removing missing data, we obtain $6.0$ millions of crimes that occurred in $22$ police districts from 2003 to 2018 (16 years). Here we restrict ourselves to the analysis of the spatial and temporal structure of Chicago crimes within the past few years, ignoring the types or categories of the crimes.

Since the $22$ districts have different area sizes, we divide the total numbers of crimes in each district by its population density (population per unit area). We will call this quantity the \textit{Chicago crime rate}. We observe that the Chicago crime rates exhibit decreasing patterns over the years. Since our study is focused on the relative comparisons among different police districts, we divide each entry by the sum of the yearly Chicago crime rates over all the $22$ district in its current year. We will call this quantity the \textit{relative crime rate}. Admittedly this preprocessing step does not reflect the difference between the \textit{residential} and the \textit{floating} populations in each district which might be important for the analysis of the crime data. For instance, around O'Hare international airport, it is very likely that the floating and the residential populations differ a lot. After the preprocessing, we obtain a three-way tensor with size $22\times 16\times 4$, for $22$ districts, $16$ years, and $4$ seasons, which is visualized in Figure~\ref{fig:chicago1}.

Our main interest is to understand the geographical structure of the relative crime rates and how the structure changes over the year. A Bayesian model is constructed for this purpose by using the graphical tools under the proposed framework. We define a graph characterizing the geographical effect by $G_1 = (V_1,E_1)$ with $V_1 = \{1,2,\cdots,22\}$ and $E_1 = \{(i,j) : \text{the $i$th and $j$th districts are adjacent} \}$. A graph characterizing the temporal effect is given by $G_2 = (V_2,E_2)$ with $V_2 = \{1,2,...,16\}$ and $E_2 = \{(i,i+1): i=1,...,15\}$. Then, the $22\times 16\times 4$ tensor is modeled by a spike-and-slab Laplacian prior with the base graph $G_1 \square G_2$, in addition to a multivariate extension (Section \ref{sec:multivariate}) along the dimension of $4$ seasons.

The result of the Bayesian model selection for the Chicago crime data is visualized in Figure~\ref{fig:chicago2}. The geographical structure of the relative crime rates exhibit four different patterns according to the partition $\{2003,2004,...,2015\}$, $\{2016\}$, $\{2017\}$, $\{2018\}$. While geographical compositions of the crimes are similar from $2003$ to $2015$, our results reveal that the last three years have witnessed dramatic changes. In particular, in these three years, the relative crime rates of Districts 11 and 15 were continuously decreasing, and the relative crime rates of Districts 1 and 18 show the opposite trend. This implies that the overall crime pattern is moving away from historically dangerous areas to downtown areas in Chicago.


\acks{We thank Matthew Stephens for helpful discussions. Research of CG is supported in part by NSF grant DMS-1712957 and NSF CAREER award DMS-1847590.}

\newpage
\appendix


\section{Some Basics on Linear Algebra}

For a symmetric matrix $\Gamma$ with rank $r$, it has an eigenvalue decomposition $\Gamma=UDU^T$ with some orthonormal matrix $U\in \mathcal{O}(p,r)=\{V\in\mathbb{R}^{p\times r}:V^TV=I_r\}$ and some diagonal matrix $D$ whose diagonal entires are all positive. Then, the Moore-Penrose pseudo inverse of $\Gamma$ is defined by
\ba
\Gamma^+=UD^{-1}U^T.
\ea

\begin{lemma} \label{lemA1}
	Consider a symmetric and invertible matrix $R\in\mathbb{R}^{r\times r}$. Then, we have
	\begin{equation}
	R=V^T(VR^{-1}V^T)^+V, \label{eq:useful}
	\end{equation}
	for any $V\in\mathcal{O}(p,r)$.
\end{lemma}
\begin{proof}
	To prove (\ref{eq:useful}), we write $R$ as its eigenvalue decomposition $W\Lambda W^T$ for some $W\in\mathcal{O}(r,r)$ and some invertible diagonal matrix $\Lambda$. Then, it is easy to see that $VW\in\mathcal{O}(p,r)$, and we thus have 
	\ba
	(VR^{-1}V^T)^+=(VW\Lambda^{-1}W^TV^T)^+=VW\Lambda W^TV^T,
	\ea
	and then
	\ba
	V^T(VR^{-1}V^T)^+V=V^TVW\Lambda W^TV^TV=W\Lambda W^T=R.
	\ea
\end{proof}

\begin{lemma} \label{lem:matrix_inv}
	Let $A, B \in \mathbb{R}^{n\times m}$ be matrices of full column rank $m$ (i.e. $m \leq n$). Let $Z_A$ and $Z_B$ span the nullspaces of $A$ and $B$, respectively. That is, $Z_A,Z_B \in \mathbb{R}^{n\times (n-m)}$ and 
	\ba
	A^T Z_A =0,\quad B^T Z_B = 0.
	\ea
	Then, we have
	\ba
	I_n - A(B^T A)^{-1} B^T = Z_B (Z_A^T Z_B)^{-1} Z_A^T.
	\ea
\end{lemma}
\begin{proof}
	Let $C= I_n - A(B^T A)^{-1} B^T$, and then it is easy to check that
	\ba
	CZ_B = Z_B,\quad Z_A^TC = Z_A^T,\quad CA = 0, \quad B^TC = 0.
	\ea
	Note that the above four equations determine the singular value decomposition of $C$ and are also satisfied by $Z_B (Z_A^T Z_B)^{-1} Z_A^T$, which immediately implies $C=Z_B (Z_A^T Z_B)^{-1} Z_A^T$.
\end{proof}

\begin{lemma}\label{lem:Y-Kim}
	Suppose for symmetric matrices $S,H\in\mathbb{R}^{p\times p}$, we have $\mathcal{M}([S;H])=\mathbb{R}^p$, where the notation $\mathcal{M}(\cdot)$ means the subspace spanned by the columns of a matrix. Then, we have
	\ba
	(tS+H)^{-1}\rightarrow R(R^THR)^{-1}R^T,
	\ea
	as $t\rightarrow\infty$,
	where $R$ is any matrix such that $\mathcal{M}(R)$ is the null space of $S$.
\end{lemma}
\begin{proof}
	We first prove the special case of $H=I_p$. Denote the rank of $S$ by $r$, and then $S$ has an eigenvalue decomposition $S=UDU^T$ for some $U\in\mathcal{O}(p,r)$ and some diagonal matrix $D$ with positive diagonal entries. Since $\mathcal{M}(R)$ is the null space of $S$, we have $I_p=UU^T+R(R^TR)^{-1}R^T$ by Lemma~\ref{lem:matrix_inv}. Then,
	\ba
	(tS+I_p)^{-1}=(U(tD+I_r)U^T+R(R^TR)^{-1}R^T)^{-1}=U(tD+I_r)^{-1}U^T+R(R^TR)^{-1}R^T,
	\ea
	which converges to $R(R^TR)^{-1}R^T$ as $t\rightarrow\infty$. The second equality above follows from $U^TR = 0$. Now we assume a general $H$ of full rank, which means $H=Q^TQ$ for some $Q\in\mathbb{R}^{p\times p}$ that is invertible. Then,
	\ba
	(tS+H)^{-1}=Q^{-1}(t(Q^T)^{-1}SQ^{-1}+I_p)^{-1}(Q^T)^{-1}.
	\ea
	Since the null space of $(Q^T)^{-1}SQ^{-1}$ is $\mathcal{M}(QR)$, we have
	\ba
	(t(Q^T)^{-1}SQ^{-1}+I_p)^{-1}\rightarrow QR(R^TQ^TQR)^{-1}R^TQ^T=QR(R^THR)^{-1}R^TQ^T,
	\ea
	and therefore $(tS+H)^{-1}\rightarrow R(R^THR)^{-1}R^T$. For a general $H$ that is not necessarily full rank, since $\mathcal{M}([S;H])=\mathbb{R}^p$, $S+H$ is a matrix of full rank. Then,
	\ba
	(tS+H)^{-1}=((t-1)S+S+H)^{-1}\rightarrow R(R^T(S+H)R)^{-1}R^T=R(R^THR)^{-1}R^T,
	\ea
	and the proof is complete.
\end{proof}

\section{Degenerate Gaussian Distributions}

A multivariate Gaussian distribution is fully characterized by its mean vector and covariance matrix. For $N(\mu,\Sigma)$ with some $\mu\in\mathbb{R}^p$ and a positive semidefinite $\Sigma\in\mathbb{R}^{p\times p}$, we call the distribution degenerate if $\rank(\Sigma)<p$. Given any $\Sigma$ such that $\rank(\Sigma)=r<p$, we have the decomposition $\Sigma=AA^T$ for some $A\in\mathbb{R}^{p\times r}$. Therefore, $X\sim N(\mu,\Sigma)$ if and only if
\begin{equation}
X=\mu+AZ, \label{eq:degen-Gau-latent}
\end{equation}
where $Z\sim N(0,I_r)$. The latent variable representation (\ref{eq:degen-Gau-latent}) immediately implies that $X-\mu\in\mathcal{M}(A)=\mathcal{M}(\Sigma)$ with probability one.

The density function of $N(\mu,\Sigma)$ is given by
\begin{equation}
(2\pi)^{-r/2}\frac{1}{\sqrt{\det_+(\Sigma)}}\exp\left(-\frac{1}{2}(x-\mu)^T\Sigma^+(x-\mu)\right)\mathbb{I}\{x-\mu\in\mathcal{M}(\Sigma)\}. \label{eq:degen-Gau-density}
\end{equation}
The formula (\ref{eq:degen-Gau-density}) can be found in \cite{khatri1968some,songgui1994advanced}. Note that the density function (\ref{eq:degen-Gau-density}) is defined with respect to the Lebesgue measure on the subspace $\{x:x-\mu\in\mathcal{M}(\Sigma)\}$.
Here, the $\det_+(\cdot)$ is used for the product of all nonzero eigenvalues of a symmetric matrix, and $\Sigma^+$ is the Moore-Penrose inverse of the covariance matrix $\Sigma$. The two characterizations (\ref{eq:degen-Gau-latent}) and (\ref{eq:degen-Gau-density}) of $N(\mu,\Sigma)$ are equivalent to each other.

The property is useful for us to identify whether a formula leads to a well-defined density function of a degenerate Gaussian distribution.
\begin{lemma}\label{lem:density-to-distribution}
	Suppose $f(x)=\exp(-\frac{1}{2}(x-\mu)^T\Omega (x-\mu))\mathbb{I}\{x-\mu\in\mathcal{M}(V)\}$ for some $\mu\in\mathbb{R}^p$, some positive semidefinite $\Omega\in\mathbb{R}^{p\times p}$ and some $V\in \mathcal{O}(p,r)$. As long as $\mathcal{M}(\Omega V)=\mathcal{M}(V)$, we have $\int f(x)dx<\infty$, and $f(x)/\int f(x)dx$ is the density function of $N(\mu,\Sigma)$ with $\Sigma=V(V^T\Omega V)^{-1}V^T$.
\end{lemma}
\begin{proof}
	Without loss of generality, assume $\mu=0$. Since $\mathcal{M}(\Omega V)=\mathcal{M}(V)$, $V^T\Omega V$ is an invertible matrix, and thus $\Sigma$ is well defined. It is easy to see that $\mathcal{M}(V)=\mathcal{M}(\Sigma)$. Therefore, in view of (\ref{eq:degen-Gau-density}), we only need to show
	\ba
	x^T\Omega x = x^T(V(V^T\Omega V)^{-1}V^T)^+x,
	\ea
	for all $x\in\mathcal{M}(V)$. Since $x=VV^Tx$ for all $x\in\mathcal{M}(V)$, it suffices to show
	\ba
	V^T\Omega V = V^T (V(V^T\Omega V)^{-1}V^T)^+V,
	\ea
	which is immediately implied by (\ref{eq:useful}) with $R=V^T\Omega V$. The proof is complete.
\end{proof}

We remark that Lemma \ref{lem:density-to-distribution} also holds for a $V\in\mathbb{R}^{p\times r}$ that satisfies $\rank(V)=r$ but is not necessarily orthonormal. This is because $V(V^T\Omega V)^{-1}V^T=W(W^T\Omega W)^{-1}W^T$ whenever $\mathcal{M}(V)=\mathcal{M}(W)$.

\section{Proofs of Propositions}

\begin{proof}{\bf of Proposition \ref{prop:Laplacian}} \vspace{0.05in}\\
	The property of the Laplacian matrix $L_{\gamma}$ is standard in spectral graph theory \citep{spielman2007spectral}.
	We apply Lemma \ref{lem:density-to-distribution} with $\mu=0$, $\Omega=\sigma^{-2}L_{\gamma}$ and $V$ is chosen arbitrarily from $\mathcal{O}(p,p-1)$ such that $w^TV=0$. Then, the condition $\mathcal{M}(\Omega V)=\mathcal{M}(V)$ is equivalent to $\mathds{1}_p^Tw\neq 0$, because $\mathds{1}_p$ spans the null space of $L_{\gamma}$. Note that (\ref{eq:useful}) immediately implies $VRV^T=VV^T(VR^{-1}V^T)^+VV^T=(VR^{-1}V^T)^+$. We then have
	\begin{eqnarray*}
		\text{det}_+(V(V^T\Omega V)^{-1}V^T) &=& \sigma^{2(p-1)}\text{det}_+(V(V^TL_{\gamma} V)^{-1}V^T) \\
		&=& \sigma^{2(p-1)}\text{det}_+((VV^TL_{\gamma} VV^T)^+) \\
		&=& \sigma^{2(p-1)}\frac{1}{\text{det}_+(VV^TL_{\gamma} VV^T)}.
	\end{eqnarray*}
	The proof is complete by realizing that $VV^T=I_p-ww^T/\|w\|^2$ from Lemma~\ref{lem:matrix_inv}.
\end{proof}

\begin{proof}{\bf of Proposition \ref{prop:reduced-model}} \vspace{0.05in}\\
	Recall the incidence matrix $D\in\mathbb{R}^{m\times p}$ defined in Section \ref{sec:model-des}. With the new notations $\Phi=D^T\diag(\gamma)D$ and $\Psi=D^T\diag(1-\gamma)D$, we can write $L_{\gamma}=v_0^{-1}\Phi+v_1^{-1}\Psi$ and $\wt{L}_{\gamma}=v_1^{-1}\Psi$.
	
	We first prove that (\ref{eq:Laplacian-prior-low-d}) is well defined. By Lemma \ref{lem:density-to-distribution}, it is sufficient to show $\wt{V}^TZ_{\gamma}^T\Psi Z_{\gamma}\wt{V}$ is invertible. This is because $\mathcal{M}([\Phi;\Psi;w])=\mathbb{R}^p$, and the columns of $Z_{\gamma}\wt{V}$ are all orthogonal to $\mathcal{M}([\Phi;w])$. Therefore, (\ref{eq:Laplacian-prior-low-d}) is well defined and its covariance matrix is given by (\ref{eq:Laplacian-prior-low-d-moment}).

	Now we will show the distribution (\ref{eq:prior-Laplacian-form}) converges to that of $Z_{\gamma}\wt{\theta}$ with $\wt{\theta}$ distributed by (\ref{eq:Laplacian-prior-low-d}) as $v_0\rightarrow 0$. Let $\wt{V}\in\mathbb{R}^{r\times r-1}$ be a matrix of rank $r-1$ that satisfies $\wt{V}^TZ_{\gamma}^Tw=0$. Then, by Lemma \ref{lem:density-to-distribution}, the distribution (\ref{eq:Laplacian-prior-low-d}) can be written as
	\begin{equation}
	\wt{\theta}\sim N(0,\sigma^2\wt{V}(\wt{V}^TZ_{\gamma}^T\wt{L}_{\gamma}Z_{\gamma}\wt{V})^{-1}\wt{V}^T),\label{eq:Laplacian-prior-low-d-moment}
	\end{equation}
	which implies
	$$Z_{\gamma}\wt{\theta}\sim N(0,\sigma^2Z_{\gamma}\wt{V}(\wt{V}^TZ_{\gamma}^T\wt{L}_{\gamma}Z_{\gamma}\wt{V})^{-1}\wt{V}^TZ_{\gamma}^T).$$
	On the other hand, the distribution (\ref{eq:prior-Laplacian-form}) can be written as
	$$\theta\sim N(0,\sigma^2V(V^TL_{\gamma}V)^{-1}V^T),$$
	where the matrix $V\in\mathbb{R}^{p\times p-1}$ can be chosen to be any matrix of rank $p-1$ that satisfies $V^Tw=0$. In particular, we can choose $V$ that takes the form of
	$$V=[Z_{\gamma}\wt{V}; U],$$
	where $U\in\mathcal{O}(p,p-r)$ and $U$ satisfies $U^Tw=0$ and $U^TZ_{\gamma}\wt{V}=0$. Since both $\theta$ and $Z_{\gamma}\wt{\theta}$ are Gaussian random vectors, we need to prove
	\begin{equation}
	V(V^TL_{\gamma}V)^{-1}V^T \rightarrow Z_{\gamma}\wt{V}(\wt{V}^TZ_{\gamma}^T\wt{L}_{\gamma}Z_{\gamma}\wt{V})^{-1}\wt{V}^TZ_{\gamma}^T,\label{eq:cov-converge}
	\end{equation}
	as $v_0\rightarrow 0$.
	Note that (\ref{eq:cov-converge}) is equivalent to
	\begin{equation}
	V(v_0^{-1}V^T\Phi V+ v_1^{-1}V^T\Psi V)^{-1}V^T \rightarrow v_1Z_{\gamma}\wt{V}(\wt{V}^TZ_{\gamma}^T\Psi Z_{\gamma}\wt{V})^{-1}\wt{V}^TZ_{\gamma}^T, \label{eq:limit-to-prove}
	\end{equation}
	as $v_0\rightarrow 0$. By the definition of $Z_{\gamma}$ and the property of graph Laplacian, the null space of $\Phi$ is $\mathcal{M}(Z_{\gamma})$. By the construction of $V$, the null space of $V^T\Phi V$ is $\mathcal{M}(V^TZ_{\gamma})$. Moreover, we have $\mathcal{M}([V^T\Phi V; V^T\Psi V])=\mathbb{R}^{p-1}$. Therefore, Lemma \ref{lem:Y-Kim} implies that
	$$(v_0^{-1}V^T\Phi V+ v_1^{-1}V^T\Psi V)^{-1}\rightarrow v_1  V^TZ_{\gamma}(Z_{\gamma}^TVV^T\Psi VV^TZ_{\gamma})^{-1}Z_{\gamma}^TV,$$
	and thus
	$$V(v_0^{-1}V^T\Phi V+ v_1^{-1}V^T\Psi V)^{-1}V^T\rightarrow v_1  VV^TZ_{\gamma}(Z_{\gamma}^TVV^T\Psi VV^TZ_{\gamma})^{-1}Z_{\gamma}^TVV^T.$$
	Since $VV^TZ_{\gamma}=Z_{\gamma}\wt{V}\wt{V}^TZ_{\gamma}^TZ_{\gamma}$, we have $\mathcal{M}(VV^TZ_{\gamma})=\mathcal{M}(Z_{\gamma}\wt{V})$. This implies
	$$VV^TZ_{\gamma}(Z_{\gamma}^TVV^T\Psi VV^TZ_{\gamma})^{-1}Z_{\gamma}^TVV^T=Z_{\gamma}\wt{V}(\wt{V}^TZ_{\gamma}^T\Psi Z_{\gamma}\wt{V})^{-1}\wt{V}^TZ_{\gamma}^T,$$
	and therefore we obtain (\ref{eq:limit-to-prove}). The proof is complete.
\end{proof}

\begin{proof}{\bf of Proposition \ref{prop:projection-clustering}} \vspace{0.05in}\\
	We only prove the case with $d=1$. The general case with $d\geq 2$ follows the same argument with more complicated notation of covariance matrices (such as Kronecker products). By Lemma \ref{lem:density-to-distribution}, (\ref{eq:spike-and-slab-clustering}) is the density function of $N(0,\Sigma)$, with
	$$\Sigma=\sigma^2U(U^TL_{\gamma}U)^{-1}U^T,$$
	where
	$$L_{\gamma}=\begin{bmatrix}
	(v_0^{-1}-v_1^{-1})I_n + v_1^{-1}kI_n & -(v_0^{-1}-v_1^{-1})\gamma - v_1^{-1}\mathds{1}_{n\times k} \\
	-(v_0^{-1}-v_1^{-1})\gamma^T - v_1^{-1}\mathds{1}_{k\times n} & (v_0^{-1}-v_1^{-1})\gamma^T\gamma + v_1^{-1}nI_k
	\end{bmatrix}.$$
	The matrix $U$ is defined by
	$$U=\begin{bmatrix}
	V & 0_{n\times k} \\
	0_{k\times (n-1)} & I_k
	\end{bmatrix},$$
	and $V\in\mathbb{R}^{n\times (n-1)}$ is a matrix of rank $n-1$ that satisfies $\mathds{1}_n^TV=0$. Note that $\theta|\gamma,\sigma^2$ follows $N(0,\Sigma_{[n]\times [n]})$. That is, the covariance matrix is the top $n\times n$ submatrix of $\Sigma$. A direct calculation gives
	$$\Sigma_{[n]\times [n]}=\sigma^2 V[V^T(A-BC^{-1}B^T)V]^{-1}V^T,$$
	where
	\begin{eqnarray*}
		A &=& (v_0^{-1}-v_1^{-1})I_n + v_1^{-1}kI_n, \\
		B &=& -(v_0^{-1}-v_1^{-1})\gamma - v_1^{-1}\mathds{1}_{n\times k}, \\
		C &=& (v_0^{-1}-v_1^{-1})\gamma^T\gamma + v_1^{-1}nI_k.
	\end{eqnarray*}
	Letting $v_1\rightarrow\infty$, we have
	\ba
	\Sigma_{[n]\times [n]}\rightarrow \sigma^2v_0 V[V^T(I_n-\gamma(\gamma^T\gamma)^{-1}\gamma^T)V]^{-1}V^T.
	\ea
	The existence of $(\gamma^T\gamma)^{T}$ is guaranteed by the condition that $\gamma$ is non-degenerate. By Lemma \ref{lem:density-to-distribution}, $p(\theta|\gamma,\sigma^2)\propto \prod_{1\leq i<l\leq n}\exp\left(-\frac{\lambda_{il}\|\theta_i-\theta_l\|^2}{2\sigma^2v_0}\right)\mathbb{I}\{\mathds{1}_n^T\theta=0\}$ is the density function of $N(0,\sigma^2v_0 V[V^T(I_n-\gamma(\gamma^T\gamma)^{-1}\gamma^T)V]^{-1}V^T)$, which completes the proof.
\end{proof}

\begin{proof}{\bf of Proposition \ref{prop:gmm}} \vspace{0.05in}\\
	We only prove the case with $d=1$. The general case with $d\geq 2$ follows the same argument with more complicated notation of covariance matrices (such as Kronecker products). The proof is basically an application of Proposition \ref{prop:reduced-model}. That is, as $v_0\rightarrow 0$, the distribution of $(\theta^T,\mu^T)^T$ weakly converges to that of $Z_{\gamma}\wt{\mu}$. In the current setting, we have
	$$Z_{\gamma}=\begin{bmatrix}
	\gamma \\
	I_k
	\end{bmatrix}.$$
	The random vector $\wt{\mu}$ is distributed by (\ref{eq:Laplacian-prior-low-d-explicit}). Note that the contracted base graph is a complete graph on $\{1,...,k\}$, and $w_{jl}=n_j+n_l$ in the current setting. The density (\ref{eq:Laplacian-prior-low-d-explicit}) thus becomes
	\ba
	p(\wt{\mu}|\gamma,\sigma^2)\propto \prod_{1\leq j<l\leq k}\exp\left(-\frac{(n_j+n_l)(\wt{\mu}_j-\wt{\mu}_l)^2}{2\sigma^2 v_1}\right)\mathbb{I}\{\mathds{1}_n^T\gamma\wt{\mu}=0\}.
	\ea
	Finally, the relations $\theta=\gamma\wt{\mu}$ and $\mu=\wt{\mu}$ lead to the desired conclusion.
\end{proof}

\begin{proof}{\bf of Proposition \ref{prop:half-gaussian}} \vspace{0.05in}\\
	Note that the integration is with respect to the Lebesgue measure on the $(n-1)$-dimensional subspace $\{\theta:\mathds{1}_n^T\theta=0\}$. Consider a matrix $V\in\mathbb{R}^{n\times n-1}$ of rank $n-1$ that satisfies $\mathds{1}_n^TV=0$, which means that the columns of $[\mathds{1}_n: V]\in\mathbb{R}^{n\times n}$ form a nondegenerate basis. Then, we can write $d\theta$ in the integral as $\frac{1}{\sqrt{\det(V^TV)}}d(V^T\theta)$. For the Laplacian matrix $L_{\gamma}$ that satisfies $\theta^TL_{\gamma}\theta=\sum_{i=1}^{n-1}\frac{(\theta_{i+1}-\theta_i)^2}{v_0\gamma_i+v_1(1-\gamma_i)}$, we have $\mathds{1}_n^TL_{\gamma}=0$. In particular, we choose $V$ such that its $i$th column is $V_{*i}=e_{i}-e_{i+1}$, where $e_i$ is a vector whose $i$th entry is $1$ and $0$ elsewhere.
	Then, we have $L_{\gamma}=VS_{\gamma}V^T$ with $S_{\gamma}=\diag(v_0^{-1}\gamma + v_1^{-1}(1-\gamma))$, and
	the integral becomes
	\begin{align}
	\nonumber {}& \int_{\mathds{1}_n^T\theta=0,\theta_1\leq...\leq\theta_n}2^{(n-1)}\frac{1}{(2\pi\sigma^2)^{(n-1)/2}}\sqrt{\text{det}_{\mathds{1}_n}(L_{\gamma})}\exp\left(-\frac{1}{2\sigma^2}\theta^TL_{\gamma}\theta\right)d\theta \\
	\nonumber ={}& \int_{\theta_1\leq...\leq\theta_n}2^{(n-1)}\frac{1}{(2\pi\sigma^2)^{(n-1)/2}}\sqrt{\frac{\text{det}_{\mathds{1}_n}(L_{\gamma})}{\text{det}(V^TV)}}\exp\left(-\frac{1}{2\sigma^2}\theta^TVS_{\gamma}V^T\theta\right)d(V^T\theta) \\
	\nonumber ={}& \int_{\delta_1\leq 0,...,\delta_n\leq 0}2^{(n-1)}\frac{1}{(2\pi\sigma^2)^{(n-1)/2}}\sqrt{\frac{\text{det}_{\mathds{1}_n}(L_{\gamma})}{\text{det}(V^TV)}}\exp\left(-\frac{1}{2\sigma^2}\delta^TS_{\gamma}\delta\right)d\delta \\
	\label{eq:half-many} ={}& \int\frac{1}{(2\pi\sigma^2)^{(n-1)/2}}\sqrt{\frac{\text{det}_{\mathds{1}_n}(L_{\gamma})}{\text{det}(V^TV)}}\exp\left(-\frac{1}{2\sigma^2}\delta^TS_{\gamma}\delta\right)d\delta \\
	\label{eq:arsenal} ={}& \sqrt{\frac{\text{det}_{\mathds{1}_n}(L_{\gamma})}{\det(S_{\gamma})\det(V^TV)}} \\
	\nonumber ={}& 1,
	\end{align}
	where the last equality is by $\text{det}_{\mathds{1}_n}(L_{\gamma})=\det_+(VS_{\gamma}V^T)=\det(S_{\gamma})\det(V^TV)$. The equality (\ref{eq:half-many}) is by the symmetry of $\exp\left(-\frac{1}{2\sigma^2}\delta^TS_{\gamma}\delta\right)$, and (\ref{eq:arsenal}) is by Lemma \ref{lem:density-to-distribution}. Finally, Lemma \ref{lem:matrix-tree} says that 
	\ba
	\text{det}_{\mathds{1}_n}(L_{\gamma}) = n\prod_{i=1}^{n-1} [v_0^{-1}\gamma_i + v_1^{-1}(1-\gamma_i)].
	\ea
	This completes the proof.
\end{proof}

\begin{proof}{\bf of Proposition \ref{prop:half-gaussian-limit}} \vspace{0.05in}\\
	We need to calculate
	$$\int _{\mathds{1}_n^T Z_\gamma \wt{\theta} = 0,\ \wt\theta_1\leq \cdots \leq \wt\theta_s}\prod_{l=1}^s \exp\left(-\frac{(\wt{\theta}_l-\wt{\theta}_{l+1})^2}{2\sigma^2v_1}\right) d\widetilde\theta,$$
	where the integral is taken with respect to the Lebesgue measure on the low-dimensional subspace $\{\wt{\theta}: \mathds{1}_n^T Z_\gamma \wt{\theta} = 0 \}$. Choose $\wt{V}_{*l} = e_l - e_{l+1}$, where $e_l \in \{0,1\}^s$ is a vector whose $l$th entry is $1$ and $0$ elsewhere. Then the columns of $[Z_\gamma^T \mathds{1}_n: (Z_\gamma^T Z_\gamma)^{-1}\wt{V}]\in\mathbb{R}^{s\times s}$ form a non-degenerate basis of $\mathbb{R}^s$. This is because
	$$ Z_\gamma^T \mathds{1}_n = (n_1,\cdots,n_s)^T,\quad Z_\gamma^T Z_\gamma =  {\rm diag}(n_1,\cdots,n_s) ,$$
	where $n_l$ is the size of $l$th cluster. Furthermore, $Z_\gamma^T \mathds{1}_n$ and $(Z_\gamma^T Z_\gamma)^{-1}\wt{V}$ are orthogonal to each other. We write $\wt{W} = (Z_\gamma^T Z_\gamma)^{-1}\wt{V}$ for simplicity. Then,
	\begin{eqnarray*}
		&&\int _{\mathds{1}_n^T Z_\gamma \wt{\theta} = 0,\ \wt\theta_1\leq \cdots \leq \wt\theta_s}\prod_{l=1}^s \exp\left(-\frac{(\wt{\theta}_l-\wt{\theta}_{l+1})^2}{2\sigma^2v_1}\right) d\widetilde\theta \\
		&=& \frac{1}{\sqrt{\det \wt{W}^T\wt{W}  }} \int_{\wt\theta_1\leq \cdots \leq \wt\theta_s} \prod_{l=1}^s \exp\left(-\frac{(\wt{\theta}_l-\wt{\theta}_{l+1})^2}{2\sigma^2v_1}\right) d(\wt{W} ^T \wt\theta) \\
		&=&  \frac{\det \wt{W}^T \wt{W} (\wt{V}^T \wt{W})^{-1} }{\sqrt{\det \wt{W}^T\wt{W}  }}   \int_{\wt\theta_1\leq \cdots \leq \wt\theta_s} \prod_{l=1}^{s-1} \exp\left(-\frac{(\wt{\theta}_l-\wt{\theta}_{l+1})^2}{2\sigma^2v_1}\right)d(\wt{V}^T \wt\theta) \\
		&=&\frac{\det \wt{W}^T \wt{W} (\wt{V}^T \wt{W})^{-1} }{\sqrt{\det \wt{W}^T\wt{W}  }}  \int_{\wt\delta_1,\cdots,\wt\delta_{s-1}\leq 0} \prod_{l=1}^{s-1} \exp\left(-\frac{{\wt\delta}_l^2}{2\sigma^2v_1}\right)d\wt\delta \\
		&=& \frac{\sqrt{\det \wt{W}^T \wt{W} }}{\det \wt{V}^T \wt{W}} \int_{\wt\delta_1,\cdots,\wt\delta_{s-1}\leq 0} \prod_{l=1}^{s-1} \exp\left(-\frac{{\wt\delta}_l^2}{2\sigma^2v_1}\right)d\wt\delta \\
		&=& \frac{(2\pi\sigma^2v_1)^{(s-1)/2}}{2^{s-1}}  \times  \frac{\sqrt{\det \wt{W}^T \wt{W} }}{\det \wt{V}^T \wt{W}}.
	\end{eqnarray*}
	The first equality is from the Lebesgue integration on the reduced space and the last equality is by symmetry. The second equality follows from the change of variables formula, because for any $\wt{\theta}$ such that $\mathds{1}_n^T Z_\gamma \wt\theta = 0$, $\wt\theta = \wt{W}U$ for some $U$, which leads to $\wt{W}^T \wt{W} (\wt{V}^T \wt{W})^{-1} \wt{V}^T \wt{\theta} = \wt{W}^T \wt{W} U = \wt{W}^T \wt\theta$.
	Finally, we observe that
	\ba
	v_1^{(s-1)/2}\sqrt{\text{det}_{Z_{\gamma}^T \mathds{1}_n}(Z_{\gamma}^T\wt{L}_{\gamma}Z_{\gamma})} {}& = \left( \frac{\det( \wt{W}^T \wt{V} \wt{V}^T \wt{W}) }{\det \wt{W}^T\wt{W}} \right)^{1/2} \\
	{}&  =\left( \frac{ ( \det \wt{V}^T \wt{W})^2 }{\det \wt{W}^T \wt{W}} \right)^{1/2} = \frac{\det \wt{V}^T \wt{W}}{\sqrt{\det \wt{W}^T \wt{W} }},
	\ea
	since $\wt{V}\wt{V}^T = v_1\wt{L}_\gamma$ is the reduced graph Laplacian and the columns of $\wt{W}$ spans the nullspace of $Z_{\gamma}^T \mathds{1}_n$.
	The proof is complete.
\end{proof}

\section{Proof of Lemma \ref{lem:matrix-tree}}

We let $U\in\mathcal{O}(p,p-1)$ be an orthonormal matrix that satisfies  $\mathds{1}_p^TU=0$, and $V\in\mathcal{O}(p,p-1)$ be an orthonormal matrix that satisfies $w^TV=0$. Write $I_p-U(V^TU)^{-1}V^T$ as $R$. Then, $RU=0$, which implies that $R$ has rank at most one. The facts $Rw=w$ and $\mathds{1}_p^TR=\mathds{1}_p^T$, together with Lemma \ref{lem:matrix_inv}, imply that
\begin{equation}
I_p-U(V^TU)^{-1}V^T = \frac{1}{\mathds{1}_p^Tw}w\mathds{1}_p^T. \label{eq:wow-interesting}
\end{equation}
Therefore,
\begin{eqnarray} \label{eq:3.1main}
\nonumber \text{det}_w(L_{\gamma}) &=& \text{det}_+(VV^TL_{\gamma}VV^T) \\
\nonumber &=& \text{det}(V^TL_{\gamma}V) \\
\label{eq:tree1} &=& \text{det}(V^TUU^TL_{\gamma}UU^TV) \\
\nonumber &=& (\det(V^TU))^2\det(U^TL_{\gamma}U).
\end{eqnarray}
The inequality (\ref{eq:tree1}) is because $UU^T$ is a projection matrix to the null space of $L_{\gamma}$. We are going to calculate $\det(V^TU)$ and $\det(U^TL_{\gamma}U)$ separately. For $\det(V^TU)$, we have
\begin{eqnarray}
\nonumber 1 &=& \det\left(\begin{bmatrix}
V & \|w\|^{-1}w
\end{bmatrix}^T\begin{bmatrix}
U & p^{-1/2}\mathds{1}_p
\end{bmatrix}\right) \\
\nonumber &=& \det\left(\begin{bmatrix}
V^TU & p^{-1/2}V^T\mathds{1}_p \\
\|w\|^{-1}w^TU & w^T\mathds{1}_p/(p^{1/2}\|w\|)
\end{bmatrix}\right) \\
\nonumber &=& \det(V^TU)\det(w^T(I_p-U(V^TU)^{-1}V^T)\mathds{1}_p)/(p^{1/2}\|w\|) \\
\nonumber &=& \det(V^TU)\frac{p^{1/2}\|w\|}{\mathds{1}_p^Tw},
\end{eqnarray}
and we thus get
\begin{eqnarray} \label{eq:tree2}
\left(\det(V^TU)\right)^2=\frac{\left(\mathds{1}_p^Tw\right)^2}{p\|w\|^2}.
\end{eqnarray}
We use (\ref{eq:wow-interesting}) for the equality (\ref{eq:tree2}).

The calculation of $\det(U^TL_{\gamma}U)$ requires the Cauchy-Binet formula. For any given matrices $A,B\in\mathbb{R}^{n\times m}$ with $n\leq m$, we have
\begin{equation}
\det(AB^T)=\sum_{\{S\subset[m]:|S|=n\}}\det(B_{*S}^TA_{*S}).\label{eq:CB}
\end{equation}
The version (\ref{eq:CB}) can be found in \cite{tao2012topics} and references therein. Let $A=B=U^TD^T(v_0^{-1/2}\diag(\gamma)+v_1^{-1/2}\diag(1-\gamma))$, and we have
$$
\det(U^TL_{\gamma}U) = \sum_{\{S\subset E:|S|=p-1\}}\left(\prod_{(i,j)\in S}\left[v_0^{-1}\gamma_{ij}+v_1^{-1}(1-\gamma_{ij})\right]\right)\det(D_{S*}UU^TD_{S*}^T).
$$
Note that $\det(D_{S*}UU^TD_{S*}^T)=\det(U^TD_{S*}^TD_{S*}U)$, and $D_{S*}^TD_{S*}$ is the graph Laplacian of a subgraph of the base graph with the edge set $S$. Since $|S|=p-1$, $S$ is either a spanning tree ($\det(U^TD_{S*}^TD_{S*}U)=p$) or is disconnected ($\det(U^TD_{S*}^TD_{S*}U)=0$), we have
\begin{equation} 
\det(U^TL_{\gamma}U)=p\sum_{T\in\text{spt}(G)}\prod_{(i,j)\in T}\left[v_0^{-1}\gamma_{ij}+v_1^{-1}(1-\gamma_{ij})\right].\label{eq:originalmatrixtree}
\end{equation}
Therefore, by plugging \eqref{eq:tree2} and \eqref{eq:originalmatrixtree} into \eqref{eq:3.1main}, we obtain the desired conclusion.

\section{Some Implementation Details}\label{app:imp}

In this section, we present a fast algorithm that solves the M-step \eqref{eq:cartesian-M}.
To simplify the notation in the discussion, we consider a special case with $X_1=I_{n_1}$, $X_2=I_{n_2}$, $w=\mathds{1}_{n_1}\mathds{1}_{n_2}^T$ and $\nu=0$, which is the most important setting that we need for the biclustering problem. In other words, we need to optimize $F(\theta;q_1,q_2)$ over $\theta\in\Theta_w$ for any $q_1$ and $q_2$, where
\begin{eqnarray} \label{minimize}
F(\theta;q_1,q_2) &=& \fnorm{y-\bar{y}\mathds{1}_{n_1}\mathds{1}_{n_2}^T-\theta}^2+{\sf vec}(\theta)^T\left(L_{q_2}\otimes I_{p_1} + I_{p_2}\otimes L_{q_1}\right){\sf vec}(\theta) \\
&=& \fnorm{y-\bar{y}\mathds{1}_{n_1}\mathds{1}_{n_2}^T-\theta}^2 + \iprod{\theta\theta^T}{L_{q_1}} + \iprod{\theta^T\theta}{L_{q_2}}.
\end{eqnarray}
Algorithm \ref{alg:dlpa} is a Dykstra-like proximal algorithm (DLPA) \citep{dykstra1983algorithm} that iteratively solves the optimization problem.
\begin{algorithm}[t!]
	{\bf Input:} Initialize $u_1,u_2$ and $z_2$.
	\begin{algorithmic}
		\REPEAT
		\STATE $z_1^{new} = (I_{n_1} + L_{q_1})^{-1}(z_2 + u_2)$
		\vspace{0.05in}
		\STATE $z_2^{new} = (z_1^{new} + u_1)(I_{n_2} + L_{q_2})^{-1}$
		\vspace{0.05in}
		\STATE $u_1^{new} = z_2 + u_2 - z_1^{new}$, $\ u_2^{new} = z_1^{new} + u_2 - z_2^{new}$
		\vspace{0.05in}
		\STATE $z_1 = z_1^{new}$, $z_2 = z_2^{new}$, $u_1 = u_1^{new}$, $u_2 = u_2^{new}$
		\vspace{0.05in}
		\UNTIL{convergence criteria met}
	\end{algorithmic}
	{\bf Output:} {$\theta = z_2$}
	\caption{A fast DLPA}
	\label{alg:dlpa}
\end{algorithm}
It is shown that Algorithm~\ref{alg:dlpa} has a provable linear convergence \cite{combettes2011proximal}. If we initialize $u_1 = u_2 = 0_{n_1\times n_2}$ and $z_2 = y - \bar{y}\mathds{1}_{n_1}\mathds{1}_{n_2}^T$ with $\hat\alpha = \bar{y}$, then the first two steps of Algorithm~\ref{alg:dlpa} can be written as the following update
\begin{equation}
\theta^{\rm new} = (I_{n_1} + L_{q_1})^{-1}\left(y - \bar{y}\mathds{1}_{n_1}\mathds{1}_{n_2}^T\right)(I_{n_2} + L_{q_2})^{-1}. \label{first-two}
\end{equation}
In practice, we suggest using (\ref{first-two}) as approximate M-step updates in the first few iterations of the EM algorithm. Then, the full version of Algorithm \ref{alg:dlpa} can be implemented in later iterations to ensure convergence.

\begin{figure}[!t]
	\centerline{\includegraphics[width=2.4in]{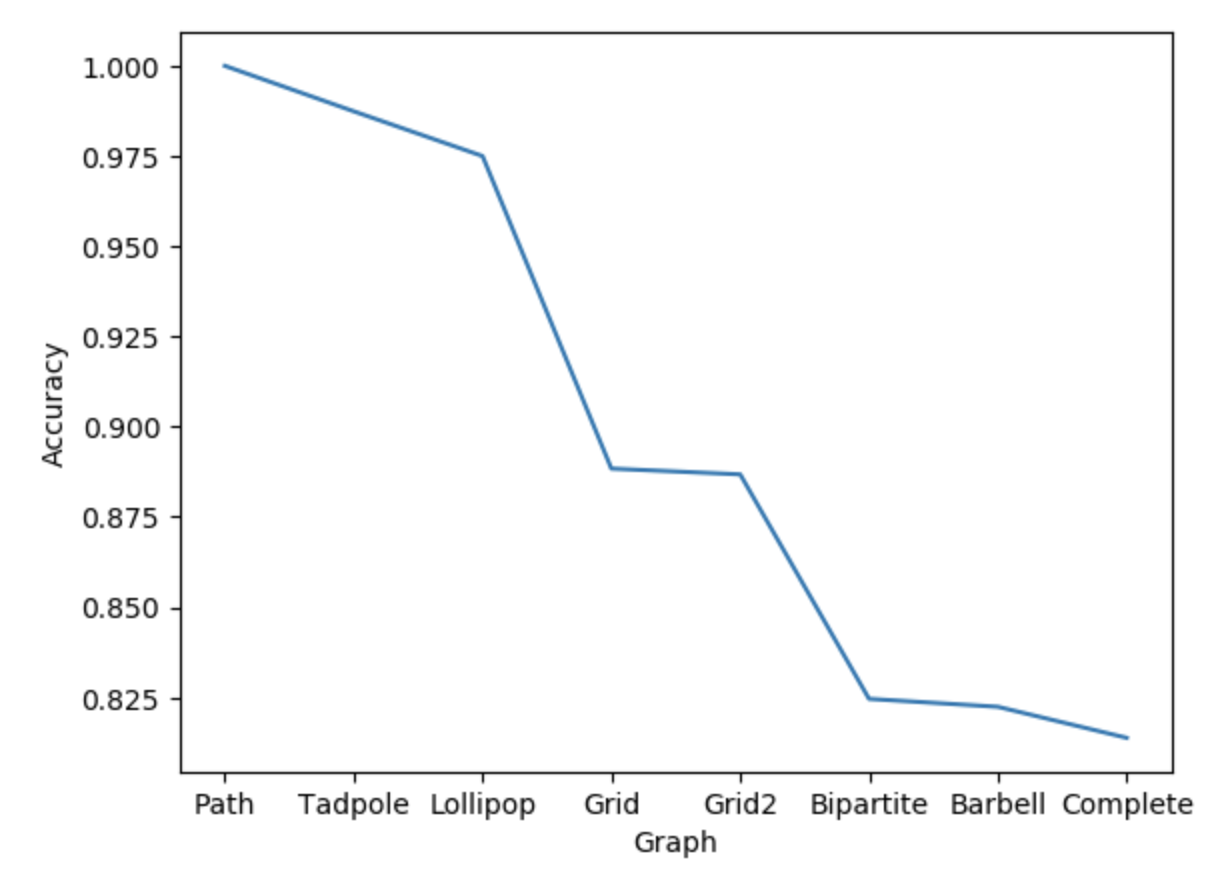}\includegraphics[width=3.7in]{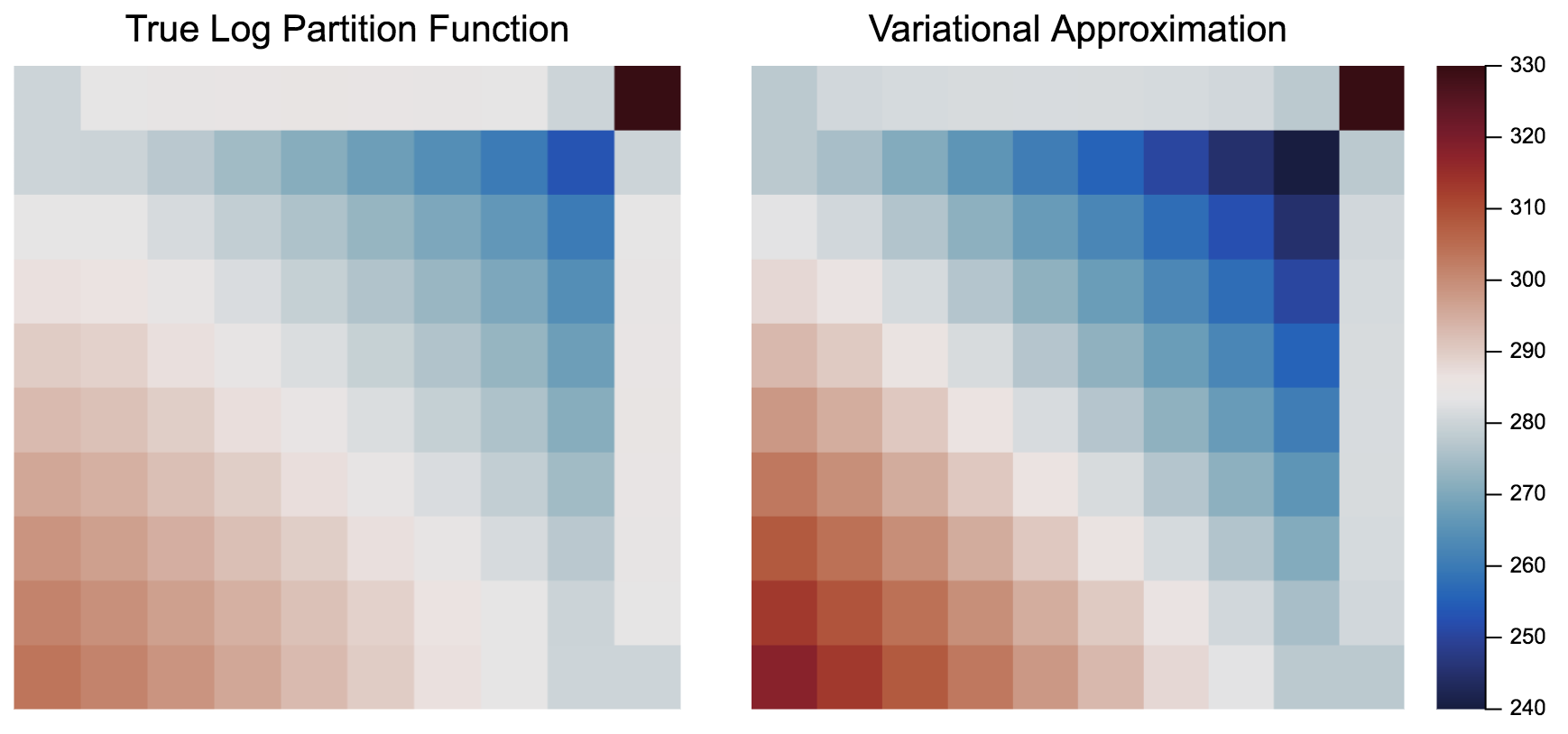}}
	\caption{(Left) Accuracy of variational approximation to the normalization constant; (Center and Right) The quality of variational approximation on Grid(10,10).}
	\label{fig:graph1}
\end{figure}
\section{Accuracy of Variational Approximation}
\label{sec:acc-var-approx}

In this section, we conduct an empirical study of the accuracy of the variational approximation.
The variational lower bound we used is
\ba
	{}& \log\sum_{T\in\text{spt}(G)}\prod_{(i,j)\in T}\left[v_0^{-1}\gamma_{ij}+v_1^{-1}(1-\gamma_{ij})\right] \\
	\geq {}& \sum_{(i,j)\in E}r_{ij}\log\left[v_0^{-1}\gamma_{ij}+v_1^{-1}(1-\gamma_{ij})\right] + \log|\text{spt}(G)|,
\ea
where $r_{ij} = |\text{spt}(G)|^{-1}\sum_{T\in\text{spt}(G)}\mathbb{I}\{(i,j)\in T\}$
is an effective resistance of an edge $(i,j)$. An effective resistance $r_e$ of an edge $e$ measures how important the edge $e$ is in the whole graph, provided that their local resistance is $1$. For instance, if $(i,j)$ is the only edge connecting two mutually exclusive subgraphs containing $i$ and $j$ respectively then $r_{ij} = 1$.  In this case, determining whether $\gamma_{ij}=1$ has a direct effect on separation of $i$ and $j$, one may want to put larger weights on $\gamma_{ij}$ in the objective function.
In the right hand side of \eqref{eq:lower}, $\log[v_0^{-1}\gamma_{ij} + v_1^{-1}(1-\gamma_{ij})]$ is weighted by the effective resistance $r_{ij}$ of the edge $(i,j)$. This implies that the variational lower bound or the right hand side of the above inequality puts larger weights to more ``important'' edges. This intuitive explanation can be verified by the following numerical experiments.

We compare the true log partition function
\[
f_1(\gamma) = \log\sum_{T\in\text{spt}(G)}\prod_{(i,j)\in T}\left[v_0^{-1}\gamma_{ij}+v_1^{-1}(1-\gamma_{ij})\right]
\]
with our variational lower bound
\[
f_2(\gamma) = \sum_{(i,j)\in E}r_{ij}\log\left[v_0^{-1}\gamma_{ij}+v_1^{-1}(1-\gamma_{ij})\right] + \log|\text{spt}(G)|
\]
for various choices of graphs. To this end, we randomly sample two weighted graphs $G_1$ and $G_2$ and then put independent random uniform weights on the edges of $G_1$ and $G_2$, respectively. Then we compare the differences between the true normalization constants and the differences between the variational lower bounds. We fix $n = 100$, $v_0 = 10^{-1}$ and $v_1 = 10^{3}$. Let us write the linear chain graph as $P_{n}$ and the complete graph as $K_n$. Let $\vee$ be a graph operator joining two graphs by adding  exactly one connecting edge between two graphs. We compare $P_n$, $K_n$ and the 6 graphs in Table~\ref{table:8}. The graph $K_{80,20}$ is sampled under the additional constraint that $\sum_j \gamma_{ij} = 1$.
\begin{table}[!t]
	\small
	\centering
	\begin{tabular}{l | l | l | l | l | l  l  l  l}
		Tadpole(50,50) & Lollipop(80,20) & Grid(10,10) & Grid(20,5) & Bipartite(80,20) & Barbell(50,50)\\
		\hline
		$P_{50} \vee C_{50}$ & $P_{80} \vee K_{20}$ & $P_{10} \otimes P_{10}$ & $P_{20} \otimes P_{5}$ & $K_{80,20}$ & $K_{50} \vee K_{50}$
	\end{tabular}
	\normalsize
	\caption{A List of Graphs Used for Evaluating Quality of Variational Approximation} \label{table:8}
\end{table}

The left panel of Figure~\ref{fig:graph1} plots the accuracy of approximation, measured by $e^{f_2(\gamma)}/e^{f_1(\gamma)}$ against the complexity of graphs. It shows that the approximation is more accurate for a sparser graph, but even for the complete graph we still obtain an accuracy over $0.8$. 

The center and the right panels of Figure~\ref{fig:graph1} displays the true log-partition function and the variational approximation when Grid(10,10) is clustered into two parts. To be specific, the true signal $\theta^* \in \mathbb{R}^{10 \times 10}$ is defined on the nodes of $G$ and has exactly 2 clusters, the bottom left square of $(a,b)$ and the rest. That is,
\ba
\mathcal{C}_1^{(a,b)} = \{\theta_{ij}: i \leq a\textrm{ and }j \leq b\} \quad \mathcal{C}_2^{(a,b)} = \{\theta_{ij}: i > a \textrm{ or } j > b\}.
\ea
For instance, if $a = 3$ and $b = 4$, then $\mathcal{C}_1$ has 12 nodes. Let $\gamma^{(a,b)}$ be the corresponding latent structure parameter, and $\gamma^{(a,b)}_{ij}=1$ if $i$ and $j$ are in the same cluster and $\gamma^{(a,b)}_{ij}=0$ otherwise. For $a = 1,\cdots,10$ and $b = 1,\cdots,10$, we compare the true log-normalization constant $f_1(\gamma^{(a,b)})$ and its variational approximation $f_2(\gamma^{(a,b)})$, and summarize the result in the center and the right panels of Figure~\ref{fig:graph1}. The result is displayed by heat-maps, and the values of $f_1(\gamma^{a,b})$ and $f_2(\gamma^{a,b})$ are reported in the $(a,b)$-th entries of the two heatmaps. The results support our conclusion that the variational approximations are reasonable for various graphical structures.



\bibliography{Bayes}

\end{document}